\documentclass[aoas]{imsart}

\RequirePackage{amsthm,amsmath,amsfonts,amssymb}
\RequirePackage{natbib}
\RequirePackage[colorlinks,citecolor=blue,urlcolor=blue]{hyperref}
\RequirePackage{graphicx}
\usepackage[hang,flushmargin]{footmisc}
\usepackage{./macros/macros_03_symbols}
\usepackage{./macros/macros_04_propernames}
\usepackage{booktabs, comment, dsfont, xcolor, enumitem, todonotes, threeparttable, algorithm, algpseudocode, float}
\usepackage[section]{placeins}
\usepackage[toc,page]{appendix}

\definecolor{editone}{HTML}{0000FF}

\newcommand{\editedinline}[1]{{\color{editone} #1}}

\startlocaldefs
\theoremstyle{plain}
\theoremstyle{remark}

\newtheorem{lemma}{Lemma}
\newtheorem{proposition}{Proposition}

\newcommand{\norm}[1]{\left\lVert#1\right\rVert}
\newcommand{\var}[1]{\text{Var}\left(#1\right)}

\newcommand{\iid}{i.i.d.$\;$}
\endlocaldefs

\endlocaldefs

\begin{document}

\begin{frontmatter}
\title{Estimating the likelihood of arrest from police records\\in presence of unreported crimes}

\runtitle{Estimating arrests from police records with unreported crimes}

\begin{aug}
\author[A]{\fnms{Riccardo} \snm{Fogliato}}, 
\author[B]{\fnms{Arun Kumar} \snm{Kuchibhotla}},
\author[B]{\fnms{Zachary} \snm{Lipton}},
\author[B]{\fnms{Daniel} \snm{Nagin}},
\author[C]{\fnms{Alice} \snm{Xiang}},
\and
\author[B]{\fnms{Alexandra} \snm{Chouldechova}}
\address[A]{Amazon Web Services\thanks{Riccardo Fogliato worked on this project during his time at Carnegie Mellon University.}}
\address[B]{Carnegie Mellon University}
\address[C]{Sony AI}
\end{aug}

\begin{abstract}
  Many important policy decisions concerning policing
hinge on our understanding of how likely 
various criminal offenses are to result in arrests.
Since many crimes
are never reported to law enforcement, estimates 
based on police records alone must be adjusted
to account for the likelihood that each crime
would have been reported to the police.
In this paper, we present a methodological framework
for estimating the likelihood of arrest from police data 
that incorporates estimates of crime reporting rates 
computed from a victimization survey. 
We propose a parametric regression-based two-step estimator
that (i) estimates the likelihood of crime reporting using logistic regression with survey weights;
and then (ii) applies a second regression step %
to model the likelihood of arrest.
Our empirical analysis focuses on racial disparities in arrests 
for violent crimes (sex offenses, robbery, aggravated and simple assaults) 
from 2006--2015 police records 
from the National Incident Based Reporting System (NIBRS), 
with estimates of crime reporting obtained using 2003--2020 data from 
the National Crime Victimization Survey (NCVS). 
We find that, after adjusting for unreported crimes, 
the likelihood of arrest computed from police records decreases significantly. 
We also find that, while
incidents with white offenders
on average result in arrests more often than those with black offenders, 
the disparities tend to be small after accounting for crime characteristics and unreported crimes. 
\end{abstract}


\end{frontmatter}

\section{Introduction}

Characterizing the likelihood that a criminal offense will result in an arrest
is central to multiple lines of criminological research, 
including crime control, deterrence, and racial disparities 
\citep{nagin2013deterrence, piquero2008assessing}. 
Analyses of arrests traditionally rely on data 
collected by law enforcement agencies. 
The offenses captured by these records, however,
represent only a fraction of all crimes that occur. %
By neglecting the ``dark figure of crime'' \citep{skogan1974validity}, 
these analyses inevitably overestimate 
the underlying arrest rate per crime committed. 
The overestimation can potentially be severe, 
as data of criminal victimization reveal 
that less than half of violent offenses in the US 
ever become known to law enforcement \citep{morgan2021criminal}. 

In order to estimate the likelihood of arrest
for all crimes that are committed,
police records can be augmented with data 
on crime reporting from victimization surveys. 
In the US, the National Crime Victimization Survey (NCVS) collects 
information on whether respondents experienced a victimization, and
whether police were made aware of the offense. 
The idea of combining victimization data with police records 
was first proposed by \citet{blumstein1979estimation}, 
who estimated arrest rates for violent offenses
in Washington D.C. in the 1970s. 
However, owing to the limited data available,
their approach could not account 
for variations in crime reporting rates
across offense characteristics. 

\begin{figure}[t]
    \includegraphics[width=\linewidth]{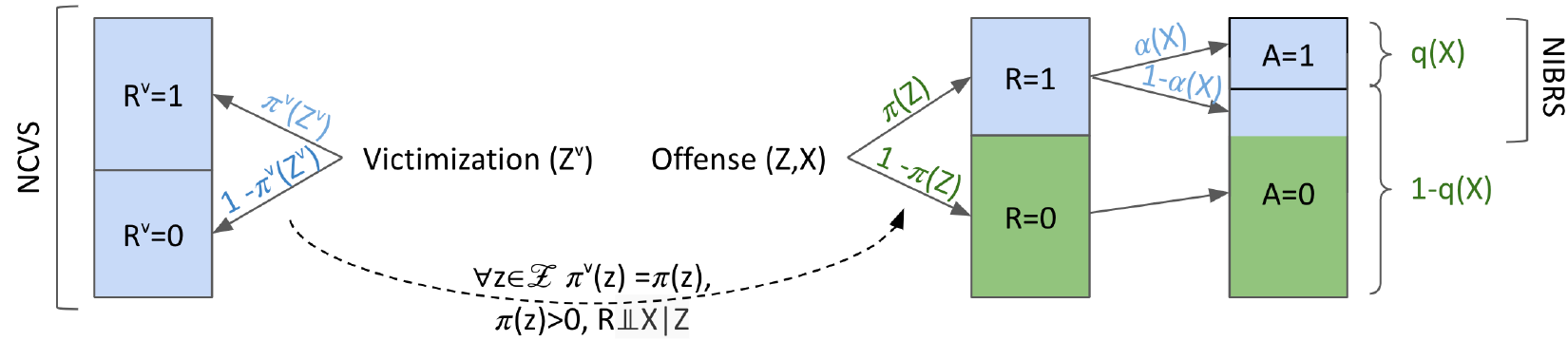}
    \caption{Summary of the proposed methodology. On the left, a victimization
    with characteristics $Z^v$ is reported to law enforcement ($R^v=1$) with
    probability $\pi^v(Z^v)$. NCVS data capture both reported and unreported
    victimizations. On the right, an offense with characteristics $(Z,X)$ is
    reported to the police ($R=1$) with some positive probability $\pi(Z)$. $X$ captures all of
    the information that is contained in $Z$, i.e., the distribution of $Z|X$ is degenerate. Only
    reported offenses ($R=1$) appear in NIBRS data. A reported crime can result in
    an arrest ($A=1$) with probability $\alpha(X)$. The target of interest in
    our work is the conditional probability $q(X)$ 
    that an offense will result in an arrest. 
    Although the estimation of $\pi$ (and consequently of $q$)
    cannot be pursued solely on police records, assuming that $\pi^v(z) =
    \pi(z)$ for all $z$'s on the support $\mathcal{Z}$ of $Z$ allows for the
    estimation of the likelihood of police notification on these data. 
    The rigorous description
    of this setup is detailed in Section \ref{sec:methods}.
    }\label{fig:methodology-plot}
\end{figure}

The increasing granularity and availability of incident-level crime data 
released by police agencies through 
the National Incident Based Reporting System (NIBRS) 
offers an opportunity to improve this analysis
by accounting for crime characteristics. 
Combined with the information collected through the NCVS, 
NIBRS records can yield more accurate estimates 
of arrest rates per crime committed.
Unfortunately, it is not possible to link records in NCVS and NIBRS directly. 
In this paper, we propose statistical methods 
to estimate the likelihood of arrest on incident-level data from police records 
while accounting for each crime's 
likelihood of police notification,
computed from victimization survey data. 
The proposed methodology consists of two simple steps 
(Figure \ref{fig:methodology-plot}).
We first estimate the likelihood of police
notification conditional on offense characteristics 
via logistic regression with survey weights on victimization data. 
Then, we derive estimators of the total number of offenses, 
the rate of police notification, and the rate of arrest
that leverage the characteristics of crimes in police records 
and the likelihood that the crime would be reported.
The likelihood of arrest conditional on crime
characteristics is modeled using logistic regression, 
and the coefficient estimates are obtained 
through a two-step estimation approach 
which accounts for the rates of crime reporting.
When fitting logistic regression on crimes with multiple offenders, 
we handle the data dependence using 
generalized estimating equations (GEEs) \citep{liang1986longitudinal}.
We show that the proposed estimators are consistent 
and asymptotically normal for the target parameters, under a series of assumptions.
Although our analytical results rely on the assumption that the models are correctly specified, the model results can be interpreted if this assumption does not hold in practice \citep{buja2019models1, buja2019models, berk2019assumption}. Although we focus on logistic regression, the proposed framework can be used to show asymptotic normality of any parametric regression model using an analysis similar to the one we conduct.

Our empirical investigation focuses on the assessment of differences in
the likelihood of arrest across racial groups on 2006--2015 NIBRS data, 
with estimates of crime reporting obtained from NCVS data. 
We focus on violent crimes (sex offenses, robbery, aggravated and simple assaults)
because the race of the offender is observed by the victim in the majority of such incidents.  By contrast, many property crimes occur without the victim present, and so there is often no opportunity for them to directly observe offender characteristics.
Our analysis reveals that on average about one in two violent offenses 
becomes known to law enforcement and 
one in five eventually results in arrest. 
Since the likelihood of crime reporting and of arrest 
vary with crime characteristics, 
arrestees do not form a representative sample of all offenders. 

In terms of racial disparities, 
we find that crimes involving black offenders
are reported at (marginally) higher rates 
than those involving white offenders, 
yet they result in arrests less often. 
Once crime characteristics are accounted for,
the estimated differences in arrests across racial groups %
tend to be small.
We further validate our results through an additional analysis 
where we employ nonparametric models instead of logistic regression 
to estimate the likelihood of police notification on survey data. 
Throughout the discussion one should keep in mind that our empirical findings rely on multiple data-related assumptions, which may not hold true (see limitations in Section \ref{sec:limitations}).

The rest of the paper is organized as follows:
Section \ref{sec:related_work} contains related work. 
Section \ref{sec:data_processing} describes 
the data sources and related data processing.
Section \ref{sec:methods} contains the methodological framework.
Section \ref{sec:emp_analysis} includes the empirical analysis 
and in Section \ref{sec:results} we present the results of our study. 
Limitations of our analysis and future work are discussed in Sections
\ref{sec:limitations} and \ref{sec:discussion} respectively.

\section{Related work}\label{sec:related_work}

The methodology developed in our work is related to the literature on
missing data problems (being a case of missingness not at random) \citep{kang2007demystifying, little2019statistical} and on capture-recapture \citep{petersen1896yearly, lee1994estimating}. 
In both lines of work, 
the target of key interest is the expected value of a random variable that is only partially observed. 
Works in this area mainly employ inverse probability weighting methods, e.g., using the Horvitz-Thompson and H\'{a}jek
estimators \citep{horvitz1952generalization, basu2011essay}. 
In this work, we use the latter estimator to compute the share of unreported crimes and arrest rates from police records. The inclusion probabilities correspond to the estimated likelihood of
police notification obtained via logistic regression. 
This approach can be thought of as a special case of the capture-recapture setting studied in \citet{huggins1989statistical} when there is only one occasion to recapture.
Similar results in the capture-recapture setting are also obtained by \citet{van2003point} and \citet{bohning2009covariate}. 
However, these works impose distributional assumptions to handle data dependence. 
Instead, following prior empirical analyses of NIBRS \citep{d2003race, fogliato2021validity}, we
assume independence of the observations. 
A popular procedure used to fit models in presence of sampling bias is the two-step approach proposed by \citet{heckman1979sample}.  This approach most commonly applies probit regression in the first stage and OLS in the second stage.  Our approach instead uses logistic regression in both steps and relies on a different set of assumptions.
The design and derivation of our
estimation procedure also draws from the literature on survey sampling
\citep{sarndal2003model} and two-step M-estimation \citep{newey1994large}. The
regression analysis of arrests for crimes involving multiple offenders (hence
with associated outcomes) via GEEs is inspired by the methods developed in the
epidemiological literature \citep{hubbard2010gee}. 
Lastly, our analysis operates
under the assumption of covariate shift, i.e., that the distribution of the
regressors but not of the outcome may vary between train and test sets
\citep{sugiyama2007covariate}. In our setting, regressors and outcomes are
represented by victimization and offense characteristics and by
whether the crime has been reported respectively, while train and test set correspond to NCVS
and NIBRS data respectively. Differently from these works, however, we assume
that the posited regression models for the likelihood of crime reporting are
well-specified and consequently no adjustments of the loss, such as by reweighting
\citep{byrd2019effect}, are required.

Our work contributes to the literature on crime control. 
Estimates of the dark figure of crime and arrest rates for violent offenses
have traditionally been obtained either from cross-sectional
data and from self-reports of offending behavior, or solely from victimization data. 
Unlike our study, these analyses generally focus on arrest rates per
individual rather than per crime committed. The approach taken by the studies in
the first line of work was pioneered by \citet{blumstein1979estimation}, and is
in spirit similar to ours. They compute arrest rates as the ratio of the
arrest rates measured on police records and of the crime reporting rates on
victimization surveys \citep{blumstein1986criminal,
blumstein1987characterizing}. 
These studies suffer from one major drawback: As we mentioned in the
Introduction, by using aggregate data from police agencies, they cannot account
for variations in the likelihood of police notification across crime types, as
\citet[page 335]{blumstein1986criminal} also noted. This is relevant to our
analysis especially because NIBRS and NCVS may capture populations of offenses
with different characteristics, as not all police agencies have adopted NIBRS
yet. Our proposed methodology addresses this issue. Within the second
line of work, \citet{blumstein2010linking} estimate arrest rates for violent
crimes on the Rand Second Inmate Survey, a survey of inmates in three US jails
conducted in the 1970s. Their estimates of arrest rates are close to ours. 
In the same study, the authors also assess arrest rates across racial groups and
find no evidence of disparities. Two analyses focused on data from the Pathways
to Desistance study, a longitudinal investigation of serious juvenile offenders
from adolescence to young adulthood \citep{piquero2008assessing,
brame2004criminal}, similarly do not find racial disparities in arrests. The
seeming contrast with our results may be explained by the different nature of
the populations of offenders we focus on, or by temporal differences.
Within the third and last line of work, \citet{buil2021measuring} estimate how
the dark figure of crime varies with crime and neighborhood characteristics 
on victimization survey data in the UK.
They conclude that this figure is associated with the socioeconomic status
of the parties involved. Although our analysis does not account for these
specific characteristics, our results similarly reveal that the likelihood of
crime reporting varies with the demographics of the victim and of the offender. 
Multiple studies have examined racial disparities in police notification and
arrests for violent offenses known to law enforcement. There is evidence that, overall, incidents with black offenders are at least as
likely as those with white offenders to be reported to law enforcement
\citep{morgan2017race, beck2018racial}. 
After accounting for contextual factors, incidents are generally more likely to be
reported when one of the parties involved is black \citep{avakame1999did,
xie2012racial, baumer2010reporting, bachman1998factors,fisher2003reporting},
although there exist both conflicting and null findings
\citep{baumer2002neighborhood, dugan2003domestic}. 
In our analysis, we find that incidents with black offenders are reported at
slightly higher rates than those with white offenders, even conditional on crime
characteristics. 

There is also mixed evidence
concerning the magnitude of differences in arrests across racial groups for crimes known to law
enforcement. While some works have concluded that crimes are more likely to result in arrest when the offender is black 
\citep{kochel2011effect,
lytle2014effects}, multiple analyses focused on violent offenses on NIBRS data have
reached a different conclusion \citep{d2003race, pope2003race,
roberts2009victim}. These works have found that, even after accounting for crime
characteristics, black offenders are less likely to be arrested than white
offenders for assault and robbery.  Differences for rape and
homicide were found to be negligible.   In our analysis, we find that accounting for unreported crimes reduces the estimated
gap in arrest rates for robbery, and the estimated gaps for assaults are close to zero.
While most studies have focused on the analysis of incidents with single
offenders and victims, \citet{lantz2019co} analyze incidents involving violent offenses where white and
black individuals offend together. They fit one single regression model and they conclude that white offenders are less
likely to be arrested than black offenders. We instead focus on all crimes with multiple
offenders, fit separate models for each crime type, and find that the likelihood of arrest is mostly similar across
racial groups of offenders. The only exception is robbery, for which arrest 
appears to be more likely for crimes involving white offenders, regardless of
whether unreported crimes are accounted for.  
Lastly, we note that the results in the literature are likely susceptible to issues stemming from model misspecification. For instance, \citet{fogliato2021validity} showed that model misspecification can impact
the magnitude and even
the direction of the estimated racial disparities. 
In this work, we reach analogous conclusions: Our model estimates vary depending on the subsets of crimes considered in the analysis.

\section{Data}\label{sec:data_processing}

Our empirical analysis leverages data from the National Crime Victimization
Survey (NCVS) and the National Incident Based Reporting System (NIBRS).
Similarly to past studies on NIBRS \citep{d2003race, fogliato2021validity}, the main
analysis in the paper centers on incidents with one victim and one offender.
The data processing to obtain the dataset of crimes involving multiple offenders
requires stronger assumptions. This is because NCVS only allows for inference at
the level of the incident, while NIBRS contains also offender-level data. We now
describe each of the two data sources and related data processing in turn. In
this step of the analysis, we wish to identify a set of incidents captured by
NCVS and NIBRS that share similar characteristics to ensure that the covariate
shift assumption underpinning our analysis plausibly holds.

\subsection{Data on criminal victimization}

The NCVS represents the primary source of information on victimization in the US
\citep{barnett2014nation}. By collecting information on the magnitude and extent
of criminal victimization from a nationally representative sample of households,
it is designed to complement data from police agencies with an alternative
measurement of crime. Survey respondents aged 12 or older are interviewed
regarding the criminal victimizations that they experience for
nonfatal personal crimes %
\citep{ncvstech16}. Our analysis focuses on data from interviews conducted
between 2003 and 2020, which we obtain from the repository of the
Inter-university Consortium for Political and Social Research (ICPSR)
\citep{ncvs9220}. The data contain information on the stratified, multi-stage cluster sampling design, namely (pseudo-)strata, primary
sampling units (PSUs), and observations (incidents) weights for serious crimes.\footnote{Similarly to past studies on NCVS \citep{xie2012racial, xie2019neighborhood}, our analysis assumes that nonresponse bias is accounted for by the use of survey weights. We acknowledge that in practice this assumption may not hold true.} Information about the sampling design and on 
the construction of the sampling weights can be found in \citet{ncvstech16}.

In the data, we consider only victimizations that satisfy the following
criteria. (i) Incidents need to include an offense of simple assault (excluding
verbal threats of assault), aggravated assault, robbery, or rape/sexual assault,
which we will refer to as ``sex offenses'' in the rest of the
paper.\footnote{The Bureau of Justice Statistics (BJS) conflates simple assault
with verbal threats of assault in their annual reports. In NIBRS, however, only
physical attacks are coded as simple assaults (c.f. page 18 in
\citet{nibrs_manual19}). In order to align the definitions of simple assaults in
NIBRS and NCVS, these crimes are excluded from the analysis. Consequently,
statistics based on our proposed taxonomy, which has been chosen to ensure the
maximal overlap of offense types between NCVS and NIBRS, will not match those in
the reports produced in the BJS reports.} (ii) We keep only incidents that have occurred
within the United States. (iii) Since the NCVS collects information (e.g.,
demographics) only about the respondent, we drop incidents involving more than
one victim. 
(iv) We consider only incidents with black or white individuals, with the
inclusion of Hispanics.\footnote{The ethnicity information for victim and
offender has been available in NCVS data since 2003 and 2012 respectively.
However, the exclusion of Hispanics from NIBRS data is rather challenging
because not all agencies report the offender ethnicity information, which was
introduced in 2013. One could potentially attempt to identify the law
enforcement agencies that generally report such information and
consider only data from those agencies, as \citet{roberts2011hispanic} have
done. That procedure, however, would introduce a geographical bias in our NIBRS
sample and thus we do not adopt such an approach.} In case of incidents
involving multiple offenders, we consider only those in which at least one of the
offenders belongs to these racial groups. Our final dataset of incidents
with single offenders consists of 11145 observations which, when reweighted
by the survey weights, correspond to about 40 million crimes. The most frequent
types of offense is simple assault (54\%, based on survey weights), followed by aggravated assault
(24\%). Robbery and sex offense are the least frequent types of crime and each
of them comprises about 10\% of the available observations. The dataset of
incidents with one or more offenders comprises 3405 additional observations and in total
it corresponds to about 50 million incidents.

The outcome of interest in NCVS is whether the police are aware of the incident,
as reported by the NCVS respondent in the survey ($R\in\{0,1\}$). 
We consider the likelihood of an incident being reported
$\pi^v(Z^v)$ to depend on a set of factors $Z^v$ which include characteristics
of the parties involved and contextual factors. In terms of demographics, we
account for the age, sex, and race of both victim and offenders. In the analysis
of multiple offenders, we consider the sex of the majority of the offenders, and
the age of the youngest and of the oldest offenders. We also consider the
relationship between victim and offenders (e.g., if they are relatives), whether
the victim suffers from a serious or minor injury, and whether the offenders
have a firearm or a different weapon. We include two variables corresponding to
whether the incident happens during the day and whether it occurs in a public
area. To account for geographical variations in the likelihood of police
notification, we account for whether the incident took place in a metropolitan
statistical area (MSA), the corresponding US Census region in which the
incident took place, and the year of the interview. Lastly, we consider whether
the offense has been only attempted, the type of crime, and, in case of sex
offenses, whether the offense consists of either rape or sexual assault. All
variables other than the victim's age and the year are categorical.

\subsection{Crime data from law enforcement agencies}

NIBRS is part of the Federal Bureau of Investigation's Uniform Crime Reporting
(UCR) data collection program. Through this program, law enforcement agencies 
submit detailed data on the characteristics of incidents that are known to them,
including information on victims and offenders, and on the nature of the
offenses. Our analysis builds on the assumption that when a crime becomes known
to law enforcement, it will be recorded in the data released by law enforcement.
Our analysis relies on 2006--2015 NIBRS data obtained from the ICPSR repository
\citep{nibrs06,nibrs07,nibrs08,nibrs09,nibrs10,nibrs11,nibrs12,nibrs13,nibrs14,nibrs15}.
Note that while we rely on NCVS data from the period 2003-2020, the NIBRS data spans a shorter time period. 

For this
analysis, we identify incidents with characteristics that are similar
to those included in our NCVS dataset. Thus, we apply the following data
restrictions. (i) We consider incidents involving crimes of rape and sexual
assault (i.e., sex offenses)\footnote{In the category of rape and sexual assault
we include crimes of forcible rape, forcible sodomy, sexual assault with an
object, and forcible fondling. We do not consider statutory rape and incest
because such offenses are unlikely to be reported by NCVS respondents in the
interviews. The definition of rape in the UCR was revised in 2013 to also
include male victims and female offenders.}, robbery, aggravated assault, and
simple assault. The majority of the incidents (about 99\% of cases)
contain only one of these offenses. In the analysis of incidents with multiple
offenders, we similarly found that in almost all of the incidents the offenders
were charged with the same offense. Thus
we can reasonably make the simplifying assumption that all
offenders involved in the same crime incident commit the same offense. 
(ii) We keep only data from the 16 states that reported most of their crime
data through the NIBRS in this time period. These states are Arkansas, Colorado, Delaware, Idaho,
Iowa, Kentucky, Michigan, Montana, New Hampshire, North Dakota, South Carolina,
South Dakota, Tennessee, Vermont, Virginia, and West Virginia. 
This exclusion makes our sample representative of a
population that is well defined, i.e., the crimes that have become known
to police and reported by agencies in the 16 states considered. 
(iii) We drop incidents that involve more than one victim and, for the analysis
of incidents with single offenders, we also drop those that involve more than
one offender. (iv) We account only for incidents where the races of victims and
offenders are either black or white, including Hispanics. We observe that, based on the data of
ethnicity that are available, about 90\% of the offenders of Hispanic origin present
in our sample are classified as whites. (v) We drop incidents that are cleared
by exceptional means due to the death of the offenders or because the offender
is in the custody of another jurisdiction.\footnote{The excluded incidents represent less than $1\%$ of all offenses, so their inclusion is unlikely to change the conclusions of our analysis. In addition, there are not large differences in clearance by exceptional means across racial
groups; see the results in Section A of the Appendix in
\citet{fogliato2021validity}.} We consider the remaining incidents that
are cleared by exceptional means, namely those for which a juvenile offender was
not taken into custody, prosecution was declined, or the victim refused to
cooperate, as having no arrest.
(vi) Lastly, to align our sample with the population of NCVS respondents, we
drop incidents with victims aged 11 or younger. Our final samples of offenses
involving only individual and one or more offenders consist of approximately 3.3
million and 4.9 million offenses respectively. As in the NCVS data, most of the
offenses are simple assault (about 70\%) and aggravated assault (about 17\%).

The outcome of interest in our analysis is whether the incident results in the
arrest of the offender (denoted $A\in\{0,1\}$). In order to estimate the likelihood that
a crime becomes known to the police for each incident in this dataset, $\pi(Z)$,
we process the features in the data to obtain a set of crimes characteristics
$Z$ that is analogous to those captured in our final NCVS dataset ($Z^v$). NIBRS
also contains additional information that can be used in estimating the likelihood of arrest. In our application, we estimate the likelihood of
the crime resulting in the offender's arrest, $q(X)$, based on crimes
characteristics, $X$. In our analysis of incidents with individual offenders, $X$
includes not only all the variables that are present in $Z$, but also
information about the state where the crime occurred, the size of the police
force in the agency, and the number of police officers per capita. This additional police agency data is obtained from police employee datasets downloaded from
 ICPSR \citep{leoka06, leoka07, leoka08, leoka09, leoka10, leoka11, leoka12,
leoka13, leoka14, leoka15}. In the analysis of incidents with multiple
offenders, $Z$ captures aggregate information about the incident, e.g., the age
of the youngest offender. $X$ contains variables measured both at the
level of the incident and of the individual offender, e.g., the age of the
individual offender.

\section{Methods}\label{sec:methods}

\begin{algorithm}[ht]
\caption{Estimation strategy on NCVS and NIBRS}\label{alg:est_strategy}
\begin{algorithmic}
\State $\hat\gamma \gets \text{solve} \sum_{i=1}^{N^v} w_i I_i h^v(R_i^v, Z_i^v;\gamma) =0$ \Comment{Likelihood of police notification $\pi^v(Z^v;\gamma)$ on NCVS data}
\State $\hat
N\gets \sum_{i=1}^N R_i/\pi(Z_i;\hat\gamma)$ \Comment{Total number of offenses $N$ on NIBRS data}
\State $\hat{\pi}^* \gets \sum_{i=1}^N R_i/\hat N$, $\hat{q}^*\gets \sum_{i=1}^N A_i/\hat N$ \Comment{Rate of police notification $\pi^*$ and arrest $q^*$ on NIBRS data}
\State $\hat\theta \gets \text{solve } \sum_{i=1}^N R_i h(A_i, Z_i, X_i;\theta,\hat\gamma) =  0$ \Comment{Likelihood of arrest $q(X;\theta)$ on NIBRS data}
\end{algorithmic}
\end{algorithm}

In this section, we present the statistical methodology behind our empirical
analysis. 
Algorithm \ref{alg:est_strategy} describes the proposed estimation approach. 
We first estimate $\pi^v(Z^v;\gamma)$, the likelihood of police
notification conditional on crime characteristics via survey-weighted logistic regression. 
We then introduce the offense data setup and
describe the assumptions underpinning our inference strategy. 
Next, we review
estimators of NIBRS summary statistics, namely the total number of offenses $N$, the rate of police notification $\pi^*$, and the arrest rate for all crimes committed $q^*$.   
Lastly, we estimate $q(X;\theta)$, the
likelihood of arrest conditional on crime characteristics via logistic
regression on NIBRS data. 
Under regularity conditions, all of the estimators that we present are asymptotically normal. 
Detailed derivations and proofs of the results are deferred to the Appendix. 

\subsection{Crime reporting on NCVS}\label{sec:est_ncvs}
Consider the finite population of criminal victimizations in the US,
denoted by $V^{N^v}=\{(Z_i^v, R_i^v)\}_{i=1}^{N^v}$. This can be viewed as an
\iid sample $(Z^v, R^v)\sim P^v$, where
$Z^v=(Z(1),\dots,Z(d_z))\in\mathcal{Z}$ indicates the victimization's
characteristics and $R^v\in\{0,1\}$ is the indicator of whether the victimization becomes known to law
enforcement.
We have
access to the NCVS survey sample of size $n^v$, which is drawn from $V^{N^v}$
under some probability sampling design $\psi$. Let the random variable $I_i=1$
if $i^{th}$ observation is included in this sample, and $I_i=0$ otherwise. The
sample has an associated set of sampling weights $\{w_i:1\leq i\leq N^v, I_i = 1\}$,
which are commonly thought as representing the number of units that each sampled
observation represents in the larger finite population \citep{lohr2007}.

We model the conditional probability of police notification, $\pi^v(Z^v)$, 
via logistic regression. That is, we take
$\pi^v(z;\gamma):=1/(1+e^{-\gamma^T z})$ to describe $\mbbP_{P^v}(R^v=1|Z=z)$, for
$\gamma\in \Gamma$ for some compact set $\Gamma\subset \mathbb{R}^{d_z}$. 
The superpopulation target parameter $\gamma_0\in \text{Int}(\Gamma)$ is defined by the moment condition $\mbbE_{P^v}[h^v(R^v, Z^v;\gamma)] = 0$ where $h^v(R^v, Z^v;\gamma) := (R^v-\pi^v(Z^v;\gamma))Z^v$. 
The design-based estimator $\hat\gamma$ of $\gamma_0$ is the solution to the estimating equation $\sum_{i=1}^{N^v} w_i I_i h^v(R_i^v, Z_i^v;\gamma) =0$ \citep{lumley2017fitting}. 
Under certain regularity conditions, $(\Sigma^v)^{-1/2}\sqrt{n^v}(\hat\gamma -
    \gamma_0)\overset{d}{\rightarrow}\mathcal{N}(0,I_{d_z})$ where $\Sigma^v$ is a positive definite matrix.

\subsection{NIBRS setup}
Let $O^N=\{(X_i, Z_i, R_i, A_i)\}_{i=1}^N$ denote the sample of all offenses
committed, which is assumed to be an \iid sample of $(X,Z,R,A)\sim P$. In this part of the
analysis, we consider crimes with only one offender, so the independence
assumption is likely to hold (but see the longer discussion in Section
\ref{sec:limitations}). Let
$X=(X(1),\dots,X(d_x))\in\mathcal{X}$ and $Z=(Z(1),\dots,Z(d_z))\in\mathcal{Z}$
indicate incident characteristics. Let $R\in\{0,1\}$ and $A\in\{0,1\}$ indicate
whether the offense is known to the police ($R=1$) and whether it results in an
arrest ($A=1$) respectively. Given that an offense can result in an arrest only
if it is known to the police, we assume that $R=0$ implies $A=0$. Note that
police-recorded data contain only offenses that have been reported, i.e., those
for which $R=1$. We denote with $\mbbE$ the expectation over $P$. 

In order to estimate parameters of interest on the entire population using
solely the observations for which $R=1$, we will make use of the following set
of assumptions. 
\begin{enumerate}[label=\textbf{A.\arabic*},ref=A.\arabic*]
    \item \label{ass_correctpi} $\forall z \in \mathcal{Z}$,
    $\pi^v(z;\gamma_0)=\mbbP_{P^v}(R^v=1|Z^v=z)$ for $\gamma_0\in\Gamma$ where $\Gamma$
    is a compact set.
    \item \label{ass_pioverz} $\forall (x,z)\in\mathcal{X}\times \mathcal{Z}$, $\mbbP(R=1|X=x,Z=z)=\mbbP(R=1|Z=z)$.%
    \item \label{ass_matchingpi}  $\forall z \in \mathcal{Z}$, $\mbbP(R=1|Z=z)=\mbbP_{P^v}(R^v=1|Z^v=z)$.
    \item \label{ass_boundedcov} $\norm{X}_\infty< M$ and $\norm{Z}_\infty< M$
    for some $M>0$.
\end{enumerate}
\ref{ass_correctpi} states that the parametric model $\pi^v(z;\gamma_0)$ is correctly specified for $\mbbP_{P^v}(R^v=1|Z^v=z)$. We empirically assess
this assumption by comparing the logistic regression model with a nonparametric
approach, and find small differences in the estimates produced by the two
methods for three of the four offense types considered. \ref{ass_pioverz} states that $R$ is independent of $X$ after conditioning on $Z$. In our empirical analysis, we study the likelihood of arrest per crime
committed, $q$, as a function of only $X$ (see \ref{sec:methods_q}) because $X$
contains at least as much information as $Z$. In other words, $X$ includes more
refined details about the incident such as specific geographical information and
characteristics about each of the offenders within an incident (hence the
distribution of $Z|X$ is degenerate). These characteristics are
available in NIBRS but not in NCVS. However, through \ref{ass_pioverz} we assume
that this additional information is not relevant to the estimation of the
distribution of $R|Z$. \ref{ass_pioverz} may be violated if $Z$ did not
capture, for instance, variations in reporting rates across police agencies, but
$X$ did.
\ref{ass_matchingpi} allows us to compare the probability of police notification
in NIBRS and NCVS. This assumption casts our learning problem into the covariate
shift setting. Together, \ref{ass_correctpi} and \ref{ass_matchingpi} imply that
$\pi^v(z;\gamma_0)=\mbbP(R=1|Z=z)$.  Thus, in what follows we drop ``$v$''
from the superscript of $\pi(Z;\gamma_0)$. \ref{ass_boundedcov} ensures that
functions of these random variables will have finite moments. This assumption
clearly holds true in our application. Lastly, \ref{ass_correctpi} and
\ref{ass_boundedcov} imply that $\pi(z;\gamma)> (1 +
e^{\sqrt{d}M\sup_{\gamma\in\Gamma}\norm{\gamma}})^{-1}>0$ for all
$z\in\mathcal{Z}$ and $\gamma\in\Gamma$, an assumption that is known as \emph{positivity} in the
causal inference literature \citep{little2019statistical}. In our setting, it
rules out the possibility that there exist offenses with certain characteristics
that will never be reported to the police. Under conditions \ref{ass_correctpi}--\ref{ass_boundedcov}, we can establish
the consistency of the estimators we propose in the next section. By imposing additional assumptions on the rate of growth of $n^v$, $N^v$, and $N$, we can derive their asymptotic distributions as well.\footnote{Note that no assumption on $n:=\mbbE[R]$ is needed here because, unlike the
survey sampling setting, the \iid assumption on the data of offenses guarantees
that $n$ and $N$ will grow at the same rate. Since the sample sizes of the NCVS
and NIBRS samples are of comparable magnitude, we assume that $\lim_{n, n^v\rightarrow\infty} n/n^v=\kappa = O(1)$. The
results readily generalize to the cases where $n\ll n^v$ or $n \gg n^v$.}

\subsection{Crime reporting and arrest rates on
NIBRS}\label{sec:methods_summary_stats} 
There are three targets of key interest on NIBRS. First, the total number of offenses, which can be estimated using the Horvitz-Thompson estimator $\hat
N:=\sum_{i=1}^N R_i/\pi(Z_i;\hat\gamma)$. 
Second, the expected rate of
police notification $\pi^*:=\mbbE[R]$. 
Third, the arrest rate $q^*:=\mbbE[A]$. 
These rates are estimated by $\hat{\pi}^*:=\sum_{i=1}^N R_i/\hat N$ and $\hat{q}^*:=\sum_{i=1}^N A_i/\hat N$ respectively. 
Under the assumption that $\hat\gamma$ is consistent for
$\gamma_0$ and asymptotically normal, as well as some regularity conditions, these estimators are asymptotically normal. The
critical step in the derivation of the limiting distributions is to leverage the fact that $\mbbE[R/\pi(Z;\gamma_0)]=1$ in order to rewrite
the unconditional expectations with respect to the event $\{R=1\}$. This make
possible the estimation based only on the sample we have access to.

\subsection{Conditional probability of arrest on NIBRS via logistic regression}\label{sec:methods_q} We
model the probability of arrest conditional on the covariates $\mbbE[A|X]$ using
logistic regression; i.e., we consider $q(x;\theta) := 1/(1+e^{-\theta^T x})$ where
$\theta\in\Theta$ for a compact set $\Theta\subset \mathbb{R}^{d_x}$.
The parameter $\theta_0\in \text{Int}(\Theta)$ is
defined by the following moment condition
\begin{gather}\label{eq:q_reg}
    \mbbE\left[ (A-q(X;\theta))X \right] = 0.
\end{gather}
Since $A=0$ whenever $R=0$, it follows that $\mbbE[AX]=\mbbE[RAX]$. Then, under Assumptions \ref{ass_correctpi}--\ref{ass_matchingpi},
the moment condition \eqref{eq:q_reg} can be rewritten as
\begin{equation*}% \label{eq:q_reg_change}
    G(\theta,\gamma_0) := \mbbE\left[\left(A-\frac{q(X;\theta)}{\pi(Z;\gamma_0)}\right)X R\right] = 0. 
\end{equation*}
Thus, in practice, we
compute the estimator $\hat\theta$ of $\theta_0$ by solving the following
estimating equation
\begin{gather}\label{eq:theta_solve}
    \hat{G}_\mathcal{N}(\theta, \hat\gamma) := \frac{1}{N} \sum_{i=1}^N R_i h(A_i, Z_i, X_i;\theta,\hat\gamma) =  0,
\end{gather}
where $h(A_i,
Z_i, X_i;\theta,\gamma) := [A_i-q(X_i;\theta)/\pi(Z_i;\gamma)]X_i$. 
The estimate $\hat\theta$ of $\theta_0$ can be found using
iteratively (re-)weighted least squares. 
Under Assumptions \ref{ass_correctpi}--\ref{ass_boundedcov}, together with the consistency and asymptotic
normality of $\hat\gamma$ as an estimator of $\gamma_0$, then   
$\Sigma^{-1/2} \sqrt{n}(\hat\theta - \theta_0)\overset{d}{\rightarrow}\mathcal{N}(0,I_{d_x})$
as $n\rightarrow\infty$ where $\Sigma$ is positive definite. %

\subsection{Conditional probability of arrest on NIBRS via GEEs}\label{sec:methods_q_gee}

In the previous sections we considered only crimes with a single offender. We now consider estimating the conditional probability of arrest from police records on crimes where one \textit{or more} offenders are
involved. When multiple offenders act together, the \iid assumption on $O^N$ clearly doesn't hold. We denote the sample of crime incidents, where each incident represents a cluster containing observations corresponding to the individual offenders, occurring in the US with
$O^N=\{(K_i, \bold{X}_i, Z_i, R_i, \bold{A}_i)\}_{i=1}^N$, which is an \iid
sample of $(K, \bold{X}, Z, R, \bold{A})\sim P$ with $K\in\mathbb{Z}_{+}$. Here
$\bold{X}_i$ is a matrix of dimension $K_i\times d_x$ whose $k^{th}$ column
corresponds to $\bold{X}_{ik}$, the characteristics relative to the offense
committed by the $k^{th}$ offender in the $i^{th}$ incident. The vector
$\bold{A}_i=(\bold{A}_{i1}, \dots, \bold{A}_{iK_i})^T$ indicates whether each of
the offenders in the $i^{th}$ incident are arrested ($\bold{A}_{ij}=1$ for
$1\leq j\leq K_i$) or not. $R_i$ indicates whether the $i^{th}$ incident is
known to the police, and $Z_i$ represents characteristics of the same incident.
Note that $Z_i$ contains information that is shared across all offenders within
the same incident (e.g., location), while $\bold{X}_i$ may also include
covariates that are specific to the individual offender (e.g., demographics of
that offender). %

Despite the likely positive correlation across outcomes within the same incident,
the logistic regression coefficient estimates discussed in the
previous section remain consistent for the target parameters.  However, the
estimates of their asymptotic variance need to be adjusted
\citep{fitzmaurice1993regression}.  
To account for the correlation in the variance estimation and to increase
efficiency, we employ generalized estimating equations (GEEs)
\citep{liang1986longitudinal}. We
assume that 
\begin{enumerate}[label=\textbf{A.\arabic*},ref=A.\arabic*]
    \setcounter{enumi}{4}
    \item \label{ass:gee_mean_correct} $\mbbE[\bold{A}_{ij}|\bold{X}_{i}] =
q(\bold{X}_{ij};\theta_0)$ where $q(\bold{X}_{ij};\theta_0):=(1 + e^{-\theta_0^T
\bold{X}_{ij}})^{-1}$ for $1\leq i \leq N$, $1 \leq j\leq K_i$,
$\theta_0\in\text{Int}(\Theta)$.
\end{enumerate}
According to this assumption, the probability of arrest for an individual does not depend on incident's characteristics related to their co-offenders; see \citet[Section 3.2]{fitzmaurice2008longitudinal} for a longer discussion of this assumption. To model the covariance across outcomes, we define the matrix
$W_i(\theta, \alpha) := W(\bold{X}_i, Z_i;\theta, \alpha)=
D_i(\theta)^{1/2} C_i (\alpha) D_i(\theta)^{1/2}$ for $1\leq
i\leq N$, where $D_i(\theta)$ is a diagonal matrix of dimension
$K_i\times K_i$ whose $k^{th}$ diagonal entry corresponds to
$q(\bold{X}_{ik};\theta)(1-q(\bold{X}_{ik};\theta))$. $C_i(\alpha)$ is the
so-called  ``exchangeable working correlation'' matrix, which has dimension
$K_i\times K_i$ with $1$ on the diagonal and any $\alpha\in[-1,1]$ elsewhere
\citep{liang1986longitudinal}.

The estimator $\hat \theta\in\text{Int}(\Theta)$ solves the following
generalized estimating equation
\begin{equation*}% \label{eq:gees}
    \begin{split}
       \frac{1}{N} \sum_{i=1}^{N} R_i h_{GEE}(\bold{X}_i, Z_i, \bold{A}_i;\theta, \hat\alpha) :=  \frac{1}{N}\sum_{i=1}^{N} R_i \bold{X}_i D_i(\theta) W_i(\theta, \hat\alpha)^{-1} 
    \left(\bold{A}_i - \frac{\bold{q}_i(\theta)}{\pi(Z_i;\hat \gamma)} \right) = 0,
    \end{split}
\end{equation*}
where $\bold{q}_i(\theta) = (q(\bold{X}_{i1};\theta), \dots,
q(\bold{X}_{iK_i};\theta))^T$ and $\hat\alpha$ is a consistent estimator of
$\alpha_0$, the true correlation parameter, given $\theta$.
Then, under certain certain conditions, $\Sigma_{GEE}^{-1/2}\sqrt{n}(\hat\theta -
\theta_0)\overset{d}{\rightarrow}\mathcal{N}(0,I_{d_x})$ as $N\rightarrow\infty$ where 
$\Sigma_{GEE}$ is a positive definite matrix.

\begin{figure}
\begin{center}
\includegraphics[width=0.95\textwidth]{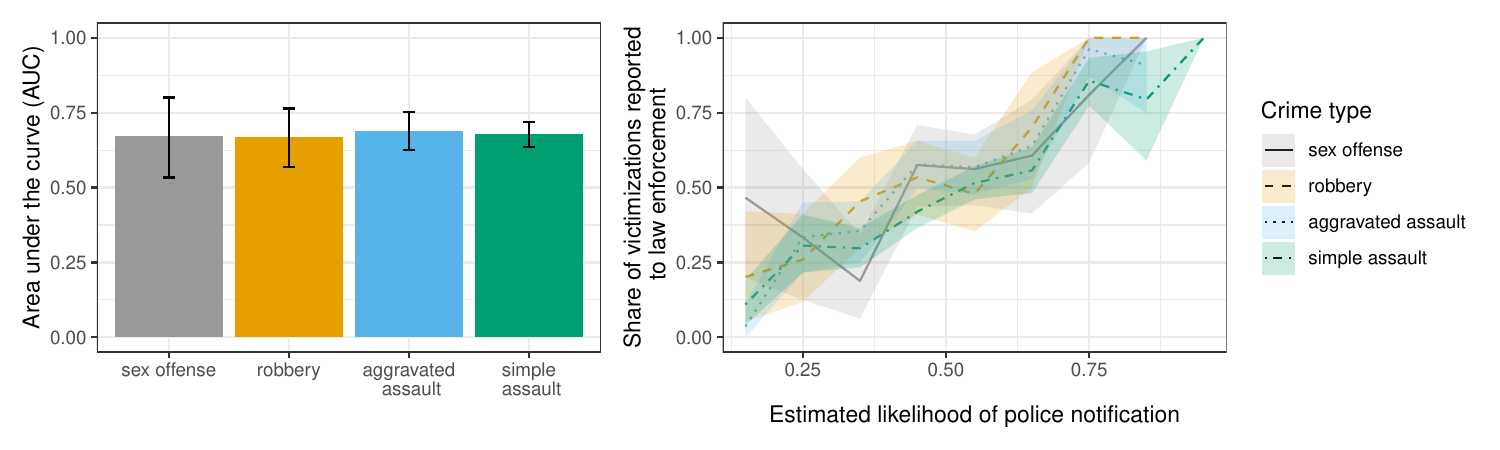}
\end{center}
\caption{Area under the curve (AUC) and calibration for the logistic regression
model with survey weights that estimates the
likelihood of police notification $\pi^v(Z^v)$ on 2003--2020 NCVS data, obtained via cross-validation for the
four types of crimes considered. Error bars and shaded regions indicate 95\% confidence intervals for the mean.}
\label{fig:logistic_regression_calandauc_ncvs}
\end{figure}

\section{Empirical strategy}\label{sec:emp_analysis}

Our empirical analysis leverages the analytical framework presented in Section
\ref{sec:methods}. The code used for data processing and the analysis is available at \href{https://github.com/ricfog/arrests-with-unreported-crimes}{github.com/ricfog/arrests-with-unreported-crimes}.

\subsection{Missing data} Both NCVS and NIBRS data contain a small share of
missing values in certain crime characteristics. In NCVS data, a handful of the
other variables are missing in less than 5\% of the cases. The only exceptions are the MSA information and offender's age in crimes of robbery, which are missing in about one fourth of the observations. In NIBRS data, a few
of the variables are missing for less than 5\% of the observations. Past work on
NCVS and NIBRS has assumed the data to be either missing completely at random
(MCAR) or missing at random (MAR) \citep{d2003race, xie2017effects,
fogliato2021validity}. We assume the data to be MAR and impute the missing
values via multiple imputation by chained equations (MICE) using predictive mean
matching for numerical data, and multinomial and logistic regression for categorical
data \citep{azur2011multiple}. We repeat the same procedure for both crimes
involving individual and multiple offenders.  It is common practice to create
multiple imputed datasets and then pool the estimates computed from each
\citep{graham2007many}. However, our preliminary analysis showed that multiple imputed datasets yielded similar estimates, likely because only a small proportion of observations have missing values. Due to the
considerable computational costs of our workflow, we decided to use only
one imputed dataset. Consequently, in the downstream inference we will not
account for the uncertainty arising from the imputation procedure, 
which we anticipate to be negligible.

\subsection{Estimation on NCVS} 
We now detail the key steps for the analysis of incidents with one offender. An
analogous procedure is carried out for incidents with multiple offenders with
similar results, so we omit the details for brevity. 

We begin by describing the process for using the 2003--2020 NCVS victimization data to estimate the likelihood that an offense becomes known
to law enforcement $\pi$ conditionally on its characteristics. We first split the data into two groups
stratifying by outcome and year. One subset, which comprises one fifth of the
data, is used for model selection. The validation subset, which consists of the rest
of the data, is used to estimate the chosen model that will be employed for
inference on NIBRS. In terms of model selection, we evaluate several regression
models that include interactions between the crime types and the other
regressors, which reflect different modeling choices made by past studies
\citep{xie2012racial, xie2019neighborhood, baumer2002neighborhood}.  
This step is
warranted by the diversity of feature sets and model choices adopted  
in regression analyses by these works. We assess model performance via model calibration and
area under the curve (AUC) separately for each
type of offense using cross-validation. Both evaluations account for survey
weights. The predictions appear to be well calibrated across the various models for
all crime types considered other than sex offenses. The AUCs produced by the predictions of the various models are similar, and all are in the range 0.55--0.7 across
offense types and models. We proceed with the logistic regression model without interactions between features and crime types, a model that was also considered in past work on NCVS similar to ours \citep{xie2012racial}.\footnote{We have conducted an additional analysis employing the same model with an interaction between offender's and victim's races. The results from this analysis and those that we report in the paper are similar.}

Next, we fit the selected model on the validation set and
compute the variance of the coefficient estimates using the information on
pseudo-strata and pseudo-primary sampling unit (PSU) information contained in
the data. %
We report the cross-validated AUC and calibration of this model on this second
set of observations in Figure \ref{fig:logistic_regression_calandauc_ncvs}.
Despite the wide confidence intervals, the model predictions are well
calibrated and the AUCs are above 0.6 across all four types of offenses. 

Simultaneously, we conduct a sensitivity analysis to compare the estimates produced by
the logistic regression with those obtained from a SuperLearner, which
represents a more flexible approach \citep{van2007super, polley2010super}. For this purpose, we
train a logistic regression model, a logistic Lasso model \citep{tibshirani1996regression},
a multilayer perceptron with one hidden layer, a Naive Bayes classifier, and a
random forest model \citep{breiman2001random}. We tune
the hyperparameters of each of these models separately via cross-validation. To account for the survey weights 
in training, we resample
the observations within each data fold selecting an observation with probability
proportional to the corresponding survey weight. We select the set of hyperparameters that
yield the highest average AUC across the four crime types. On the second subset
of the data, we select by cross-validation the weights that correspond to the
convex combination of the predictions produced by these models achieving the
best predictive performance. As we describe in Section \ref{sec:results}, the estimates of the likelihood of police notification generated by
this nonparametric model are close to those obtained through the logistic
regression approach.

\subsection{Estimation on NIBRS}
In the next stage, we use the results from the weighted logistic regression analysis on NCVS data. These results help us calculate the likelihood that the police will be notified about each individual incident in the NIBRS dataset.
Particularly low values of the assessed probabilities would represent a
potential violation of the positivity assumption, which would skew our
estimates. Accordingly, we examine the predictions generated by the two models
across crime types. The smallest detected values range between $0.05$ for sex
offenses to above $0.1$ for the other offense types. Thus, the results of our
analysis on NIBRS will not be overly influenced by a few outliers.   
We first estimate the
total number of crimes $N$, the rate of police notification $\pi^*$, and the
arrest rate $q^*$, along with their corresponding variances. Then, we follow the procedure described in
Section \ref{sec:methods_q} to estimate the likelihood
of arrest conditional on covariates, $q(X;\theta_0)$, for incidents with individual
offenders via logistic regression.

We additionally perform a number of robustness checks to assess the sensitivity of the downstream estimates to the modeling assumptions.  First, we repeat these analyses using estimates of the likelihood of police notification obtained using the SuperLearner in place of the weighted logistic regression.
We also investigate whether the logistic regression for $q(X;\theta_0)$ may be
misspecified (provided that the model for $\pi(Z;\gamma_0)$ is correct), using  ``focal slope'' model diagnostics proposed by \citet{buja2019models}.
Using these graphical tools, we analyze how the coefficient estimates (specifically, offender's race) change when fitting the regression model on various configurations of the regressor distributions. To implement the reweighting procedure, we proceed as follows. We
first construct a grid of five evenly spaced values for the numeric features,
and use the grid values of $\{0,1\}$ for the binary features. For each regressor
separately, we split the observations into groups based on the grid's cell
center that is closest to each observation's feature value in absolute distance.
For each feature-grid cell pair, we then obtain 50 estimates of the logistic
regression coefficients by bootstrap resampling 10,000 observations from the
given group. We conclude the discussion by presenting our final analysis of
incidents with one or more offenders which employs the GEEs framework described
in Section \ref{sec:methods_q_gee}.

\section{Results}\label{sec:results}
This section is organized as follows. 
We first empirically demonstrate how our approach, by virtue of accounting for covariate shift, strictly improves upon prior analyses such as \citet{blumstein1979estimation}. 
Then, we assess racial disparities 
in the rates of police notification, %
and disparities in arrest rates on all crimes and on only those known to law enforcement. %
This first part of the analysis focuses 
on crime incidents with individual offenders. 
Next, we present the regression results 
for the estimation of the likelihood of arrest 
conditional on crime characteristics via GEEs. 
When describing results as statistically significant, we apply a significance level of 0.01.

\subsection{The necessity of accounting for covariate shift}
In their approach, \citet{blumstein1979estimation} assume
that all incidents in NIBRS are equally likely to be reported to law enforcement once we condition on the offender's race and the crime type. 
Their method naively estimates the number of actual crimes
underlying the reported crimes in NIBRS
by applying the (fixed) ratio 
of actual crimes to reported crimes in NCVS.
However, because reporting rates may vary according to crime characteristics,
and because of potential covariate shift between NCVS and NIBRS
in the distribution of crime characteristics, 
these estimates may be significantly biased.

Both phenomena are observed in the data. As previously discussed, covariate shift arises 
in part due to the different geographical coverage of NIBRS and NCVS. 
For example, 35\% of the offenses of simple assault 
known to law enforcement in 2006--2015 NCVS data 
occur in the southern regions of the US, 
compared to 60\% of those in the NIBRS. 
As another example, in about 60\% of the offenses 
of aggravated assault recorded in NCVS 
the offender is known to the victim. 
By contrast, this occurs in 85\% of the cases in NIBRS data. 
The coefficient estimates produced 
by the logistic regression fitted on NCVS data  
reveal that the likelihood of reporting varies
across most of the crime characteristics considered by our analysis,
often quite substantially (see the Appendix). 
Unlike \citet{blumstein1979estimation}, 
we account for these variations in our analysis.

\subsection{Racial disparities in crime reporting} %

In the available NIBRS data, 59\% of all offenders are white. We estimate that the NIBRS data capture 
44\% (standard error=5\%) and 48\% (5\%) of all violent offenses 
committed by white and black offenders, respectively. 
Equivalently, slightly more than half of the crimes 
that occur in the jurisdictions covered by NIBRS are 
not reported to law enforcement for both racial groups. 
The lower reporting rates for white offenders relative to black offenders indicates a (not statistically significant) marginal
overrepresentation of black offenders in
the data recorded by police agencies compared to their representation in the
larger population of offenders. 
More specifically, we estimate that 61\% of all
crimes that occur are committed by white offenders. 

Table
\ref{tab:summary_stats_nibrs} shows the breakdown of the rates of police
notification by offense types and offenders' racial groups. Sex offenses are the least likely to be reported to police, with only one in four incidents being reported compared to one
in two for the other crime types. Since these crimes are unlikely to be reported
to law enforcement, the estimates suffer from large variance. We find
that crimes with black offenders are associated with higher rates of reporting
than those with white offenders across all offense types. However, the
regression model fitted on NCVS data reveals that there is only a weak and not
statistically significant association between reporting and the offender's
racial group once other crime characteristics are taken into account. 

It is possible that the logistic regression model fitted on NCVS data is
misspecified. Thus, we conduct an analysis of the reporting rates where the
likelihood of police notification is estimated via the SuperLearner.
By comparing the estimates produced by the two models on NIBRS data, we find that the estimates are close for the crimes of robbery
and assaults.
The estimates of arrest rates obtained using the two models
are virtually identical for robbery and assault, even when we condition on the
offender's racial group (the SuperLearner estimates are within a 2\% difference from
those in Table~\ref{tab:summary_stats_nibrs}). 
For sex offenses, on the other hand, the logistic
regression tends to underestimate the likelihood of reporting compared to the
SuperLearner. This underestimation is quite substantial. 
The resulting rates of police notification based on the SuperLearner are larger than those produced by the
logistic regression. %

\subsection{Racial disparities in arrest rates}

\begin{table}
  \centering
\begin{threeparttable}[b]
  \caption{{\normalfont Summary Statistics (2006--2015 NIBRS Data): Estimating Unreported Crimes Using the Likelihood of Police Notification Computed from NCVS Data.} 
  }\label{tab:summary_stats_nibrs}
  \begin{tabular}{lcccc}
    \toprule
  Variable & sex offense & robbery & aggravated assault & simple assault \\ 
    \% police notification &  &  &  &  \\ 
    \textbullet\; black offenders & 21\% (19\%) & 55\% (6\%) & 63\% (4\%) & 47\% (5\%) \\ 
  \textbullet\; white offenders & 19\% (17\%) & 51\% (7\%) & 60\% (4\%) & 45\% (4\%) \\ 
    \midrule 
    \% arrests (reported crimes) &  &  &  &  \\ 
     \textbullet\; black offenders & 22\% ($<$1\%) & 17\% ($<$1\%) & 40\%  ($<$1\%) & 38\%  ($<$1\%) \\ 
  \textbullet\; white offenders & 24\%  ($<$1\%)  & 33\%  ($<$1\%)  & 56\%  ($<$1\%) & 50\%  ($<$1\%) \\ 
    \midrule
    \% arrests (all crimes) &  &  &  &  \\ 
\textbullet\; black offenders & 5\% (20\%) & 9\% (3\%) & 25\% (3\%) & 18\% (4\%) \\ 
  \textbullet\; white offenders & 5\% (19\%) & 17\% (5\%) & 34\% (4\%) & 22\% (4\%) \\ 
     \bottomrule
  \end{tabular}
\begin{tablenotes}[flushleft]
\item Notes: Standard errors are reported within
  parentheses. The summary statistics are computed on incidents with
  only one offender. 
The likelihood of police notification is estimated
on NCVS data via logistic regression with survey weights using the
methodology described in Section \ref{sec:emp_analysis}. 
\end{tablenotes}
\end{threeparttable}
  \end{table}

We now turn to the estimation of arrest rates. Overall, 49\% (standard
error$<$1\%) of the offenses known to law enforcement involving white offenders
resulted in arrest, compared to 37\% ($<$1\%) of those involving black offenders.
Table \ref{tab:summary_stats_nibrs} reveals that arrest rates are similar across
racial groups for sex offenses, while robbery and assault incidents
white offenders result in arrest considerably more often than those with black
offenders. Past works on NIBRS have reached qualitatively similar conclusions \citep{d2003race,
lantz2019co}. Despite the lower crime reporting rates, crimes with white
offenders remain more likely to result in arrests than those with black
offenders once unreported crimes are accounted for. Overall, arrest rates for crimes
are 21\% (7\%) for white offenders  and 17\% (5\%) for black offenders. 
Table \ref{tab:summary_stats_nibrs} shows that arrest rates are higher for
white offenders in case of assaults and robbery, and are comparable across racial groups in case sex
offenses. As in the observed police data, these rates greatly vary across
offense types: Arrests occur in about one in twenty sex offenses and one in five simple assaults. 
The sensitivity analysis via the SuperLearner produces quantitatively similar results except for sex offenses.

\begin{table}[t]
    \centering
    \caption{Regression Results (2006--2015 NIBRS Data): Assessing Racial Differences in Arrest Likelihood for Single-Offender Incidents, Using 2003--2020 NCVS Data for Police Notification Estimates.}\label{tab:regression_q} \resizebox{\textwidth}{!}{%
    \begin{tabular}{lllll}
      \toprule
      \multicolumn{1}{l}{\textbf{Variable}} & 
    \multicolumn{1}{c}{\textbf{Sex offense}} & 
      \multicolumn{1}{c}{\textbf{Robbery}} & 
    \multicolumn{1}{c}{\textbf{Aggravated assault}} & 
      \multicolumn{1}{c}{\textbf{Simple assault}} \\
      \midrule
Intercept & \phantom{-}0.05 (0.01)*** & \phantom{-}0.48 (0.08)*** & \phantom{-}0.46 (0.06)*** & \phantom{-}0.20 (0.02)*** \\ 
  Age of offender & \phantom{-}1.01 (0.00)*** & \phantom{-}1.01 (0.00)*** & \phantom{-}1.01 (0.00)*** & \phantom{-}1.00 (0.00)    \\ 
  Off. is male & \phantom{-}0.99 (0.08)    & \phantom{-}0.87 (0.05)*   & \phantom{-}0.89 (0.04)*   & \phantom{-}0.90 (0.05).   \\ 
  Off. is white & \phantom{-}1.00 (0.07)    & \phantom{-}1.23 (0.06)*** & \phantom{-}1.03 (0.05)    & \phantom{-}1.01 (0.06)    \\ 
  Age of victim & \phantom{-}1.00 (0.00)    & \phantom{-}1.01 (0.00)*** & \phantom{-}1.01 (0.00)*** & \phantom{-}1.01 (0.00)*** \\ 
  Victim is male & \phantom{-}0.84 (0.06)**  & \phantom{-}0.87 (0.04)**  & \phantom{-}0.96 (0.04)    & \phantom{-}0.93 (0.05)    \\ 
  Victim is white & \phantom{-}0.83 (0.07)*   & \phantom{-}1.00 (0.06)    & \phantom{-}1.05 (0.06)    & \phantom{-}1.02 (0.07)    \\ 
  Off. is acquaintance & \phantom{-}0.87 (0.06)*   & \phantom{-}1.34 (0.07)*** & \phantom{-}0.98 (0.05)    & \phantom{-}0.67 (0.04)*** \\ 
  Off. is family member & \phantom{-}1.20 (0.14)    & \phantom{-}2.12 (0.19)*** & \phantom{-}1.54 (0.12)*** & \phantom{-}1.24 (0.12)*   \\ 
  Off. is intimate partner & \phantom{-}1.33 (0.13)**  & \phantom{-}2.30 (0.18)*** & \phantom{-}1.92 (0.13)*** & \phantom{-}1.58 (0.13)*** \\ 
  Minor injury & \phantom{-}1.70 (0.11)*** & \phantom{-}1.24 (0.06)*** & \phantom{-}1.45 (0.07)*** & \phantom{-}1.85 (0.10)*** \\ 
  Serious injury & \phantom{-}2.74 (0.28)*** & \phantom{-}1.99 (0.12)*** & \phantom{-}2.19 (0.14)*** &  \\ 
  During day & \phantom{-}0.90 (0.05).   & \phantom{-}1.24 (0.05)*** & \phantom{-}0.95 (0.04)    & \phantom{-}0.93 (0.04)    \\ 
  Private location & \phantom{-}1.33 (0.10)*** & \phantom{-}1.08 (0.05).   & \phantom{-}1.40 (0.07)*** & \phantom{-}1.34 (0.08)*** \\ 
  Firearm present & \phantom{-}1.03 (0.16)    & \phantom{-}0.97 (0.10)    & \phantom{-}0.99 (0.10)    &  \\ 
  Other weapon present & \phantom{-}0.90 (0.13)    & \phantom{-}0.94 (0.10)    & \phantom{-}0.90 (0.08)    & \phantom{-}0.85 (0.10)    \\ 
  Multiple offenses & \phantom{-}1.90 (0.03)*** & \phantom{-}1.58 (0.04)*** & \phantom{-}1.14 (0.01)*** & \phantom{-}1.15 (0.01)*** \\ 
  Offense only attempted & \phantom{-}0.86 (0.12)    & \phantom{-}0.98 (0.09)    &  &  \\ 
  MSA, central city & \phantom{-}0.68 (0.05)*** & \phantom{-}0.89 (0.05).   & \phantom{-}0.94 (0.05)    & \phantom{-}0.92 (0.06)    \\ 
  MSA, not central city & \phantom{-}0.84 (0.07)*   & \phantom{-}1.02 (0.06)    & \phantom{-}1.10 (0.06)    & \phantom{-}1.02 (0.07)    \\ 
  Nb. of officers per 1000 capita (ORI) & \phantom{-}1.00 (0.00)*** & \phantom{-}0.99 (0.00)*** & \phantom{-}0.99 (0.00)*** & \phantom{-}1.00 (0.00)*** \\ 
Log population served (ORI) & \phantom{-}0.96 (0.00)*** & \phantom{-}0.83 (0.01)*** & \phantom{-}0.88 (0.00)*** & \phantom{-}0.91 (0.00)*** \\ 
       \bottomrule
    \end{tabular}
    }
    {\raggedright {\footnotesize Significance codes: $p<0.001$ ‘***’, $p<0.01$ ‘**’, $p<0.05$ ‘*’, $p<0.1$ ‘.’}. \\
    Notes: The table shows the odds ratios of the logistic regression coefficients
    for $q$, the likelihood of arrest that accounts for unreported crimes. The
    model is fitted on 2006--2015 NIBRS data and uses the estimates of $\pi$
    obtained from 2003--2020 NCVS data. Standard errors are reported inside
    parentheses. Significance codes correspond to the p-values ($p$) of Wald
    tests to assess the statistical significance of the odds ratios. Year- and
    state-level fixed effects are included in the regression model but are
    omitted from the table. ``ORI'' stands for ``originating agency identifier'', a regressor whose value is specific to that law enforcement agency.\par} 
    \end{table}

\subsection{Racial disparities in the likelihood of arrest accounting for crime characteristics}
We estimate the likelihood of arrest conditional on crimes
characteristics via the two-step logistic regression detailed in Section
\ref{sec:methods}. The resulting odds ratios of the coefficient estimates are reported in Table
\ref{tab:regression_q}. In case of robbery offenses, we find that there is a
positive and statistically significant association between whether the offender
is white and the likelihood that the incident results in arrest. Provided that our
model is correctly specified, these results would indicate that white offenders
are more likely to be arrested for robbery than black offenders, ceteris
paribus. The estimates of this coefficient for the other types of crimes are
close to zero and not statistically significant. Thus, the estimated disparities disappear
once we account for crime characteristics other than the offender's race. 

One outstanding concern is that our logistic regression model
estimated on NCVS data may not accurately capture the location-specific patterns in crime
reporting existing in the data (e.g., due to omitted variable bias or modeling
misspecification). For example, by studying restricted-use NCVS data
\citet{baumer2002neighborhood} and \citet{xie2012racial} report significant
variations in crime reporting rates across neighborhoods. Although the available data do
not allow us to analyze reporting rates at the level of the individual law
enforcement agencies, we can still assess whether regional patterns are
accounted for by employing a flexible modeling approach. Thus, we run the
two-step regression analysis using the estimates of the SuperLearner in place of 
those from the logistic regression on NCVS. The estimates of the offender's race
coefficients produced by this approach are close to those
presented in Table \ref{tab:regression_q}.  

We also analyze how the odds ratios of the offender's race coefficient estimates vary under
various configurations of the covariates distributions through the focal slope
model diagnostics described in Section \ref{sec:emp_analysis}. The results of
the diagnostics are reported in the Appendix. We find that both the magnitude and the direction of the coefficients
estimates vary with the characteristics of the crimes. The most notable pattern
is the change in the association between the likelihood of arrest and the
offender's racial group when either only black or white victims are considered. 
This suggests an interaction between the two covariates. For example, in case of
assaults we observe that the estimates association between the offender being white and the likelihood of arrest is close to zero when the victims being considered are white individuals, but it is large and positive in case of black victims. This suggests
that, ceteris paribus, white offenders may be more likely to be arrested than
black offenders only when they commit interracial crimes.
For sex offenses, however, we find that the association is negative in case of
crimes with black victims and close to zero otherwise.

The estimates in Table \ref{tab:regression_q} 
can be compared with those obtained from a logistic
regression model fitted directly on NIBRS data without accounting for unreported
crimes (see the Appendix). The association between the offender being white and arrests estimated by the models without adjustments for unreported crimes is stronger (and positive) compared to the estimates in Table \ref{tab:regression_q} in case of robberies and assaults. 
We focus on two other examples of differences between coefficient estimates that stand out. 
First,
the regression model that accounts for unreported crimes estimates a stronger
positive association between the victim being injured (vs. no injury) and the
occurrence of an arrest. This pattern could be explained by the fact that
incidents without injuries are less likely to be reported to law enforcement
(see the Appendix). Second, the logistic
regression without adjustments estimates a negative and strong association between the
presence of a firearm (vs. no weapon) and arrests. This association disappears
once unreported crimes are accounted for, again potentially because incidents with
firearms tend to be more likely to be reported.

\subsection{Racial disparities in arrests for incidents with multiple offenders}
The estimates of the crime reporting rates for incidents involving more than one
offender are similar to those for crimes with individual offenders in Table
\ref{tab:summary_stats_nibrs}. However, we find that arrest rates computed solely on police-recorded data for these
incidents are substantially lower in case of aggravated assaults (by more than
10\%), followed by robbery and simple assaults (within a 5\% difference). By
contrast, arrest rates for sex offenses with multiple offenders are marginally
higher than those of crimes with individual offenders. 
Arrest rates shrink proportionally within racial groups once unreported
crimes are taken into account. We continue to observe that white offenders are arrested more often than
black offenders across all crime types. The only exception is robbery for which
the reduction is limited to white offenders but it is not large enough to
reverse the sign of the disparity.

We first fit the two-step regression model using GEEs on only incidents with
multiple offenders. The model estimates that, conditionally on other crime characteristics, white offenders
face a higher likelihood of arrest than black offenders across all offense types. 
We next fit the same model specification on incidents with both single and multiple offenders. In doing so, we need to keep in mind that only about one in ten
incidents of assault and sex offenses are committed by multiple offenders, and
these incidents generally have few offenders. An exception is
represented by robbery for which half of the incidents involve multiple
offenders. We find that white offenders are associated with a
higher likelihood of arrest in case of robbery (estimate is 0.2 with standard error equal to 0.04),
while the other estimated associations are virtually zero (full results in Table
\ref{tab:regression_q_gees}).

\section{Limitations}\label{sec:limitations}
Our empirical analysis relies on a series 
of assumptions about the modeling and data that may not hold in reality. 
One key limitation of the modeling is that the assumed independence across incidents may be violated.
For example, when the same offender is part
of multiple separate crime incidents,
arrest outcomes become correlated. 
Given identifying offender-level information, we could, in principle, correct 
the variance estimates to account for this dependence \citep{andrews1992improved, white2014asymptotic}, but such data is not available. 

Another limitation of our modeling approach
is the potential misspecification of the regression function
used in estimating the likelihood of police notification. 
Through the sensitivity analysis in Section \ref{sec:results}, 
we have shown that employing a more flexible classifier on
NCVS data yields results that are similar 
to those of the logistic regression for most offense types. Leveraging the model diagnostics, we have shown that
the logistic regression model fitted on NIBRS was misspecified.
Consequently, the coefficient estimates require careful interpretation; 
see \citet{buja2019models} and \citet{berk2019assumption} for detailed treatments of this topic, 
and \citet{fogliato2021validity} for a discussion 
of the limitations of similar approaches.

Certain variation in reporting
rates may also fail to be captured specifically 
due to omitted variable bias. 
For example, the NCVS data that are used in our analysis 
contain little information on the geographical location.  Obtaining and incorporating this information may influence the results 
\citep{baumer2002neighborhood, xie2012racial, xie2019crime}. 
Although analyses of the restricted-use NCVS data 
would overcome some of these issues, 
relevant pieces of information, 
such as the specific location of the victimization,
may simply be missing from the data \citep{cernat2021estimating}.

Even more importantly, our analysis suffers 
from limitations related to the nature of the data. 
These limitations are not unique to our study; they have been discussed in a plethora of
criminological works. 
Firstly, the recorded data may be of poor quality. 
With respect to survey data, measurement errors arising from sampling design, data
collection, victims' recollection of the events and untruthful reporting affect
the quality of the data. What victims report in the survey may not always
coincide with the same information that is recorded in police data.

Information in NIBRS may not always accurately reflect the characteristics of the crime incident. 
In this work, for example, we have observed that NCVS respondents 
were far more likely to report serious injuries in case of sex offenses 
than what was recorded in NIBRS data. 
This pattern is unlikely to be explained 
solely by differences in the underlying populations. 
Overall, police data can be seen as an artefact 
of a manipulation process \citep{richardson2019dirty}.
It is also possible that instances of wrongful arrests, 
which we do not consider in the analysis,
may be present in the data \citep{loeffler2019measuring}. 

In addition to issues of data quality, the data are missing certain information that we hypothesize being relevant to our analysis.
For example, we included Hispanics in the analysis
because, as ethnicity information is not always recorded 
(and when recorded it can be imprecise), this population could not be entirely excluded from the sample.
However, there is evidence that this ethnic group 
may be characterized by unique offending and reporting behaviors
\citep{steffensmeier2011reassessing,roberts2011hispanic,rennison2010investigation}. 

One further limitation of our analysis concerns the
matching of offense categories between the NCVS and the NIBRS, that do not
perfectly map.
However, even if the definitions were to fully overlap,
the type of offense that is reported by the victim 
may not correspond to the coding of the same offense 
done by law enforcement. 
This potential issue may affect mainly simple assaults, 
which represent the least serious type of crime 
considered in this analysis.
We also do not consider incidents where the victim does not personally see
the offender, which represent a minimal share of all incidents reported by
NCVS respondents. Thus, together with the fact that not all reported crime may
be recorded, this implies that our estimates of the arrest rates represent upper
bounds of the true quantities.

\section{Discussion}\label{sec:discussion}

In this work, we have proposed estimators
of the rates of police notification and of arrest 
for nonfatal violent crime on NIBRS 
that leverage data of unreported crimes from NCVS. 
These estimators are consistent 
and asymptotically normal under some assumptions. 
Our empirical investigation of racial disparities 
revealed that incidents are marginally more likely 
to be reported to the police 
when the offender is black.
However, in cases of assaults and robbery, crimes with black offenders are generally less likely to result in arrests.
These differences are small after accounting for crime characteristics.
Additionally, the model diagnostics showed 
that the direction of these disparities 
varies with crime characteristics.

We envision three directions in which the proposed methodology 
can be further developed. 
First, we could employ nonparametric methods 
in place of the logistic regression model.
In this work, we obtained asymptotic normality for a two-step
estimation approach where logistic regression was used in both steps.
Nonparametric approaches that yield similar convergence rates 
could be applied in the first step.
For example, the method of kernels 
introduced in \citet{racine2004nonparametric}
can handle both categorical and continuous data.
Although our empirical analysis did not uncover 
significant differences in the estimates 
of the likelihood of police notification produced by parametric and nonparametric models, 
the latter is more flexible and thus 
may be more suitable in certain applications. 

Secondly, mixed effects logistic regression models 
could be employed for the estimation of the likelihood of arrest. 
This represents a modeling approach 
often used in the social sciences.
In this work, we have employed a model 
that does not account for city- or agency-level effects, 
which may drive many of the disparities, 
e.g., see the results of \citet{fogliato2021validity}.   
Third, we assume covariate shift between NCVS and NIBRS.
Future work could use a reweighted loss 
to adjust for the shift in the two datasets. 

Our study opens multiple avenues of research
in the criminology field as well.
Despite a longstanding interest 
in the ``dark figure of crime'' \citep{skogan1977dimensions},
how to accurately characterize this figure 
remains challenging and not well understood. 
Our methodology represents one way through
which it can be described and its magnitude be assessed. 
It would be interesting to compare results 
obtained through our methodology 
with those from the simulation-based approach 
proposed by \citet{buil2021accuracy}.
Future work may also leverage information about the socioeconomic status of the victim, which is available in NCVS data, and of the characteristics of the population in the police agency, which can be obtained from auxiliary data sources and merged with NIBRS data. These aspects were not considered in our work.  

Similar to past studies on NCVS and NIBRS, the results described in this paper build on several assumptions. Some of these assumptions may be violated.  
We hope that, over time, 
police records will become more accurate and comprehensive,
and that detailed information 
about incidents will be made available through NIBRS,
allowing for improved analyses.

\begin{acks}
The authors thank Richard Berk for valuable feedback and suggestions, and David Buil-Gil for helpful discussions. We also thank reviewers, associate editor, and editor for their insightful comments. 
\end{acks}

\clearpage

\begin{appendix}

\section{Additional results}\label{sec:app_results}

This section contains additional results. Table \ref{tab:regression_pi} shows
the odds ratios of the coefficients estimates of the logistic regression model run on NCVS data.
Table \ref{tab:regression_alpha} presents the results of the regression analysis
targeting the likelihood of arrest for crimes known to law enforcement with
individual offenders. Figure
\ref{fig:sensitivity_logistic_vs_superlearner} shows the predicted likelihood of crime reporting $\pi$
produced by logistic regression and SuperLearner for each observation in NCVS data. 
Lastly, Figure \ref{fig:model_diagnostics} shows the focal slope model diagnostics.

\editedinline{
\begin{table}[!htb]
\centering
\begin{threeparttable}[b]
  \caption{Regression Results (2003--2020 NCVS Data): Estimating Police Notification Likelihood for Single-Offender Incidents Using Logistic Regression with Survey Weights.}
  \centering
  \label{tab:regression_pi}
  \begin{tabular}{ll}
    \toprule
  \textbf{Variable} & \textbf{Odds ratio estimate} \\ 
    \midrule
  Age of off. 12-14 & \phantom{-}1.36 (0.43)    \\ 
  Age of off. 15-17 & \phantom{-}2.17 (0.58)**  \\ 
  Age of off. 18-20 & \phantom{-}2.00 (0.62)*   \\ 
  Age of off. 21-29 & \phantom{-}2.96 (0.84)*** \\ 
  Age of off. 30+ & \phantom{-}2.66 (0.72)*** \\ 
  Off. is male & \phantom{-}0.87 (0.08)    \\ 
  Off. is white & \phantom{-}0.95 (0.08)    \\ 
  Age of victim & \phantom{-}1.01 (0.00)*** \\ 
  Victim is male & \phantom{-}0.90 (0.07)    \\ 
  Victim is white & \phantom{-}0.85 (0.09)    \\ 
  Off. is acquaintance & \phantom{-}0.71 (0.06)*** \\ 
  Off. is family member & \phantom{-}0.94 (0.13)    \\ 
  Off. is intimate partner & \phantom{-}0.95 (0.12)    \\ 
  Minor injury & \phantom{-}1.41 (0.11)*** \\ 
  Serious injury & \phantom{-}2.83 (0.35)*** \\ 
  During day & \phantom{-}0.96 (0.07)    \\ 
  Private location & \phantom{-}1.37 (0.12)*** \\ 
  Firearm present & \phantom{-}1.39 (0.27).   \\ 
  Other weapon present & \phantom{-}0.81 (0.14)    \\ 
  Offense only attempted & \phantom{-}0.82 (0.14)    \\ 
  MSA, central city & \phantom{-}0.90 (0.09)    \\ 
  MSA, not central city & \phantom{-}1.03 (0.10)    \\ 
  Crime is robbery & \phantom{-}0.85 (0.14)    \\ 
  Crime is sex offense & \phantom{-}0.23 (0.05)*** \\ 
  Crime is simple assault & \phantom{-}0.57 (0.10)**  \\ 
     \bottomrule
  \end{tabular}
  \begin{tablenotes}[flushleft] \item Significance codes: $p<0.001$ ‘***’, $p<0.01$ ‘**’, $p<0.05$ ‘*’, $p<0.1$ ‘.’.
  \end{tablenotes}
  \end{threeparttable}
  \end{table}
  }

\begin{table}[!htb]
  \centering
\begin{threeparttable}
\caption{Regression Results: Assessing Racial Disparities in Arrest Likelihood Based on Known Incidents to Police Agencies, Using Crime Characteristics $\alpha(X)$.}\label{tab:regression_alpha}
\begin{tabular}{lllll}
  \toprule
  \multicolumn{1}{l}{\textbf{Variable}} & 
\multicolumn{1}{c}{\textbf{Sex offense}} & 
  \multicolumn{1}{c}{\textbf{Robbery}} & 
\multicolumn{1}{c}{\textbf{Aggravated assault}} & 
  \multicolumn{1}{c}{\textbf{Simple assault}} \\
  \midrule
Intercept & \phantom{-}0.42 (0.03)*** & \phantom{-}1.53 (0.16)*** & \phantom{-}1.55 (0.05)*** & \phantom{-}0.93 (0.01)*** \\ 
  Age of offender & \phantom{-}1.01 (0.00)*** & \phantom{-}1.01 (0.00)*** & \phantom{-}1.00 (0.00)*** & \phantom{-}0.99 (0.00)*** \\ 
  Off. is male & \phantom{-}1.15 (0.04)*** & \phantom{-}0.93 (0.03)*   & \phantom{-}0.94 (0.01)*** & \phantom{-}0.98 (0.00)*** \\ 
  Off. is white & \phantom{-}1.04 (0.02)*   & \phantom{-}1.33 (0.03)*** & \phantom{-}1.09 (0.01)*** & \phantom{-}1.05 (0.00)*** \\ 
  Age of victim & \phantom{-}0.98 (0.00)*** & \phantom{-}1.00 (0.00)*   & \phantom{-}1.00 (0.00)*** & \phantom{-}1.01 (0.00)*** \\ 
  Victim is male & \phantom{-}0.90 (0.02)*** & \phantom{-}0.93 (0.02)*** & \phantom{-}1.01 (0.01)    & \phantom{-}0.99 (0.00)*** \\ 
  Victim is white & \phantom{-}0.94 (0.02)*** & \phantom{-}1.12 (0.02)*** & \phantom{-}1.21 (0.01)*** & \phantom{-}1.22 (0.01)*** \\ 
  Off. is acquaintance & \phantom{-}1.24 (0.02)*** & \phantom{-}1.80 (0.04)*** & \phantom{-}1.25 (0.01)*** & \phantom{-}0.85 (0.00)*** \\ 
  Off. is family member & \phantom{-}1.41 (0.03)*** & \phantom{-}2.67 (0.15)*** & \phantom{-}1.98 (0.02)*** & \phantom{-}1.58 (0.01)*** \\ 
  Off. is intimate partner & \phantom{-}1.56 (0.03)*** & \phantom{-}2.91 (0.11)*** & \phantom{-}2.49 (0.02)*** & \phantom{-}2.04 (0.01)*** \\ 
  Minor injury & \phantom{-}1.35 (0.02)*** & \phantom{-}1.02 (0.02)    & \phantom{-}1.27 (0.01)*** & \phantom{-}1.76 (0.00)*** \\ 
  Serious injury & \phantom{-}1.31 (0.03)*** & \phantom{-}1.26 (0.04)*** & \phantom{-}1.38 (0.01)*** &  \\ 
  During day & \phantom{-}0.94 (0.01)*** & \phantom{-}1.35 (0.02)*** & \phantom{-}0.98 (0.01)**  & \phantom{-}0.98 (0.00)*** \\ 
  Private location & \phantom{-}0.98 (0.01)*   & \phantom{-}0.89 (0.02)*** & \phantom{-}1.23 (0.01)*** & \phantom{-}1.07 (0.00)*** \\ 
  Firearm present & \phantom{-}0.74 (0.04)*** & \phantom{-}0.79 (0.02)*** & \phantom{-}0.82 (0.01)*** &  \\ 
  Other weapon present & \phantom{-}1.08 (0.03)*** & \phantom{-}1.07 (0.02)**  & \phantom{-}1.04 (0.01)*** & \phantom{-}0.99 (0.01)*   \\ 
  Multiple offenses & \phantom{-}2.37 (0.05)*** & \phantom{-}1.82 (0.06)*** & \phantom{-}1.20 (0.01)*** & \phantom{-}1.22 (0.01)*** \\ 
  Offense only attempted & \phantom{-}1.01 (0.03)    & \phantom{-}1.12 (0.03)*** &  &  \\ 
 MSA, central city & \phantom{-}0.68 (0.01)*** & \phantom{-}0.90 (0.03)**  & \phantom{-}0.97 (0.01)**  & \phantom{-}0.96 (0.00)*** \\ 
  MSA, not central city & \phantom{-}0.77 (0.01)*** & \phantom{-}0.99 (0.03)    & \phantom{-}1.11 (0.01)*** & \phantom{-}0.99 (0.00)**  \\ 
  Nb. of officers per 1000 capita (ORI) & \phantom{-}1.00 (0.00)*** & \phantom{-}0.99 (0.00)*** & \phantom{-}0.99 (0.00)*** & \phantom{-}1.00 (0.00)*** \\ 
Log population served (ORI) & \phantom{-}0.94 (0.00)*** & \phantom{-}0.80 (0.01)*** & \phantom{-}0.84 (0.00)*** & \phantom{-}0.86 (0.00)*** \\ 
   \bottomrule
\end{tabular}
\begin{tablenotes}[flushleft] \item Significance codes: $p<0.001$ ‘***’, $p<0.01$ ‘**’, $p<0.05$ ‘*’, $p<0.1$ ‘.’. \\
Notes: The table shows the odds ratios of the logistic regression coefficients for $\alpha$,
the likelihood of arrest for incidents known to police agencies, fitted on the
NIBRS data considered in the analysis. Standard errors are
reported inside parentheses. Significance codes correspond to the p-values
($p$) of Wald tests to assess the statistical significance of the odds ratios.
Year- and
    state-level fixed effects are included in the regression model but are
    omitted from the table.
  \end{tablenotes}
\end{threeparttable}
\end{table}

\begin{table}[!htb]
  \centering
\begin{threeparttable}
    \caption{Regression Results (NIBRS Data): Assessing Racial Disparities in Arrest Likelihood for Incidents with One or More Offenders, Using NCVS Data for Police Notification Estimates.}\label{tab:regression_q_gees} %
    \begin{tabular}{lllll}
      \toprule
      \multicolumn{1}{l}{\textbf{Variable}} & 
    \multicolumn{1}{c}{\textbf{Sex offense}} & 
      \multicolumn{1}{c}{\textbf{Robbery}} & 
    \multicolumn{1}{c}{\textbf{Aggravated assault}} & 
      \multicolumn{1}{c}{\textbf{Simple assault}} \\
      \midrule
Intercept & \phantom{-}0.07 (0.01)*** & \phantom{-}0.92 (0.12)    & \phantom{-}0.40 (0.05)*** & \phantom{-}0.18 (0.02)*** \\ 
  Age of offender & \phantom{-}1.01 (0.00)*** & \phantom{-}1.01 (0.00)*** & \phantom{-}1.01 (0.00)*** & \phantom{-}1.00 (0.00)*   \\ 
  Off. is male & \phantom{-}0.88 (0.06).   & \phantom{-}0.82 (0.03)*** & \phantom{-}0.91 (0.04)*   & \phantom{-}1.00 (0.06)    \\ 
  Off. is white & \phantom{-}1.00 (0.06)    & \phantom{-}1.18 (0.04)*** & \phantom{-}1.04 (0.04)    & \phantom{-}1.00 (0.05)    \\ 
  Age of victim & \phantom{-}0.99 (0.00)**  & \phantom{-}1.00 (0.00)*** & \phantom{-}1.01 (0.00)*** & \phantom{-}1.01 (0.00)*** \\ 
  Victim is male & \phantom{-}0.83 (0.06)**  & \phantom{-}0.77 (0.03)*** & \phantom{-}0.76 (0.03)*** & \phantom{-}0.77 (0.04)*** \\ 
  Victim is white & \phantom{-}0.78 (0.06)**  & \phantom{-}0.93 (0.05)    & \phantom{-}0.96 (0.05)    & \phantom{-}0.95 (0.06)    \\ 
  Off. is known & \phantom{-}0.90 (0.06)    & \phantom{-}1.01 (0.04)    & \phantom{-}1.12 (0.05)**  & \phantom{-}0.97 (0.05)    \\ 
  Minor injury & \phantom{-}1.64 (0.11)*** & \phantom{-}1.21 (0.05)*** & \phantom{-}1.50 (0.07)*** & \phantom{-}1.89 (0.10)*** \\ 
  Serious injury & \phantom{-}2.83 (0.29)*** & \phantom{-}1.95 (0.10)*** & \phantom{-}2.30 (0.13)*** &  \\ 
  During day & \phantom{-}0.95 (0.05)    & \phantom{-}1.30 (0.05)*** & \phantom{-}1.00 (0.04)    & \phantom{-}0.94 (0.04)    \\ 
  Private location & \phantom{-}1.47 (0.10)*** & \phantom{-}1.32 (0.05)*** & \phantom{-}1.70 (0.07)*** & \phantom{-}1.73 (0.09)*** \\ 
  Firearm present & \phantom{-}1.05 (0.14)    & \phantom{-}1.08 (0.09)    & \phantom{-}0.97 (0.08)    &  \\ 
  Other weapon present & \phantom{-}1.05 (0.13)    & \phantom{-}1.03 (0.09)    & \phantom{-}1.00 (0.08)    & \phantom{-}0.96 (0.10)    \\ 
  Multiple offenses & \phantom{-}1.84 (0.05)*** & \phantom{-}1.60 (0.03)*** & \phantom{-}1.12 (0.01)*** & \phantom{-}1.14 (0.01)*** \\ 
  Offense only attempted & \phantom{-}0.84 (0.11)    & \phantom{-}0.91 (0.08)    &  &  \\ 
  MSA, central city & \phantom{-}0.60 (0.05)*** & \phantom{-}0.76 (0.04)*** & \phantom{-}0.85 (0.04)**  & \phantom{-}0.83 (0.05)**  \\ 
  MSA, not central city & \phantom{-}0.82 (0.07)*   & \phantom{-}0.94 (0.05)    & \phantom{-}1.08 (0.06)    & \phantom{-}1.01 (0.07)    \\
  Nb. of officers per 1000 capita (ORI) & \phantom{-}1.00 (0.00)*** & \phantom{-}0.99 (0.00)*** & \phantom{-}0.99 (0.00)*** & \phantom{-}1.00 (0.00)*** \\ 
  Log population served (ORI) & \phantom{-}0.95 (0.00)*** & \phantom{-}0.81 (0.00)*** & \phantom{-}0.88 (0.00)*** & \phantom{-}0.91 (0.00)*** \\ 
       \bottomrule
    \end{tabular}
\begin{tablenotes}[flushleft] \item  Significance codes: $p<0.001$ ‘***’, $p<0.01$ ‘**’, $p<0.05$ ‘*’, $p<0.1$ ‘.’. \\
    Notes: The table shows the odds ratios of the regression coefficients
    for $q$, the likelihood of arrest that accounts for unreported crimes, estimated via generalized estimating equations (GEEs). The
    model is fitted on NIBRS data and uses the estimates of the likelihood of crime reporting $\pi$
    obtained from NCVS data. Standard errors are reported inside
    parentheses. Significance codes correspond to the p-values ($p$) of Wald
    tests to assess the statistical significance of the odds ratios. Year- and
    state-level fixed effects are included in the regression model but are
    omitted from the table.
  \end{tablenotes}
\end{threeparttable}
\end{table}

  \begin{figure}[!htb]
    \begin{center}
    \includegraphics[width=0.9\textwidth]{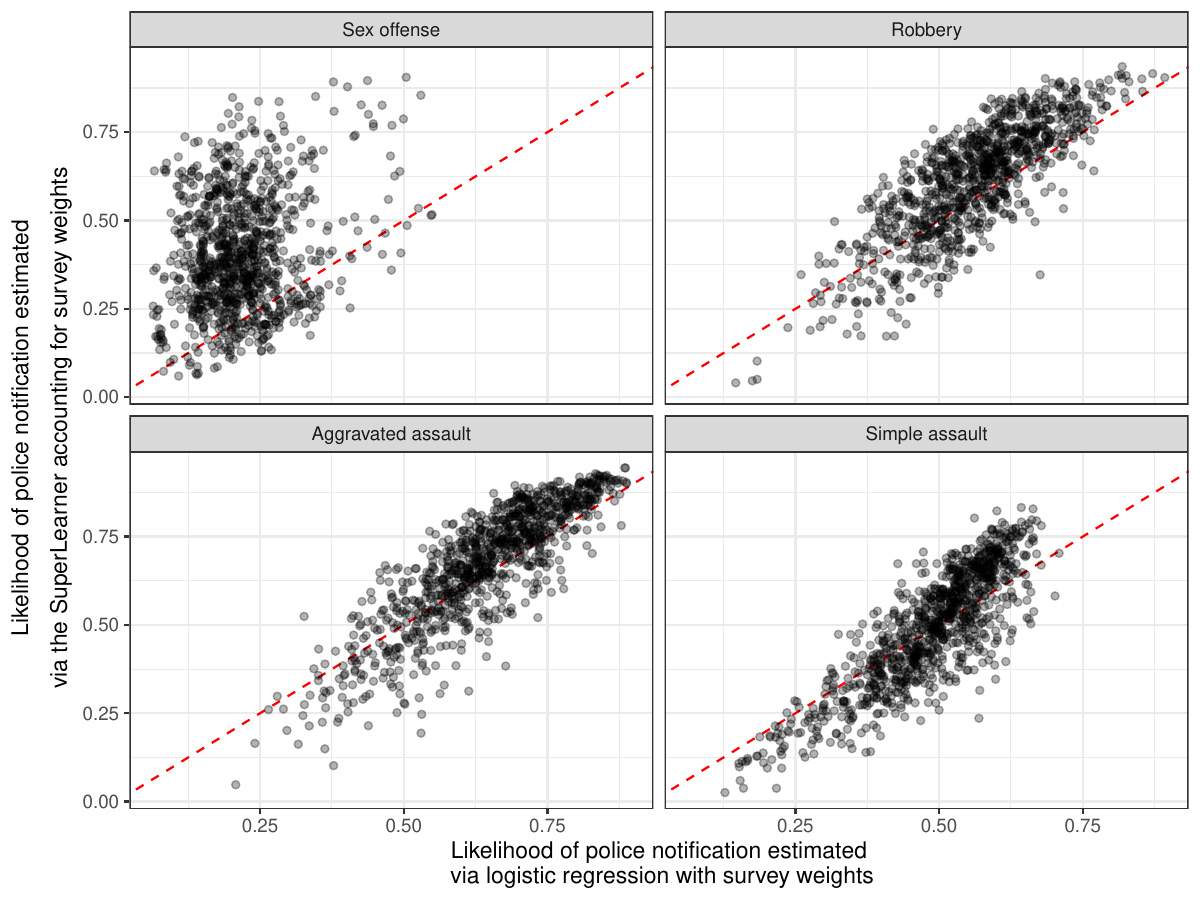}
    \end{center}
    \caption{Estimates of the likelihood of police notification $\pi$ for
    observations in 2006--2015 NIBRS data produced by the logistic regression
    (horizontal axis) and by the SuperLearner (vertical axis) fitted on 2003--2020
    NCVS data. For visualization purposes, we show the estimates relative to 1000
    randomly sampled observations for each crime type.}
    \label{fig:sensitivity_logistic_vs_superlearner}
\end{figure}

\begin{figure}[!htb]
  \centering
  \includegraphics[width=\linewidth,  keepaspectratio]{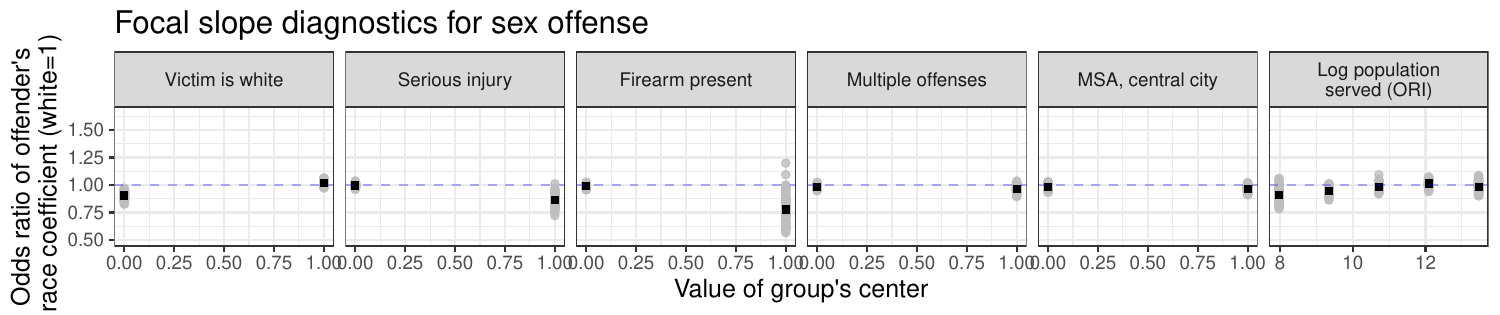}
  \includegraphics[width=\linewidth, keepaspectratio]{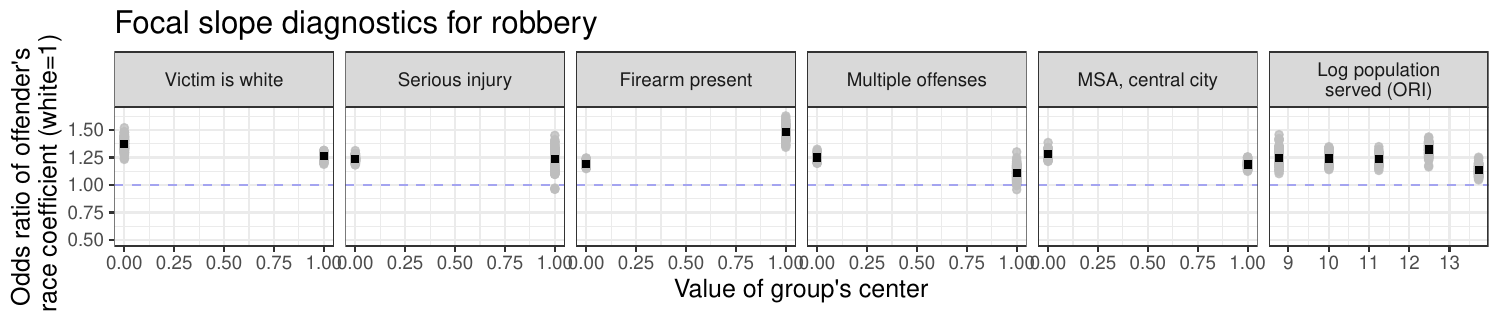}
  \includegraphics[width=\linewidth, keepaspectratio]{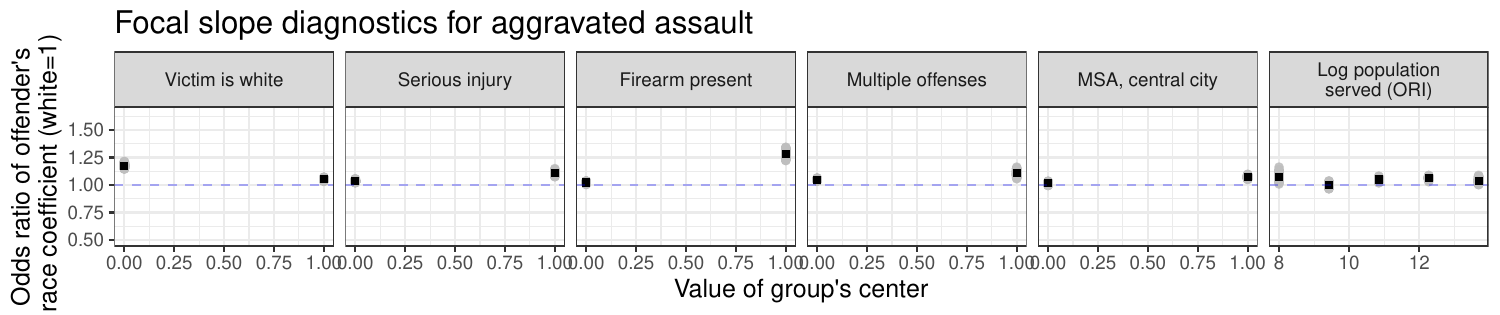}
  \includegraphics[width=\linewidth,  keepaspectratio]{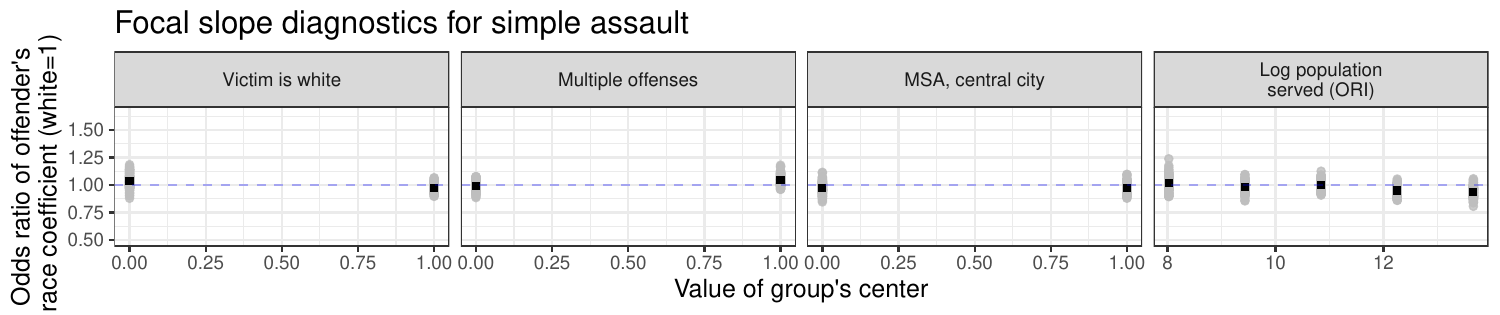}
  \caption{{\normalfont ``Focal slope'' model diagnostic for the logistic
  regression model for $q$, the likelihood of arrest that accounts for unreported crimes, on 2006--2015 NIBRS data whose odds ratios of the coefficients estimates
  are presented in Table~\ref{tab:regression_q}. Only incidents with one
  offender are considered. Methodological details are described
  in Section \ref{sec:emp_analysis}. The grey points correspond to the
  coefficient estimates relative to the offender's race (white=1) obtained by
  fitting the logistic regression model on each of 100 bootstrapped datasets, for
  each variable (panel's title) and variable's grid value (value on the grid,
  horizontal axis). The black dots correspond to the means of such estimates. We
  observe that the size and sign of the values of the black dots vary across the
  range of the regressors. This suggests the presence of interactions between
  race and the regressors, which in turn indicates that our modeling approach is
  misspecified. }}\label{fig:model_diagnostics}
\end{figure}
\clearpage
\section{Details and Proofs}\label{sec:app_proofs}

This section contains the proofs of the results presented in Section
\ref{sec:methods}. In Section \ref{app:proofs_ncvs}, we present the proof relative to the consistency and
asymptotic normality of the coefficient estimates for the logistic regression parameters 
obtained on survey data (Proposition
\ref{app_prop:wlog_asymptotics}). Then, in Section \ref{app:proofs_nibrs} we present the asymptotic properties of
the estimators of the total number of offenses $N$, expected rate of police
notification $\pi^*$, and expected rate of arrest $q^*$ (Lemma
\ref{app_lemma:change_measure}, Propositions \ref{app_prop:asymptotic_N},
\ref{app_prop:asymptotic_pi}, and \ref{app_prop:asymptotic_q}). Lastly, we describe
the results for estimation via the two-step logistic regression (Propositions \ref{app_prop:q_asymptotics_logreg} and
\ref{prop:gee_convergence}). 

\subsection{Estimation on NCVS}\label{app:proofs_ncvs}

We provide some additional details on the framework presented in Section \ref{sec:est_ncvs} before turning to the proof of Proposition \ref{app_prop:wlog_asymptotics}. As a reminder, %
our aim is to make inference on superpopulation parameters. This differs from the finite population framework, for which logistic regression parameter estimation has been studied by \citet{binder1983variances}.
In the following, the subscripts $P^v$ and $\psi$ in the probability $\mbbP$ and expectation $\mbbE$ operators denote superpopulation and sampling design,
respectively.

Formally, let the superpopulation target
parameter $\gamma_0\in \text{Int}(\Gamma)$ be defined by the following moment
condition
\begin{gather*}% \label{eq:survey_superpop}
    \mbbE_{P^v}[h^v(R^v, Z^v;\gamma)] = 0,
\end{gather*}
where $h^v(R^v, Z^v;\gamma) := (R^v-\pi^v(Z^v;\gamma))Z^v$. The parameter
$\tilde{\gamma}\in \text{Int}(\Gamma)$ and the design-based estimator
$\hat\gamma\in\Gamma$ are the solutions to
\citep{lumley2017fitting},
\begin{gather}
    \sum_{i=1}^{N^v} h^v(R_i^v, Z_i^v;\gamma) = 0, \text{ and} \label{eq:survey_full}\\
    \sum_{i=1}^{N^v} w_i I_i h^v(R_i^v, Z_i^v;\gamma) =0 \label{eq:design_based}%
\end{gather}
respectively. 
The estimating equation \eqref{eq:survey_full} is unbiased for $\gamma_0$.
Conditionally on the finite population $V^{N^v}$, equation \eqref{eq:design_based} is
unbiased for $\tilde\gamma$ provided that $\mbbE_{\psi}[I_iw_i]=1$ for
$i=1,\dots,N^v$. Since we only have access to the observations for which $I_i=1$,
our estimation will be based on the estimating equation \eqref{eq:design_based}.
In the presence of endogenous or informative sampling, the estimate $\hat\gamma$ obtained by solving
\eqref{eq:design_based} may not coincide with the one we would obtain by solving
the unweighted estimating equation \citep{solon2015we}.

In order to establish the asymptotic properties of the estimator $\hat\gamma$
obtained by solving equation \eqref{eq:design_based} on the sample $V^{N^v}$, we
assume that the observations we have access to are sampled from a finite number
of strata with known size. Thus, for each stratum, our survey data represent a
sample of the finite population belonging to that stratum, which in turn
represents an \iid sample of the superpopulation distribution specific to that
stratum. 
The following
proposition borrows the setup from Theorem 1.3.9 in \citet{fuller2011sampling} %
and leverages the results of \citet{rubin2005two}.

\begin{proposition}\label{app_prop:wlog_asymptotics} Consider an increasing
    sequence of finite populations where the $N^v$-th population has size $N^v$
    and consists of $H\in\mathbb{Z}_{+}$ strata. The $h$-th stratum is formed by
    the $N^v_h$ observations $\mathcal{F}_{N^vh}=\{(Z^v_{N^vhi},
    R^v_{N^vhi})\}_{i=1}^{N^v_h}$ which represent an \iid sample of $(Z^v_{h},
    R^v_{h})\sim P^v_{h}$, for $h=1,\dots,H$, where $P^v_h$ is the distribution
    of the superpopulation of the specific stratum. Assume that
    $\norm{Z^v_h}_\infty < M$ for some $M>0$ and $h=1,\dots,H$. For the $h$-th
    stratum, we have access to a sample of observations that are drawn from
    $\mathcal{F}_{N^vh}$ according to some sampling design $\psi_{N^vh}$ and let
    $I_{N^vhi}=1$ if the $i$-th observation is selected, and $I_{N^vhi}=0$
    otherwise. Let $\{w_{N^vhi}\}_{i=1}^{N^v_h}$ indicate the set of weights
    associated with the sample in the $h$-th stratum where
    $w_{N^vhi}:=\mbbE_{\psi_{N^v}}[I_{N^vhi}]^{-1}$, and assume that $\max_{h,i}
    w_{N^vhi}<K$ for some $K>0$. We denote with $n^v_{N^vh}$ the (expected or
    fixed) sample size of the $h$-th stratum, with $n^v_{N^v}:=\sum_{h=1}^{H}n^v_{N^vh}$
    the size of the entire survey sample, and with $\lambda:=\lim_{N^v\rightarrow\infty}
    n_{N^v}^v/N^v$ the limit of the size of the surveyed population compared to the
    entire finite population. 
    Consider a sequence of stratified samples that is selected such that
    $N^v_{h}\rightarrow\infty$, $n^v_{N^vh}\rightarrow\infty$, and
    $\lim_{N^v\rightarrow\infty}N^v_{h}/N^v=\lim_{N^v\rightarrow\infty}n^v_{N^vh}/n^v_{N^v}
    = \beta_h\in(0,1]$, for $h=1,\dots,H$. The parameters $\gamma_0$ and
    $\tilde{\gamma}_{N^v}$, and the estimator $\hat\gamma_{N^v}$,
    with $\gamma_0, \tilde{\gamma}_{N^v}\in \text{Int}(\Gamma)$ and
    $\hat\gamma_{N^v}\in\Gamma$, solve respectively
    \begin{gather*}
    \sum_{h=1}^{H}\beta_h \mbbE_{P^v_h}[\mbbE_{\psi_{N^v h}}[h^v(R^v_{N^vhi}, Z^v_{N^vhi};\gamma)]]=0,\\
    G^v_{N^v}(\gamma) := \frac{1}{N^v}\sum_{h=1}^{H} \sum_{i=1}^{N^v_h}
    h^v(R^v_{N^vhi},
     Z^v_{N^vhi};\gamma)=0,\\
     \hat{G}^v_{N^v}(\gamma) := \frac{1}{n^v_{N^v}}\sum_{h=1}^{H} \sum_{i=1}^{N^v_h}
     w_{N^vhi} I_{N^vhi} h^v(R^v_{N^vhi},
     Z^v_{N^vhi};\gamma)=0.
    \end{gather*} 
    Assume that, conditionally on the finite population,
    \begin{equation*}
        \sqrt{n^v_{N^v}}\hat{G}^v_{N^v}(\gamma_{N^v})\overset{d} {\rightarrow}\mathcal{N}\left(0,\sum_{h=1}^H\beta_h\Xi^f_{h}\right)
    \end{equation*}
    as $n^v_{N^v}\rightarrow\infty$ and, in addition, for $\gamma\in\Gamma$, 
    \begin{equation}\label{app_eq:condition_on_sampling_mechanism}
        \lim_{N^v\rightarrow\infty} %
        \frac{1}{n^v_{N^v}}
        \sum_{h=1}^{H} \sum_{i=1}^{N_h} I_{N^vhi}w_{N^vhi}\nabla_{\gamma} h^v(Z_{N^vhi}, R_{N^vhi};\gamma)  = J^v(\gamma)
    \end{equation}
    where the positive definite covariance matrices $\Xi^f_h$, for
    $h=1,\dots,H$, and $J^v(\gamma)$ are nonstochastic in the population, and $J^v:=J^v(\gamma_0)=\lim_{N_h^v\rightarrow\infty}\nabla G^v_{N^v}(\gamma_0)$. Then 
    \begin{equation*}
     (\Sigma^v)^{-1/2}\sqrt{n^v_{N^v}}(\hat\gamma_{N^v} -
    \gamma_0)\overset{d} {\rightarrow}\mathcal{N}(0,I_d) 
     \end{equation*}
     as $n^v_{N^v}\rightarrow\infty$ where $\Sigma^v := (J^v)^{-1} [\sum_{h=1}^{H}\beta_h
     (  \mbbE_{P^v_h}\Xi^f_{h} + \lambda \Xi^s_{h})]  (J^v)^{-1}$ with the following
     matrices %
    \begin{gather*}
        \Xi^s_h := \text{Var}_{P^v_h}\left(h^v(R^v_{h}, Z^v_{h};\gamma_0)\right), \\
        J:=\sum_{h=1}^{H}\beta_h \mbbE_{P^v_h}\left[
    \nabla_\gamma h^v(R^v_h, Z^v_h;\gamma_0) (\nabla_\gamma h^v(R^v_h,
    Z^v_h;\gamma_0))^T \right]. 
    \end{gather*}
 \end{proposition}

 \begin{proof}
    To simplify the notation, we will drop ``$v$'' and ``$N^v$'' from most of
    the subscripts and superscripts.   
    Consistency and asymptotic normality of $\hat\gamma$ for $\gamma_0$ follow from
    Theorem 6.1 of \citet{rubin2005two}, which relies on the following five Assumptions.
    \begin{enumerate}[label=\textbf{C.\arabic*},ref=C.\arabic*]
        \item \label{app_rubin_cond1} $G_{N}(\gamma_0)\overset{p}{\rightarrow}0$ as
        $N\rightarrow\infty$.
        \item \label{app_rubin_cond2} There is a compact neighborhood $\Gamma$ of
        $\gamma_0$ on which with probability one all $G_N(\gamma)$ are
        continuously differentiable and $\nabla_\gamma G_N(\gamma)$
        converge uniformly in $\gamma$ to a nonstochastic limit $J^v(\gamma)$
        that is nonsingular at $\gamma_0$. 
        \item \label{app_rubin_cond3}
        $\sqrt{N}G_N(\gamma_0)\overset{d} {\rightarrow}\mathcal{N}(0,\sum_{h=1}^H\beta_h\Xi^s_h)$
        as $N\rightarrow\infty$.
        \item \label{app_rubin_cond4} Conditionally on the finite population, there is a compact neighborhood $\Gamma$ of
        $\gamma_0$ on which $\nabla_\gamma \hat G_N(\gamma)$ converge uniformly
        in the design probability to limit that is nonstochastic in the design
        probability and coincides with $J^v(\gamma)$ at $\gamma_0$ almost surely.
        \item \label{app_rubin_cond5} Conditionally on the finite population,
        $\sqrt{n}\hat{G}_N(\gamma_N)\overset{d} {\rightarrow}\mathcal{N}(0,\sum_{h=1}^H\beta_h\Xi^f_h)$
        as $n\rightarrow\infty$ where the covariance matrices $\Xi^f_h$ are
        nonstochastic in the superpopulation. 
    \end{enumerate}
    Note that \ref{app_rubin_cond1} is implied by \ref{app_rubin_cond3}. To show that
    \ref{app_rubin_cond3} holds, we can prove that the Lindeberg condition is
    satisfied and then apply the central limit theorem (proposition 2.27 in
    \citet{van2000asymptotic}). For any $\epsilon>0$, 
    \begin{gather}\label{app_eq:rubin_cond1_eq}
        \begin{split}
        \frac{1}{N} \sum_{h=1}^{H} \sum_{i=1}^{N_h} \mbbE_{P_h} [ \norm{(R_{hi} - \pi(Z_{hi};\gamma_0))Z_{hi}}^2 \mathds{1}(\norm{(R_{hi} - \pi(Z_{hi};\gamma_0))Z_{hi}}>\epsilon\sqrt{N})]\\
    < \frac{1}{N}\sum_{h=1}^{H} \sum_{i=1}^{N_h} dM^2 \mathds{1}(\sqrt{d}M>\epsilon\sqrt{N})]
        \end{split}
    \end{gather}
    where we have used the fact that $\norm{Z_{hi}}\leq
    \sqrt{d}\norm{Z_{hi}}_\infty<\sqrt{d}M$ and $|R_{hi} -
    \pi(Z_{hi};\gamma_0)|^2\leq 1$. Then
    $\lim_{N\rightarrow\infty}\mathds{1}(\sqrt{d}M>\epsilon\sqrt{N})=0$, and
    thus the RHS of \eqref{app_eq:rubin_cond1_eq} converges to $0$. 
    In addition, 
    \begin{gather*}
        \lim_{N\rightarrow\infty}\frac{1}{N}  \sum_{i=1}^{H} \sum_{i=1}^{N_h} \text{Var}_{P_h}((R_{hi} - \pi(Z_{hi};\gamma_0))Z_{hi})
        = \lim_{N\rightarrow\infty}\sum_{h=1}^H \frac{N_h}{N} \Xi^s_h = \sum_{h=1}^H \beta_h \Xi^s_h
    \end{gather*}
    where the equality follows from the fact that observations are identically
    distributed within strata, and $\Xi^s_h$ represents the covariance matrix for
    stratum $h$. Condition \ref{app_rubin_cond3} follows from an application of
    the central limit theorem.

    In order to show that \ref{app_rubin_cond2} holds, it suffices to prove that for
    any random vector $\gamma_N\in\Gamma$ converging in probability to
    $\gamma_0$, $\nabla_\gamma
    G_N(\gamma_N)\overset{p}{\rightarrow}J^v(\gamma_0)$ for some
    nonstochastic limit $J:=J(\theta_0)$ (Theorem 1 in \citet{iseki1957theorem}). 
    We can decompose $\nabla_\gamma
    G_N(\gamma_N)$ as follows
    \begin{gather*}% \label{app_eq:rubin_cond2_decomposition}
       \frac{1}{N} \sum_{h=1}^{H} \sum_{i=1}^{N_h} \frac{e^{-\gamma_0 Z_{hi}}}{(1+e^{-\gamma_0 Z_{hi}})^2}Z_{hi}Z_{hi}^T
       + \frac{1}{N} \sum_{h=1}^{H} \sum_{i=1}^{N_h}\bigg[\frac{e^{-\gamma_N Z_{hi}}}{(1+e^{-\gamma_N^TZ_{hi}})^2} - \frac{e^{-\gamma_0 Z_{hi}}}{(1+e^{-\gamma_0^TZ_{hi}})^2}  \bigg] Z_{hi}Z_{hi}^T
    \end{gather*}
   where the first term 
   converges in probability to \begin{equation*}
    J^v(\gamma_0) := \sum_{h=1}^{H}\beta_h \mbbE_{P_h}\left[ e^{-\gamma_0^TZ_h}{(1 + e^{-\gamma_0^TZ_h})}^{-2}Z_hZ_h^T \right]. 
   \end{equation*}
   The second term can be rewritten as
    \begin{gather}\label{app_eq:rubin_cond2_decomposition_te}
        \frac{1}{N} \sum_{h=1}^{H} \sum_{i=1}^{N_h}\bigg[\frac{e^{\gamma_0^TZ_{hi}}(1- e^{(\gamma_N - \gamma_0)^TZ_{hi}}) + e^{-\gamma_0^TZ_{hi}}(1 - e^{(\gamma_0 - \gamma_N)^TZ_{hi}})}{(1+e^{-\gamma_N^TZ_{hi}})(1+e^{\gamma_N^TZ_{hi}})(1+e^{-\gamma_0^TZ_{hi}})(1+e^{\gamma_0^TZ_{hi}})}\bigg]
    \end{gather} %
    which can be upper bounded by $ e^{\norm{\gamma_0}\sqrt{d_z}M}(e^{\sqrt{d_z}M\norm{\gamma_N - \gamma_0}} -1 )$. 
    By the continuous mapping theorem this bound is $o_p(1)$.
    It follows that that \ref{app_rubin_cond2} is verified. 

    Condition \ref{app_rubin_cond5} follows from the Assumptions. For
    a discussions of the specific conditions needed under various sampling
    designs, see Section 3.5 of \citet{thompson1997theory}. Similarly, condition
    \ref{app_rubin_cond4} follows from \eqref{app_eq:condition_on_sampling_mechanism}.

    The result then follows from Theorem 6.1 of \citet{rubin2005two}. 

\end{proof}

\subsection{Estimation on NIBRS}\label{app:proofs_nibrs}

Throughout the proofs, we will use the following lemma. 

\begin{lemma}\label{app_lemma:change_measure} Let
    $f:\mathcal{X}\mapsto\mathbb{R}$. Assume that
    \ref{ass_correctpi}--\ref{ass_matchingpi} hold. Then
    \begin{equation*}
        \mbbE[f(X)] = \mbbE\left[ \frac{f(X)}{\pi(Z;\gamma_0)}\bigg\vert R=1\right]\pi^*.
    \end{equation*}
\end{lemma}

\begin{proof}
    We can show that
    \begin{equation*}
        \mbbE[f(X)] = \mbbE\left[f(X) \frac{R}{\mbbP(R=1|Z,X)} \right] = \mbbE\left[f(X) \frac{R}{\mbbP(R=1|Z)} \right]
    \end{equation*}
    where the first equality follows from the law of iterated expectations, while the second follows from \ref{ass_pioverz}.
    Now, thanks to \ref{ass_matchingpi} and \ref{ass_correctpi} we obtain 
    that 
    \begin{equation*}
                \mbbE\left[f(X) \frac{R}{\mbbP(R=1|Z)} \right] = \mbbE\left[f(X) \frac{R}{\pi(Z;\gamma_0)} \right].
    \end{equation*}
    The result follows. 
\end{proof}

Then we can derive the asymptotic properties of the estimator $\hat N$. 

\begin{proposition}\label{app_prop:asymptotic_N} Consider the conditions of
    Proposition \ref{app_prop:wlog_asymptotics} to be satisfied, and Assumptions
    \ref{ass_correctpi}--\ref{ass_boundedcov} to hold. Then 
    \begin{equation}\label{app_eq:asymptotic_N}
        V_N^{-1/2}\sqrt{n}(\hat N/N - 1)\overset{d} {\rightarrow}\mathcal{N}(0,1)
    \end{equation}
    as $N\rightarrow\infty$ where
    \begin{gather}\label{app_eq:var_N}
        V_N := (\pi^*)^2\left[ \mbbE\left[ \frac{1-\pi(Z;\gamma_0)}{\pi(Z;\gamma_0)^2}\bigg|R=1 \right] + 
        \kappa  W^T\Sigma^v W \right] .
    \end{gather}
    with $W:=\mbbE[e^{-Z^T\gamma_0}\pi(Z;\gamma_0)^{-1}Z|R=1]$.%
\end{proposition}

\begin{proof}
    Consider the following first-order Taylor expansion of $\hat N/N-1$
    \begin{multline*}%\label{app_eq:taylor_hatn_first}
            \frac{1}{N}\sum_{i=1}^N \left( \frac{R_i}{\pi(Z_i;\gamma_0)} -1\right) \\-  
            (\hat\gamma-\gamma_0)^T\frac{1}{N}\sum_{i=1}^N R_i e^{-Z_i^T\gamma_0} Z_i  
            + (\hat\gamma-\gamma_0)^T\frac{1}{N}\sum_{i=1}^N R_i e^{-Z_i^T\tilde\gamma} Z_i Z_i^T (\hat\gamma-\gamma_0)
    \end{multline*}
    where $\tilde\gamma$ is a vector between $\hat\gamma$ and $\gamma_0$. Using
    the Cauchy-Schwarz inequality together with \ref{ass_boundedcov},
    we obtain $((\tilde\gamma - \gamma_0)^T Z_i)^2\leq \norm{Z_i}^2 \norm{\hat\gamma -
    \gamma_0}^2< \sqrt{d_z} M \norm{\hat\gamma - \gamma_0}^2$. Thus, we can rewrite
    $\sqrt{n}(\hat N/N-1)$ as
    \begin{multline}\label{app_eq:taylor_hatn}
        \frac{\sqrt{n}}{\sqrt{N}}\frac{1}{\sqrt{N}}\sum_{i=1}^N \left( \frac{R_i}{\pi(Z_i;\gamma_0)} -1\right)\\ -  
        \sqrt{n^v}(\hat\gamma-\gamma_0)^T\sqrt{\kappa}\frac{1}{N}\sum_{i=1}^N Z_i R_i e^{-Z_i^T\gamma_0} 
        + \sqrt{\kappa} O_p(\sqrt{n^v}\norm{\hat\gamma-\gamma_0}^2).
    \end{multline}
    The last term in \eqref{app_eq:taylor_hatn} can be rewritten as
    $O_p(\norm{\sqrt{n^v}(\hat\gamma-\gamma_0)}^2/\sqrt{n^v}) =
    O_p(1/\sqrt{n^v})$ thanks to Proposition \ref{app_prop:wlog_asymptotics}.
    The first term is a sum of \iid random variables that are bounded by
    \ref{ass_boundedcov} and thus it is asymptotically normal with mean $0$
    thanks to the central limit theorem. By \ref{ass_boundedcov},
    $N^{-1}\sum_{i=1}^N R_i Z_i e^{-Z_i^T\gamma_0}$ in the second term is an
    average of \iid bounded random variables which converges in probability to
    $\mbbE[RZe^{-Z^T\gamma_0}] + O_p(1/\sqrt{N})$. We can then rewrite this
    expectation as $\mbbE[Ze^{-Z^T\gamma_0}\pi(Z;\gamma_0)^{-1}|R=1]\pi^*$
    thanks to Lemma \ref{app_lemma:change_measure}. Then
    $\sqrt{n}(\hat\gamma-\gamma_0)$ is asymptotically normal by Proposition
    \ref{app_prop:wlog_asymptotics} and consequently the second term in
    \eqref{app_eq:taylor_hatn} is asymptotically normal by Slutsky. Note that
    the first two terms in \eqref{app_eq:taylor_hatn} are asymptotically
    independent because they arise from different samples, and thus we have
    proved \eqref{app_eq:asymptotic_N}.
    The variance in
    \eqref{app_eq:var_N} follows by an application of Lemma
    \ref{app_lemma:change_measure}.  
\end{proof}

\begin{proposition}\label{app_prop:asymptotic_pi} Consider the conditions of
    Proposition \ref{app_prop:wlog_asymptotics} to be satisfied, and Assumptions
    \ref{ass_correctpi}--\ref{ass_boundedcov} to hold. Then 
    \begin{equation*}
    V_{\pi^*}^{-1/2}\sqrt{n}(\hat{\pi}^*-\pi^*)\overset{d} {\rightarrow}\mathcal{N}(0,1)
    \end{equation*}
    as $N\rightarrow\infty$ where 
    \begin{gather*}
        V_{\pi^*} := (\pi^*)^2 \left[ \mbbE \left[ \frac{\pi^* - \pi(Z;\gamma_0)}{\pi(Z;\gamma_0)^2} \bigg\vert R=1 \right] + 
        \kappa (\pi^*)^2 W^T \Sigma^v W \right] .  
    \end{gather*}
    with $W:=\mbbE[e^{-Z^T\gamma_0}\pi(Z;\gamma_0)^{-1}Z|R=1]$
\end{proposition}

\begin{proof}
In order to show asymptotic normality, we can first rewrite $\hat{\pi}^* - \pi^*$ as
\begin{multline}\label{app_eq:pi_expansion}
         \frac{\sum_{i=1}^N R_i}{N} - \pi^* + \pi^* \left(1 - \frac{\hat N}{N} \right) +  \left( \frac{N}{\hat N} -1 \right) \left( \frac{\sum_{i=1}^N R_i}{N} - \pi^*\right)\\ + \pi^* \left( \frac{N}{\hat N} -1 \right)\left( 1 - \frac{\hat N}{N}\right)
\end{multline}
The third term in \eqref{app_eq:pi_expansion} is
$O_p\left(\frac{1}{\sqrt{n}}\max \left\{\frac{1}{\sqrt{n}},
\frac{1}{\sqrt{n^v}}\right\}\right)$ thanks to Proposition
\ref{app_prop:asymptotic_N}, the weak law of law of large numbers, and Slutsky.
The last term is $O_p\left(\max \left\{\frac{1}{n},
\frac{1}{n^v}\right\}\right)$ by Proposition \ref{app_prop:asymptotic_N}.
Thus, $\sqrt{n}(\hat{\pi}^* - \pi^*)$ is equal to
\begin{gather}\label{app_eq:pi_exp}
     \sqrt{n}\left( \frac{\sum_{i=1}^N R_i}{N} - \pi^*\right) + \sqrt{n}\pi^* \left(1 -  \frac{\hat N}{N} \right) 
    + O_p\left(\max \left\{\frac{1}{\sqrt{n}},\frac{\sqrt{n}}{n^v}\right\}\right).
\end{gather}
We can plug in the expansion of $\sqrt{n}(1 - \hat N / N)$ in
\eqref{app_eq:taylor_hatn} to rewrite \eqref{app_eq:pi_exp} as
\begin{multline}
    \frac{\sqrt{n}}{N} \sum_{i=1}^N R_i\left(1 - \frac{\pi^*}{\pi(Z_i;\gamma_0)}\right) + \pi^* \sqrt{n^v}(\hat\gamma - \gamma_0)\sqrt{\kappa}\frac{1}{N}\sum_{i=1}^N R_i e^{-Z_i^T\gamma_0}Z_i \\
    + O_p\left(\max \left\{\frac{1}{\sqrt{n}},\frac{\sqrt{n}}{n^v}\right\}\right) + \sqrt{\kappa} O_p(\sqrt{n^v}\norm{\hat\gamma-\gamma_0}^2)
\end{multline}
Note
that the first term is a sum of \iid random variables bounded by Assumption
\ref{ass_boundedcov} and thus converges in distribution to $\mathcal{N}(0,
(\pi^*)^2(\mbbE[\pi(Z;\gamma_0)^{-1}]\pi^*-1))$. 
Then asymptotic normality of $\sqrt{n}(1 - \hat N / N)$ follows from analogous
arguments as those in the proof of Proposition \ref{app_prop:asymptotic_N}.
\end{proof}

\begin{proposition}\label{app_prop:asymptotic_q} Consider the conditions of
    Proposition \ref{app_prop:wlog_asymptotics} to be satisfied, and Assumptions
    \ref{ass_correctpi}--\ref{ass_boundedcov} to hold. Then 
    \begin{equation*}
    V_{q^*}^{-1/2}\sqrt{n}(\hat{q}^*-q^*)\overset{d} {\rightarrow}\mathcal{N}(0,1)
    \end{equation*}
    as $N\rightarrow\infty$ where 
    \begin{gather*}
        V_{q^*} =  \pi^* q^* \left[ \pi^*  \mbbE\left[ \frac{q^* - \alpha^*\pi(Z;\gamma_0)}{\pi(Z;\gamma_0)^2} \bigg\vert R=1 \right] 
        + 1 - \alpha^* + \pi^* \kappa W^T \Sigma^v W \right]
    \end{gather*}
    with $W:=\mbbE[e^{-Z^T\gamma_0}\pi(Z;\gamma_0)^{-1}Z|R=1]$, and
    $\alpha^*:=\mbbE[A|R=1]$.
\end{proposition}

\begin{proof}
This proof follows from the same set of arguments as the proof of Propositions
\ref{app_prop:asymptotic_N} and \ref{app_prop:asymptotic_pi}, hence it is
omitted. 
\end{proof}

\begin{proposition}\label{app_prop:q_asymptotics_logreg} %
    Consider the  conditions of Proposition
    \ref{app_prop:wlog_asymptotics} and Assumptions
    \ref{ass_correctpi}--\ref{ass_boundedcov} to hold. Let
    $\theta_0\in\text{Int}(\Theta)$ be defined by the moment condition
    \eqref{eq:q_reg} and $\hat\theta\in\Theta$ be the estimator that solves the
    estimating equation \eqref{eq:theta_solve}. 
    Then 
    \begin{equation*}
       \Sigma^{-1/2} \sqrt{n}(\hat\theta -
    \theta_0)\overset{d} {\rightarrow}\mathcal{N}(0,I_d)
    \end{equation*}
    as $n\rightarrow\infty$ with $\Sigma:=J_\theta^{-1}\Xi J_\theta^{-1}$ and 
    \begin{gather*}
        \Xi := \mbbE\left[ h(A, Z, X;\theta_0, \gamma_0) h(A, Z, X;\theta_0, \gamma_0)^T |R=1\right] 
       +   \kappa J_\gamma \Sigma^v J_\gamma^{T},\\
      J_\theta:= \nabla_\theta G(\theta_0, \gamma_0) = \mbbE[q(X;\theta_0)(1-q(X;\theta_0))\pi(Z;\gamma_0)^{-1}XX^T|R=1],\\
        J_\gamma:= \nabla_\gamma G(\theta_0, \gamma_0) = \mbbE[q(X;\theta_0)e^{-\gamma_0^TZ} XZ^T|R=1].
    \end{gather*}
\end{proposition}
\begin{proof}%
    To show that $\hat\theta$ is consistent for $\theta_0$ and asymptotically normal, we can verify the 
    following three Assumptions from \citet{yuan1998asymptotics}:
    \begin{enumerate}[label=\textbf{C.\arabic*},ref=C.\arabic*]
        \item \label{app_yuan_cond1} $\hat{G}_N(\theta_0,\hat\gamma)\overset{p}{\rightarrow}0$ as $N\rightarrow\infty$.
        \item \label{app_yuan_cond2} There exists a neighborhood $\Theta$ of
        $\theta_0$ on which with probability one all $\hat{G}_N(\theta,
        \hat\gamma)$ are continuously differentiable and $\nabla_\theta
        \hat{G}_N(\theta, \hat\gamma)$ converge uniformly to a nonstochastic
        limit that is nonsigular at $\theta_0$. 
        \item \label{app_yuan_cond3} $\Xi^{-1/2}\sqrt{n}\hat{G}_N(\theta_0,
        \hat\gamma)\overset{d} {\rightarrow}\mathcal{N}(0,I_{d_x})$ as $N\rightarrow\infty$ for
        some matrix $\Xi$. 
    \end{enumerate}
    Since \ref{app_yuan_cond3} implies \ref{app_yuan_cond1}, we 
    only need to prove \ref{app_yuan_cond2} and \ref{app_yuan_cond3}. 

    We first show that \ref{app_yuan_cond3} holds. Consider the following
   Taylor expansion of $\sqrt{n}\hat{G}_N(\theta_0, \hat\gamma)$
    \begin{multline*}%\label{app_eq:c3_te}
       \frac{\sqrt{n}}{N}\sum_{i=1}^N R_i h(A_i, Z_i, X_i;\theta_0,\gamma_0) 
        -   \frac{1}{N}\sum_{i=1}^N R_i q(X_i;\theta_0) e^{-\gamma_0^T Z_i} X_i  Z_i^T\frac{\sqrt n}{\sqrt{n^v}}\sqrt{n^v}(\hat\gamma-\gamma_0) \\
        +   \frac{1}{N}\sum_{i=1}^N R_i q(X_i;\theta_0) e^{-\tilde\gamma^T Z_i} X_i \frac{\sqrt n}{n^v} \left( \sqrt{n^v}(\hat\gamma-\gamma_0)^T Z_i \right)^2.
    \end{multline*}
    where $\tilde\gamma$ is a convex combination of $\hat\gamma$ and $\gamma_0$.
    The first term is a \iid sum whose terms have bounded moments by
    \ref{ass_boundedcov}. Thus, it is asymptotically normal by the central limit
    theorem. 
    The term $N^{-1}\sum_{i=1}^N R_i  q(X_i;\theta_0) e^{-\gamma_0^T Z_i} X_i
    Z_i^T $ is an \iid average formed by terms that have finite moments by
    \ref{ass_boundedcov}, so it converges in probability to $J_\gamma:=\mbbE[R
    q(X;\theta_0) e^{-\gamma_0^T Z} X Z^T]$ by the weak law of large numbers.
    Thus, the second term is asymptotically normal by Proposition
    \ref{app_prop:wlog_asymptotics} and Slutsky. Using similar arguments as in
    the proof of Proposition \ref{app_prop:asymptotic_N}, we can show that the
    third term is $o_p(1)$. The first two terms are asymptotically independent
    because they are arise from separate samples, hence it follows that
    \ref{app_yuan_cond3} is verified. 

    To show that \ref{app_yuan_cond2} holds, it suffices to show that for any
    random vector $\theta_N\in\Theta$ converging in probability to
    $\theta_0$, $\hat{\dot{G}}_N(\theta_N, \hat\gamma):=\nabla_\theta \hat{G}_N(\theta, \hat
    \gamma)\overset{p}{\rightarrow}J_\theta$ for some nonstochastic function
    $J_\theta:=J(\theta_0)$ (Theorem 1 in \citet{iseki1957theorem}). Consider the following Taylor expansion of
    $\hat{\dot{G}}_N(\theta_N, \hat\gamma)$
    \begin{equation}\label{app_eq:yuan_cond2exp}
         \frac{1}{N}\sum_{i=1}^N R_i  q(X_i;\theta_N)(q(X_i;\theta_N) - 1)(1 + e^{-Z_i^T\gamma_0}) X_i X_i^T  
        +  (\hat\gamma - \gamma_0)^T \nabla_\gamma \hat{\dot{G}}_N(\theta_N, \tilde\gamma)
    \end{equation}
    where $\tilde\gamma$ is a convex combination of $\hat\gamma$ and $\gamma_0$.
    The first term in \eqref{app_eq:yuan_cond2exp} can be
    rewritten as
    \begin{multline*}
        \frac{1}{N}\sum_{i=1}^N R_i  q(X_i;\theta_0)(q(X_i;\theta_0) - 1)(1 + e^{-Z_i^T\gamma_0}) X_i X_i^T \\
       +  \frac{1}{N}\sum_{i=1}^N R_i  \left[ q(X_i;\theta_N)(q(X_i;\theta_N) - 1) - q(X_i;\theta_0)(q(X_i;\theta_0) -1 )\right](1 + e^{-Z_i^T\gamma_0}) X_i X_i^T
    \end{multline*}
    where the first term converges to
    $J_\theta:=\mbbE[RXX^Te^{-X^T\theta_0}/\pi(Z;\gamma)]$ by the weak law of
    large numbers. The second term is $o_p(1)$ by Cauchy-Schwarz and
    \ref{ass_boundedcov}; the upper bound can be derived using a similar
    strategy as in expression \eqref{app_eq:rubin_cond2_decomposition_te} of the
    proof of Proposition \ref{app_prop:wlog_asymptotics}.  
    For the second term in \eqref{app_eq:yuan_cond2exp}, we have that
    \begin{multline}\label{app_eq:yuan_cond2exp_second}
        (\hat\gamma - \gamma_0)^T \nabla_\gamma \hat{\dot{G}}_N(\theta_N, \tilde\gamma)\\= (\hat\gamma - \gamma_0)^T \frac{1}{N}\sum_{i=1}^N R_i Z_i q(X_i;\theta_N)(1 - q(X_i;\theta_N))e^{-Z_i^T\tilde\gamma} X_i X_i^T 
    \end{multline}
    where each element of the $d_x \times d_x$ matrix can be upper bounded by
    $$\norm{\hat\gamma -
    \gamma_0}M^2e^{\sqrt{d}M\sup_{\gamma\in\Gamma}\norm{\gamma}}\sqrt{d_z}.$$
    Together with Proposition \ref{app_prop:wlog_asymptotics}, this implies each
    of the elements in \eqref{app_eq:yuan_cond2exp_second} is $o_p(1)$. 
   It follows that \ref{app_yuan_cond2} holds true.

    Under Assumptions \ref{app_yuan_cond1}, \ref{app_yuan_cond2}, and \ref{app_yuan_cond3},
   the result follows by an application of Theorem 4 in
   \citet{yuan1998asymptotics} and our Lemma \ref{app_lemma:change_measure}. 
\end{proof}

Finally, we turn to the result on generalized estimation equations (GEEs) given by Proposition \ref{prop:gee_convergence}. In the proof, we will use $\alpha_0$, which is such that for each $1\leq i\leq N$, $K\geq 2$, and $1\leq
k<j\leq K_i$, %
\begin{equation*}
    \alpha_0:=\frac{\mbbE[\bold{A}_{ik} \bold{A}_{ij}|\bold{X}_i] - \mbbE[\bold{A}_{ij}|\bold{X}_i]\mbbE[\bold{A}_{ik}|\bold{X}_i]}{\sqrt{\var{\bold{A}_{ik}|\bold{X}_i}}\sqrt{\var{\bold{A}_{ij}|\bold{X}_i}}}
\end{equation*}
where $\mbbE[\bold{A}_{ij}|\bold{X}_i]=q(\bold{X}_i;\theta_0)$ and
$\var{\bold{A}_{ik}|\bold{X}_i}=q(\bold{X}_i;\theta_0)(1-q(\bold{X}_i;\theta_0))$
by \ref{ass:gee_mean_correct}.

\begin{proposition}\label{prop:gee_convergence} Assume that the 
    conditions of Proposition \ref{app_prop:wlog_asymptotics} and
    \ref{ass_correctpi}--\ref{ass:gee_mean_correct} hold. 
    Assume that the entries of $W(\bold{X}, \theta, \alpha)^{-1}$ and their
    derivatives are continuous.
    Let $\hat\theta$ be the estimate of $\theta$ obtained by solving the estimating equation 
    \begin{equation*}
        \hat{G}_{N}(\theta, \hat\alpha,\hat\gamma) := \frac{1}{N}\sum_{i=1}^{N} R_i \bold{X}_i D_i(\theta) W_i(\theta, \hat\alpha)^{-1} 
    \left(\bold{A}_i - \frac{\bold{q}_i(\theta)}{\pi(Z_i;\hat \gamma)} \right) = 0.
    \end{equation*}
    Let $\hat\alpha$ be an estimator of $\alpha_0$ such that
    $\hat\alpha-\alpha_0 = O_p(1/\sqrt{N})$. Then 
    \begin{equation*}
       \Sigma^{-1/2}\sqrt{n}(\hat\theta - \theta_0)\overset{d} {\rightarrow}\mathcal{N}(0,I_{d_x})
    \end{equation*}
    as $N\rightarrow\infty$. $\Sigma:=J_\theta^{-1}(\Xi)J_\theta^{-1}$ with 
    \begin{gather*}
        \Xi:=\mbbE[h(\bold{X}, Z, \bold{A};\theta_0, \alpha_0)h(\bold{X}, Z, \bold{A};\theta_0, \alpha_0)^T|R=1] + \kappa  J_\gamma \Sigma^v J^T_\gamma
        \\
        J_\theta := \mbbE\left[ \bold{X} D(\bold{X};\theta_0)W(\bold{X};\theta_0, \alpha_0)^{-1}\nabla_\theta \bold{q}(\theta_0) \pi(Z;\gamma_0)^{-1}|R=1\right]\\
        J_\gamma := \mbbE\left[ \bold{X} D(\bold{X};\theta_0)W(\bold{X};\theta_0,
         \alpha_0)^{-1}\bold{q}(\theta_0) Z e^{-Z^T\gamma_0}|R=1\right]
    \end{gather*}
    where $h(\bold{X}, Z, \bold{A};\theta_0, \alpha_0):=\bold{X}
    D(\bold{X};\theta_0)W(\bold{X};\theta_0, \alpha_0)^{-1}(\bold{A}-
    \bold{q}(\theta_0)\pi(Z;\gamma_0)^{-1})$,
    $\bold{q}(\theta_0):=(q(X_1;\theta_0), \dots, q(X_K;\theta_0))^T$ with $X_k$
    for $1\leq k\leq K$ being the $k^{th}$ column of $\bold{X}$. 
 \end{proposition}

\begin{proof}

    To prove consistency of $\hat\theta$ for $\theta_0$ and its asymptotic
    normality, we will use the results of \citet{yuan1998asymptotics} which rely on
    the following three conditions.
    \begin{enumerate}[label=\textbf{C.\arabic*},ref=C.\arabic*]
        \item \label{cond:yuan_cond1_gee} $ \hat{G}_N(\theta_0,
        \hat\alpha,\hat\gamma)\overset{p}{\rightarrow}0$ as
        $N\rightarrow\infty$.
        \item \label{cond:yuan_cond2_gee} There exists a neighborhood $\Theta$
        of $\theta_0$ on which with probability one $\nabla_\theta \hat{G}_{N}(\theta,
        \hat\alpha,\hat\gamma)$ is continuously differentiable and its
        derivatives converge uniformly to a nonstochastic limit that is
        nonsigular at $\theta_0$. 
        \item \label{cond:yuan_cond3_gee}
        $\Xi^{-1/2}\sqrt{n}\hat{G}_N(\theta_0,\hat\alpha,\hat\gamma)\overset{d} {\rightarrow}\mathcal{N}(0,I_{d_x})$
        for some positive definite matrix $\Xi$.
    \end{enumerate}

    To show that \ref{cond:yuan_cond2_gee} holds, it suffices to show that for
    any random vector $\theta_N\in\Theta$ converging in probability to
    $\theta_0$, $\hat{\dot{G}}_N(\theta, \hat\alpha, \hat\gamma):= \hat{\dot{G}}_N(\theta_N, \hat\alpha,
    \hat\gamma)\overset{p}{\rightarrow}J(\theta_0)$ as $N\rightarrow\infty$ for
    some nonstochastic limit $J_\theta:=J(\theta_0)$. First note that $\nabla_\theta
    \hat{G}_N(\theta, \alpha, \gamma)$ represents the mean of $N$ \iid
    observations. For $i=1,\dots,N$, the $i^{th}$ observation is finite by
    \ref{ass_boundedcov} and the fact that $\max_{i,j=1}^{N_i}|W_i(\theta,
    \alpha)^{-1}|$ is bounded. To simplify the presentation, let us rewrite
     $\hat{\dot{G}}_N(\theta, \alpha, \gamma)=\hat{\dot{G}}^D_N(\theta, \alpha, \gamma) + \hat{\dot{G}}^W_N(\theta, \alpha, \gamma) + \hat{\dot{G}}^q_N(\theta, \alpha, \gamma)$ where
        \begin{gather*}
            \hat{\dot{G}}^D_N(\theta, \alpha, \gamma) := \frac{1}{N}\sum_{i=1}^N \bold{X}_i [\nabla_\theta D_i(\theta)] W_i(\theta, \alpha)^{-1} \left(\bold{A}_i - \bold{q}_i(\theta)\frac{1}{\pi(Z_i;\gamma)}\right)\label{eq:gee_derd} \\
            \hat{\dot{G}}^W_N(\theta, \alpha, \gamma) := \frac{1}{N}\sum_{i=1}^N \bold{X}_i D_i(\theta) [\nabla_\theta W_i(\theta, \alpha)^{-1}] \left(\bold{A}_i - \bold{q}_i(\theta)\frac{1}{\pi(Z_i;\gamma)}\right)\label{eq:gee_derw} \\
            \hat{\dot{G}}^q_N(\theta, \alpha, \gamma) := - \frac{1}{N}\sum_{i=1}^N R_i\bold{X}_i D_i(\theta) W_i(\theta, \alpha)^{-1} \frac{1}{\pi(Z_i;\gamma)}\left[\nabla_\theta \bold{q}_i(\theta)\right]\label{eq:gee_dera} 
        \end{gather*}
    Consider the following Taylor expansion of $\hat{\dot{G}}_N(\theta_N,
    \hat\alpha, \hat\gamma)$:
    \begin{multline}\label{eq:gee_cond2_te}
         \hat{\dot{G}}_N(\theta_N, \alpha_0, \gamma_0)
        + (\hat\alpha - \alpha_0)\nabla_\alpha \hat{\dot{G}}_N(\theta_N, \tilde\alpha, \gamma_0) \\
        + (\hat\gamma - \gamma_0)^T \nabla_\gamma \hat{\dot{G}}_N(\theta_N, \alpha_0, \tilde\gamma) 
        + (\hat\gamma - \gamma_0)^T \nabla_\alpha \nabla_\gamma \hat{\dot{G}}_N(\theta_N, \tilde\alpha, \tilde\gamma) (\hat\alpha - \alpha_0)
    \end{multline}
    where $\tilde\gamma$ is a convex combination of $\hat\gamma$ and $\gamma_0$,
    while $\tilde\alpha$ is a convex combination of $\hat\alpha$ and $\alpha_0$.

    The first term in \eqref{eq:gee_cond2_te} can be rewritten as 
    \begin{equation}\label{eq:gee_cond2_te_firstterm}
        \hat{\dot{G}}_N(\theta_0, \alpha_0, \gamma_0) 
        +  \hat{\dot{G}}_N(\theta_N, \alpha_0, \gamma_0)  -  \hat{\dot{G}}_N(\theta_0, \alpha_0, \gamma_0). 
    \end{equation}

    The first term in \eqref{eq:gee_cond2_te_firstterm} is an average of $N$
    terms that are \iid and finite by \ref{ass_correctpi},
    \ref{ass_boundedcov}, and the boundedness of $W_i(\theta, \alpha)^{-1}$.
    Since the model for the mean is corectly specified by \ref{ass:gee_mean_correct}, the first term converges in 
    probability $-\mbbE[R D(\bold{X};\theta_0)W(\bold{X};\theta_0, \alpha_0)^{-1}\nabla_\theta
    \bold{q}(\theta_0)\pi(Z;\gamma_0)^{-1}]$ by the weak law of large 
    numbers and iterated expectations. 

    We can then show that the difference between the remaining two terms in
    \eqref{eq:gee_cond2_te_firstterm} converges to $0$ in probability. Let
    $\bar{w}:=\max_{\theta\in\theta, \alpha\in[-1,1]}\max_{i,j,k} |(W_i(\theta,
    \alpha)^{-1})_{jk}|$. We have that $\hat{\dot{G}}^D_N(\theta_N, \alpha_0,
    \gamma_0)$ is equal to
    \begin{multline*}%\label{eq:gee_convergence_te_secondterm_diff1}
     \frac{1}{N}\sum_{i=1}^N\bold{X}_i \nabla_\theta D_i(\theta_N)
    W_i(\theta_N, \alpha_0)^{-1}\left(A_i - \frac{q_i(\theta_N)}{\pi(Z_i;\gamma_0)}\right)
    \\ < \frac{1}{N} \sum_{i=1}^N M^2\bar{w} u_{d_x}u_{K_i}^T I_{K_i} u_{K_i}u_{K_i}^T \left[\left(A_i -  \frac{q_i(\theta_0)}{\pi(Z_i;\gamma_0)}\right) + 
    \frac{1}{\pi(Z_i;\gamma_0)}\left(q_i(\theta_0) - q_i(\theta_N)\right)\right] u_{d_x}^T
    \end{multline*}
    where $u_n:=(1, \dots, 1)^T$ has length $n\in\mathbb{Z}^{+}$. The inequality
    follows from \ref{ass_boundedcov} and the boundedness of $W(\theta_N,
    \alpha)^{-1}$. Then $N^{-1}\sum_{i=1}^N (A_i -
    q_i(\theta_0)/\pi(Z_i;\theta_0))$ is an average of \iid terms that are
    finite and thus converges in proability to $0$ by
    the weak law of large numbers and \ref{ass:gee_mean_correct}. By
    \ref{ass_correctpi}, \ref{ass_boundedcov}, and Cauchy-Schwarz, we have the following upper bound
    \begin{equation*}
        \frac{1}{N}\sum_{i=1}^N \frac{1}{\pi(Z_i;\gamma_0)} \left( q_i(\theta_0 - q_i(\theta_N) )  \right) < \frac{1}{\epsilon} e^{\norm{\theta_0}M\sqrt{d}}\left( e^{\norm{\theta_N - \theta_0}M\sqrt{d}} -1   \right) 
    \end{equation*}
    where the RHS converges to $0$ in probability by the weak law of large
    numbers and the continuous mapping theorem. It follows that
    $\hat{\dot{G}}^D_N(\theta_N, \alpha_0, \gamma_0)$ converges to $0$ in
    probability. Using similar arguments, we can show that
    $\hat{\dot{G}}^W_N(\theta_N, \alpha_0, \gamma_0)$ converges in probability
    to $0$. Clearly, $\hat{\dot{G}}^D_N(\theta_0, \alpha_0, \gamma_0)$ and
    $\hat{\dot{G}}^W_N(\theta_0, \alpha_0, \gamma_0)$ converge to $0$ in
    probability by the weak law of large numbers. Next, we need to show that $
    \hat{\dot{G}}^q_N(\theta_N, \alpha, \gamma)- \hat{\dot{G}}^q_N(\theta_0,
    \alpha, \gamma)$ converges to $0$ in probability. This difference can be rewritten as 
    \begin{multline*}
        \frac{1}{N}\sum_{i=1}^N \frac{R_i}{\pi(Z_i;\gamma_0)}\bold{X}_i \left[ D_i(\theta_0)W_i(\theta_0, \alpha_0)^{-1}\nabla \bold{q}(\theta_0) - D_i(\theta_N)W_(\theta_N, \alpha_0)^{-1}\nabla \bold{q}(\theta_N)\right].
        \end{multline*}
    By \ref{ass_boundedcov}, the convergence boils down
    to showing that for all $i=1,\dots, N$ and $1\leq k,j\leq K_i$, 
    \begin{multline*}%\label{eq:gee_convergence_te_secondterm_diff2}
       (W_i(\theta_0, \alpha_0)^{-1})_{jk} q(X_{ik};\theta_0)(1 - q(X_{ik};\theta_0)) q(X_{ij};\theta_0)(1 - q(X_{ij};\theta_0))\\
        -  (W_i(\theta_N, \alpha_0)^{-1})_{jk} q(X_{ik};\theta_N)(1 - q(X_{ik};\theta_N)) q(X_{ij};\theta_N)(1 - q(X_{ij};\theta_N))
    \end{multline*}
    converges in probability to $0$. This can be upper bounded by 
    \begin{multline}\label{eq:gee_convergence_te_secondterm_diff3}
       e^{\theta_0^T(X_{ik} + X_{ij})}(1 + e^{\theta_0^T X_{ik}})(1 + e^{\theta_0^T X_{ij}}) (W_i(\theta_0, \alpha_0)^{-1})_{kj}\\
        \bigg\vert \frac{e^{(\theta_N^T(X_{ik} + X_{ij})}}{e^{(\theta_0^T(X_{ik} + X_{ij})}}
         \frac{1+e^{\theta_N^T X_{ik}}}{1+e^{\theta_0^T X_{ik}}} \frac{1+e^{\theta_N^T X_{ij}}}{1+e^{\theta_0^T X_{ij}}} \frac{(W_i(\theta_N, \alpha_0)^{-1})_{kj}}{(W_i(\theta_0, \alpha_0)^{-1})_{kj}}-1 \bigg\vert
         \\
         <  e^{2\norm{\theta_0}\sqrt{d}M}(1 + e^{\norm{\theta_0}\sqrt{d}M})^2
         \left\vert e^{4\norm{\theta_N - \theta_0}\sqrt{d}M} (1 + o_p(1)) - 1  \right\vert
    \end{multline}
    where the inequality follows from Cauchy-Schwarz, the fact that
    $\norm{X}\leq \sqrt{d}\norm{X}_\infty < M$, and the continuous mapping
    theorem thanks to the fact that $W_i$ has continuous derivatives and
    $\theta_N\rightarrow\theta_0$ as $N\rightarrow\infty$. Since
    $\theta_N\overset{p}{\rightarrow}\theta_0$ as $N\rightarrow\infty$, the RHS
    in \eqref{eq:gee_convergence_te_secondterm_diff3} is $o_p(1)$. Thus, we can
    conclude that the difference between the second and third terms in
    \eqref{eq:gee_cond2_te_firstterm} converges in probability to $0$. 
    
    Next, we turn again to the Taylor expansion of $\hat{\dot{G}}_N(\theta_N,
    \hat\alpha, \hat\gamma)$ in \eqref{eq:gee_cond2_te}. Note that
    $\nabla_\alpha \hat{\dot{G}}_N (\theta_N, \tilde\alpha, \gamma_0)$ is
    bounded and $\hat\alpha - \alpha_0 = o_p(1)$, so their product converges to
    $0$ in probability by Slutsky. 
    The third term $\nabla_\gamma \hat{\dot{G}}_N(\theta_N, \alpha_0,
    \tilde\gamma)$ is also bounded and, by Proposition
    \ref{app_prop:wlog_asymptotics}, $\hat\gamma - \gamma_0
    \overset{p}{\rightarrow}0$ as $N\rightarrow\infty$, thus we can appply
    Cauchy-Schwarz to show that the product converges to $0$ in probability. The
    fourth term in the RHS of \eqref{eq:gee_cond2_te} can be shown to converge
    in probability to $0$ using analogous arguments. It follows that
    \ref{cond:yuan_cond2_gee} holds. 

       Since \ref{cond:yuan_cond3_gee} implies \ref{cond:yuan_cond1_gee}, we
    only need to show that \ref{cond:yuan_cond3_gee} holds.  To show that \ref{cond:yuan_cond3_gee} holds, consider the following
    Taylor expansion of $\sqrt{n}\hat{G}(\theta_0, \hat\alpha, \hat\gamma)$:
    \begin{multline}\label{eq:gee_cond3_taylor}
        \sqrt{n}\bigg[ \hat{G}_N(\theta_0, \alpha_0, \gamma_0) 
        + (\hat\alpha - \alpha_0)\nabla_\alpha \hat{G}_N(\theta_0, \alpha_0, \gamma_0)
        + (\hat\gamma - \gamma_0)^T\nabla_\gamma \hat{G}_N(\theta_0, \alpha_0, \gamma_0)\\
        + (\hat\alpha - \alpha_0)^2\nabla^2_\alpha \hat{G}_N(\theta_0, \tilde\alpha, \gamma_0)
        + (\hat\alpha - \alpha_0)(\hat\gamma - \gamma_0)^T\\\nabla_\alpha \nabla_\gamma \hat{G}_N(\theta_0, \alpha_0, \gamma_0)
        + (\hat\gamma - \gamma_0)^T \nabla^2_\gamma \hat{G}_N(\theta_0, \alpha, \tilde\gamma) (\hat\gamma - \gamma_0) \\
        + (\hat\alpha - \alpha_0)^2(\hat\gamma - \gamma_0)^T\nabla^2_\alpha \nabla_\gamma \hat{G}_N(\theta_0, \tilde\alpha, \gamma_0)
        + (\hat\alpha - \alpha_0)(\hat\gamma - \gamma_0)\nabla_\alpha \nabla^2_\gamma \hat{G}_N(\theta_0, \alpha_0, \tilde\gamma)(\hat\gamma - \gamma_0)\\
        + (\hat\alpha - \alpha_0)^2 (\hat\gamma - \gamma_0)^T \nabla^2_\alpha \nabla^2_\gamma \hat{G}_N(\theta_0, \alpha, \tilde\gamma) (\hat\gamma - \gamma_0)
        \bigg].
    \end{multline}

    The first term in \eqref{eq:gee_cond3_taylor} is composed by $N$ \iid bounded random variables and thus,
    by the central limit theorem,
    \begin{equation*}
        \sqrt{n}\hat{G}_N(\theta_0, \alpha_0, \gamma_0) \overset{d}{\rightarrow} \mathcal{N}(0,\Xi)
    \end{equation*}
    as $N\rightarrow\infty$ where $\Xi:= \text{Var}(G(\theta_0, \alpha_0, \gamma_0)$. 
    
    For the second term in \eqref{eq:gee_cond3_taylor}, $\sqrt{n}(\hat\alpha -
    \alpha_0)=O_p(1)$ by Assumption, while $\nabla_\alpha \hat{G}_N(\theta_0,
    \alpha_0, \gamma_0)$ converges in probability to $0$ by the weak law of
    large numbers and \ref{ass:gee_mean_correct}. By the continuous mapping
    theorem, their product then converges in probability to $0$. 
    
   For the third term,
    \begin{equation*}
       \sqrt{n}(\hat\gamma - \gamma_0)^T\nabla_\gamma \hat{G}_N(\theta_0, \alpha_0, \gamma_0) =  \sqrt{n^v}(\hat\gamma - \gamma_0)^T\sqrt{\kappa} \nabla_\gamma \hat{G}_N(\theta_0, \alpha_0, \gamma_0)
    \end{equation*}
    where $\nabla_\gamma \hat{G}_N(\theta_0, \alpha_0,
    \gamma_0)\overset{p}{\rightarrow}J_\gamma:=\mbbE[\nabla_\gamma G(\theta_0,
    \alpha_0, \gamma_0)]$ as $n^v\rightarrow\infty$ by the weak law of large
    numbers. In addition, $\sqrt{n^v}(\hat\gamma -
    \gamma_0)\overset{d} {\rightarrow}\mathcal{N}(0,\Sigma^v)$ by propositon
    \ref{app_prop:wlog_asymptotics}. Thus, the term converges in distribution by
    Slutsky to $\mathcal{N}(0,J_\gamma\Sigma^vJ_\gamma)$. 

    For the fourth term, 
    $\nabla^2_\alpha \hat{G}_N(\theta_0, \tilde\alpha, \gamma_0)$ is bounded and
    $\sqrt{n}(\hat\alpha - \alpha_0)=O_p(1)$, so their product converges in
    probability to $0$ as $N\rightarrow\infty$. Using similar arguments, it is
    easy to see that all remaining terms converge to $0$ in probability as well.
    Thus \ref{cond:yuan_cond3_gee} is satisied. 
    
    The result of the Proposition follows from Theorem 4 of
    \citet{yuan1998asymptotics}. 

\end{proof}

\end{appendix}

% \begin{supplement}
% \stitle{Code}
% \sdescription{This supplement contains the code for replicating the results in the paper. The code is also available at \href{https://github.com/ricfog/arrests-with-unreported-crimes}{github.com/ricfog/arrests-with-unreported-crimes}.}
% \end{supplement}

\bibliographystyle{imsart-nameyear}
\bibliography{references}

\begin{thebibliography}{101}
% BibTex style file: imsart-nameyear.bst, 2017-11-03
% Default style options (sort=1,type=nameyear).
% Used options (sort=1,type=nameyear).

\bibitem[\protect\citeauthoryear{Andrews and Monahan}{1992}]{andrews1992improved}
\begin{barticle}[author]
\bauthor{\bsnm{Andrews},~\bfnm{Donald~WK}\binits{D.~W.}} \AND \bauthor{\bsnm{Monahan},~\bfnm{J~Christopher}\binits{J.~C.}}
(\byear{1992}).
\btitle{An improved heteroskedasticity and autocorrelation consistent covariance matrix estimator}.
\bjournal{Econometrica: Journal of the Econometric Society}
\bpages{953--966}.
\end{barticle}
\endbibitem

\bibitem[\protect\citeauthoryear{Avakame, Fyfe and McCoy}{1999}]{avakame1999did}
\begin{barticle}[author]
\bauthor{\bsnm{Avakame},~\bfnm{Edem~F}\binits{E.~F.}}, \bauthor{\bsnm{Fyfe},~\bfnm{James~J}\binits{J.~J.}} \AND \bauthor{\bsnm{McCoy},~\bfnm{Candace}\binits{C.}}
(\byear{1999}).
\btitle{“Did you call the police? What did they do?” An empirical assessment of Black's theory of mobilization of law}.
\bjournal{Justice Quarterly}
\bvolume{16}
\bpages{765--792}.
\end{barticle}
\endbibitem

\bibitem[\protect\citeauthoryear{Azur et~al.}{2011}]{azur2011multiple}
\begin{barticle}[author]
\bauthor{\bsnm{Azur},~\bfnm{Melissa~J}\binits{M.~J.}}, \bauthor{\bsnm{Stuart},~\bfnm{Elizabeth~A}\binits{E.~A.}}, \bauthor{\bsnm{Frangakis},~\bfnm{Constantine}\binits{C.}} \AND \bauthor{\bsnm{Leaf},~\bfnm{Philip~J}\binits{P.~J.}}
(\byear{2011}).
\btitle{Multiple imputation by chained equations: what is it and how does it work?}
\bjournal{International journal of methods in psychiatric research}
\bvolume{20}
\bpages{40--49}.
\end{barticle}
\endbibitem

\bibitem[\protect\citeauthoryear{Bachman}{1998}]{bachman1998factors}
\begin{barticle}[author]
\bauthor{\bsnm{Bachman},~\bfnm{Ronet}\binits{R.}}
(\byear{1998}).
\btitle{The factors related to rape reporting behavior and arrest: New evidence from the National Crime Victimization Survey}.
\bjournal{Criminal justice and behavior}
\bvolume{25}
\bpages{8--29}.
\end{barticle}
\endbibitem

\bibitem[\protect\citeauthoryear{Barnett-Ryan, Langton and Planty}{2014}]{barnett2014nation}
\begin{barticle}[author]
\bauthor{\bsnm{Barnett-Ryan},~\bfnm{Cindy}\binits{C.}}, \bauthor{\bsnm{Langton},~\bfnm{Lynn}\binits{L.}} \AND \bauthor{\bsnm{Planty},~\bfnm{Michael}\binits{M.}}
(\byear{2014}).
\btitle{The nation’s two crime measures, 2014}.
\bjournal{US Department of Justice, Washington, DC}.
\end{barticle}
\endbibitem

\bibitem[\protect\citeauthoryear{Basu}{2011}]{basu2011essay}
\begin{bincollection}[author]
\bauthor{\bsnm{Basu},~\bfnm{Debabrata}\binits{D.}}
(\byear{2011}).
\btitle{An essay on the logical foundations of survey sampling, part one}.
In \bbooktitle{Selected Works of Debabrata Basu}
\bpages{167--206}.
\bpublisher{Springer}.
\end{bincollection}
\endbibitem

\bibitem[\protect\citeauthoryear{Baumer}{2002}]{baumer2002neighborhood}
\begin{barticle}[author]
\bauthor{\bsnm{Baumer},~\bfnm{Eric~P}\binits{E.~P.}}
(\byear{2002}).
\btitle{Neighborhood disadvantage and police notification by victims of violence}.
\bjournal{Criminology}
\bvolume{40}
\bpages{579--616}.
\end{barticle}
\endbibitem

\bibitem[\protect\citeauthoryear{Baumer and Lauritsen}{2010}]{baumer2010reporting}
\begin{barticle}[author]
\bauthor{\bsnm{Baumer},~\bfnm{Eric~P}\binits{E.~P.}} \AND \bauthor{\bsnm{Lauritsen},~\bfnm{Janet~L}\binits{J.~L.}}
(\byear{2010}).
\btitle{Reporting crime to the police, 1973--2005: A multivariate analysis of long-term trends in the National Crime Survey (NCS) and National Crime Victimization Survey (NCVS)}.
\bjournal{Criminology}
\bvolume{48}
\bpages{131--185}.
\end{barticle}
\endbibitem

\bibitem[\protect\citeauthoryear{Beck and Blumstein}{2018}]{beck2018racial}
\begin{barticle}[author]
\bauthor{\bsnm{Beck},~\bfnm{Allen~J}\binits{A.~J.}} \AND \bauthor{\bsnm{Blumstein},~\bfnm{Alfred}\binits{A.}}
(\byear{2018}).
\btitle{Racial disproportionality in US state prisons: Accounting for the effects of racial and ethnic differences in criminal involvement, arrests, sentencing, and time served}.
\bjournal{Journal of Quantitative Criminology}
\bvolume{34}
\bpages{853--883}.
\end{barticle}
\endbibitem

\bibitem[\protect\citeauthoryear{Berk et~al.}{2019}]{berk2019assumption}
\begin{barticle}[author]
\bauthor{\bsnm{Berk},~\bfnm{Richard}\binits{R.}}, \bauthor{\bsnm{Buja},~\bfnm{Andreas}\binits{A.}}, \bauthor{\bsnm{Brown},~\bfnm{Lawrence}\binits{L.}}, \bauthor{\bsnm{George},~\bfnm{Edward}\binits{E.}}, \bauthor{\bsnm{Kuchibhotla},~\bfnm{Arun~Kumar}\binits{A.~K.}}, \bauthor{\bsnm{Su},~\bfnm{Weijie}\binits{W.}} \AND \bauthor{\bsnm{Zhao},~\bfnm{Linda}\binits{L.}}
(\byear{2019}).
\btitle{Assumption lean regression}.
\bjournal{The American Statistician}.
\end{barticle}
\endbibitem

\bibitem[\protect\citeauthoryear{Binder}{1983}]{binder1983variances}
\begin{barticle}[author]
\bauthor{\bsnm{Binder},~\bfnm{David~A}\binits{D.~A.}}
(\byear{1983}).
\btitle{On the variances of asymptotically normal estimators from complex surveys}.
\bjournal{International Statistical Review/Revue Internationale de Statistique}
\bpages{279--292}.
\end{barticle}
\endbibitem

\bibitem[\protect\citeauthoryear{Blumstein and Cohen}{1979}]{blumstein1979estimation}
\begin{barticle}[author]
\bauthor{\bsnm{Blumstein},~\bfnm{Alfred}\binits{A.}} \AND \bauthor{\bsnm{Cohen},~\bfnm{Jacqueline}\binits{J.}}
(\byear{1979}).
\btitle{Estimation of individual crime rates from arrest records}.
\bjournal{J. Crim. L. \& Criminology}
\bvolume{70}
\bpages{561}.
\end{barticle}
\endbibitem

\bibitem[\protect\citeauthoryear{Blumstein and Cohen}{1987}]{blumstein1987characterizing}
\begin{barticle}[author]
\bauthor{\bsnm{Blumstein},~\bfnm{Alfred}\binits{A.}} \AND \bauthor{\bsnm{Cohen},~\bfnm{Jacqueline}\binits{J.}}
(\byear{1987}).
\btitle{Characterizing criminal careers}.
\bjournal{Science}
\bvolume{237}
\bpages{985--991}.
\end{barticle}
\endbibitem

\bibitem[\protect\citeauthoryear{Blumstein et~al.}{1986}]{blumstein1986criminal}
\begin{bbook}[author]
\bauthor{\bsnm{Blumstein},~\bfnm{Alfred}\binits{A.}} \betal{et~al.}
(\byear{1986}).
\btitle{Criminal Careers and" Career Criminals,"}
\bvolume{2}.
\bpublisher{National Academies}.
\end{bbook}
\endbibitem

\bibitem[\protect\citeauthoryear{Blumstein et~al.}{2010}]{blumstein2010linking}
\begin{barticle}[author]
\bauthor{\bsnm{Blumstein},~\bfnm{Alfred}\binits{A.}}, \bauthor{\bsnm{Cohen},~\bfnm{Jacqueline}\binits{J.}}, \bauthor{\bsnm{Piquero},~\bfnm{Alex~R}\binits{A.~R.}} \AND \bauthor{\bsnm{Visher},~\bfnm{Christy~A}\binits{C.~A.}}
(\byear{2010}).
\btitle{Linking the crime and arrest processes to measure variations in individual arrest risk per crime (Q)}.
\bjournal{Journal of Quantitative Criminology}
\bvolume{26}
\bpages{533--548}.
\end{barticle}
\endbibitem

\bibitem[\protect\citeauthoryear{B{\"o}hning and Van Der~Heijden}{2009}]{bohning2009covariate}
\begin{barticle}[author]
\bauthor{\bsnm{B{\"o}hning},~\bfnm{Dankmar}\binits{D.}} \AND \bauthor{\bsnm{Van Der~Heijden},~\bfnm{Peter~GM}\binits{P.~G.}}
(\byear{2009}).
\btitle{A covariate adjustment for zero-truncated approaches to estimating the size of hidden and elusive populations}.
\bjournal{The Annals of Applied Statistics}
\bvolume{3}
\bpages{595--610}.
\end{barticle}
\endbibitem

\bibitem[\protect\citeauthoryear{Brame et~al.}{2004}]{brame2004criminal}
\begin{barticle}[author]
\bauthor{\bsnm{Brame},~\bfnm{Robert}\binits{R.}}, \bauthor{\bsnm{Fagan},~\bfnm{Jeffrey}\binits{J.}}, \bauthor{\bsnm{Piquero},~\bfnm{Alex~R}\binits{A.~R.}}, \bauthor{\bsnm{Schubert},~\bfnm{Carol~A}\binits{C.~A.}} \AND \bauthor{\bsnm{Steinberg},~\bfnm{Laurence}\binits{L.}}
(\byear{2004}).
\btitle{Criminal careers of serious delinquents in two cities}.
\bjournal{Youth Violence and Juvenile Justice}
\bvolume{2}
\bpages{256--272}.
\end{barticle}
\endbibitem

\bibitem[\protect\citeauthoryear{Breiman}{2001}]{breiman2001random}
\begin{barticle}[author]
\bauthor{\bsnm{Breiman},~\bfnm{Leo}\binits{L.}}
(\byear{2001}).
\btitle{Random forests}.
\bjournal{Machine learning}
\bvolume{45}
\bpages{5--32}.
\end{barticle}
\endbibitem

\bibitem[\protect\citeauthoryear{Buil-Gil, Medina and Shlomo}{2021}]{buil2021measuring}
\begin{barticle}[author]
\bauthor{\bsnm{Buil-Gil},~\bfnm{David}\binits{D.}}, \bauthor{\bsnm{Medina},~\bfnm{Juanjo}\binits{J.}} \AND \bauthor{\bsnm{Shlomo},~\bfnm{Natalie}\binits{N.}}
(\byear{2021}).
\btitle{Measuring the dark figure of crime in geographic areas: Small area estimation from the crime survey for England and Wales}.
\bjournal{The British journal of criminology}
\bvolume{61}
\bpages{364--388}.
\end{barticle}
\endbibitem

\bibitem[\protect\citeauthoryear{Buil-Gil, Moretti and Langton}{2021}]{buil2021accuracy}
\begin{barticle}[author]
\bauthor{\bsnm{Buil-Gil},~\bfnm{David}\binits{D.}}, \bauthor{\bsnm{Moretti},~\bfnm{Angelo}\binits{A.}} \AND \bauthor{\bsnm{Langton},~\bfnm{Samuel~H}\binits{S.~H.}}
(\byear{2021}).
\btitle{The accuracy of crime statistics: Assessing the impact of police data bias on geographic crime analysis}.
\bjournal{Journal of Experimental Criminology}
\bpages{1--27}.
\end{barticle}
\endbibitem

\bibitem[\protect\citeauthoryear{Buja et~al.}{2019a}]{buja2019models1}
\begin{barticle}[author]
\bauthor{\bsnm{Buja},~\bfnm{Andreas}\binits{A.}}, \bauthor{\bsnm{Brown},~\bfnm{Lawrence}\binits{L.}}, \bauthor{\bsnm{Berk},~\bfnm{Richard}\binits{R.}}, \bauthor{\bsnm{George},~\bfnm{Edward}\binits{E.}}, \bauthor{\bsnm{Pitkin},~\bfnm{Emil}\binits{E.}}, \bauthor{\bsnm{Traskin},~\bfnm{Mikhail}\binits{M.}}, \bauthor{\bsnm{Zhang},~\bfnm{Kai}\binits{K.}} \AND \bauthor{\bsnm{Zhao},~\bfnm{Linda}\binits{L.}}
(\byear{2019}a).
\btitle{Models as approximations I: Consequences illustrated with linear regression}.
\bjournal{Statistical Science}
\bvolume{34}
\bpages{523--544}.
\end{barticle}
\endbibitem

\bibitem[\protect\citeauthoryear{Buja et~al.}{2019b}]{buja2019models}
\begin{barticle}[author]
\bauthor{\bsnm{Buja},~\bfnm{Andreas}\binits{A.}}, \bauthor{\bsnm{Brown},~\bfnm{Lawrence}\binits{L.}}, \bauthor{\bsnm{Kuchibhotla},~\bfnm{Arun~Kumar}\binits{A.~K.}}, \bauthor{\bsnm{Berk},~\bfnm{Richard}\binits{R.}}, \bauthor{\bsnm{George},~\bfnm{Edward}\binits{E.}}, \bauthor{\bsnm{Zhao},~\bfnm{Linda}\binits{L.}} \betal{et~al.}
(\byear{2019}b).
\btitle{Models as Approximations II: A Model-Free Theory of Parametric Regression}.
\bjournal{Statistical Science}
\bvolume{34}
\bpages{545--565}.
\end{barticle}
\endbibitem

\bibitem[\protect\citeauthoryear{Byrd and Lipton}{2019}]{byrd2019effect}
\begin{binproceedings}[author]
\bauthor{\bsnm{Byrd},~\bfnm{Jonathon}\binits{J.}} \AND \bauthor{\bsnm{Lipton},~\bfnm{Zachary}\binits{Z.}}
(\byear{2019}).
\btitle{What is the effect of importance weighting in deep learning?}
In \bbooktitle{International Conference on Machine Learning}
\bpages{872--881}.
\bpublisher{PMLR}.
\end{binproceedings}
\endbibitem

\bibitem[\protect\citeauthoryear{Cernat et~al.}{2021}]{cernat2021estimating}
\begin{barticle}[author]
\bauthor{\bsnm{Cernat},~\bfnm{Alexandru}\binits{A.}}, \bauthor{\bsnm{Buil-Gil},~\bfnm{David}\binits{D.}}, \bauthor{\bsnm{Pina-S{\'a}nchez},~\bfnm{Jose}\binits{J.}}, \bauthor{\bsnm{Murri{\`a}-Sangen{\'\i}s},~\bfnm{Marta}\binits{M.}} \betal{et~al.}
(\byear{2021}).
\btitle{Estimating crime in place: Moving beyond residence location}.
\end{barticle}
\endbibitem

\bibitem[\protect\citeauthoryear{D'Alessio and Stolzenberg}{2003}]{d2003race}
\begin{barticle}[author]
\bauthor{\bsnm{D'Alessio},~\bfnm{Stewart~J}\binits{S.~J.}} \AND \bauthor{\bsnm{Stolzenberg},~\bfnm{Lisa}\binits{L.}}
(\byear{2003}).
\btitle{Race and the probability of arrest}.
\bjournal{Social forces}
\bvolume{81}
\bpages{1381--1397}.
\end{barticle}
\endbibitem

\bibitem[\protect\citeauthoryear{Dugan}{2003}]{dugan2003domestic}
\begin{barticle}[author]
\bauthor{\bsnm{Dugan},~\bfnm{Laura}\binits{L.}}
(\byear{2003}).
\btitle{Domestic violence legislation: Exploring its impact on the likelihood of domestic violence, police involvement, and arrest}.
\bjournal{Criminology \& Public Policy}
\bvolume{2}
\bpages{283--312}.
\end{barticle}
\endbibitem

\bibitem[\protect\citeauthoryear{Fisher et~al.}{2003}]{fisher2003reporting}
\begin{barticle}[author]
\bauthor{\bsnm{Fisher},~\bfnm{Bonnie~S}\binits{B.~S.}}, \bauthor{\bsnm{Daigle},~\bfnm{Leah~E}\binits{L.~E.}}, \bauthor{\bsnm{Cullen},~\bfnm{Francis~T}\binits{F.~T.}} \AND \bauthor{\bsnm{Turner},~\bfnm{Michael~G}\binits{M.~G.}}
(\byear{2003}).
\btitle{Reporting sexual victimization to the police and others: Results from a national-level study of college women}.
\bjournal{Criminal justice and behavior}
\bvolume{30}
\bpages{6--38}.
\end{barticle}
\endbibitem

\bibitem[\protect\citeauthoryear{Fitzmaurice, Laird and Rotnitzky}{1993}]{fitzmaurice1993regression}
\begin{barticle}[author]
\bauthor{\bsnm{Fitzmaurice},~\bfnm{Garrett~M}\binits{G.~M.}}, \bauthor{\bsnm{Laird},~\bfnm{Nan~M}\binits{N.~M.}} \AND \bauthor{\bsnm{Rotnitzky},~\bfnm{Andrea~G}\binits{A.~G.}}
(\byear{1993}).
\btitle{Regression models for discrete longitudinal responses}.
\bjournal{Statistical Science}
\bpages{284--299}.
\end{barticle}
\endbibitem

\bibitem[\protect\citeauthoryear{Fitzmaurice et~al.}{2008}]{fitzmaurice2008longitudinal}
\begin{bbook}[author]
\bauthor{\bsnm{Fitzmaurice},~\bfnm{Garrett}\binits{G.}}, \bauthor{\bsnm{Davidian},~\bfnm{Marie}\binits{M.}}, \bauthor{\bsnm{Verbeke},~\bfnm{Geert}\binits{G.}} \AND \bauthor{\bsnm{Molenberghs},~\bfnm{Geert}\binits{G.}}
(\byear{2008}).
\btitle{Longitudinal data analysis}.
\bpublisher{CRC press}.
\end{bbook}
\endbibitem

\bibitem[\protect\citeauthoryear{Fogliato et~al.}{2021}]{fogliato2021validity}
\begin{binproceedings}[author]
\bauthor{\bsnm{Fogliato},~\bfnm{Riccardo}\binits{R.}}, \bauthor{\bsnm{Xiang},~\bfnm{Alice}\binits{A.}}, \bauthor{\bsnm{Lipton},~\bfnm{Zachary}\binits{Z.}}, \bauthor{\bsnm{Nagin},~\bfnm{Daniel}\binits{D.}} \AND \bauthor{\bsnm{Chouldechova},~\bfnm{Alexandra}\binits{A.}}
(\byear{2021}).
\btitle{On the Validity of Arrest as a Proxy for Offense: Race and the Likelihood of Arrest for Violent Crimes}.
In \bbooktitle{Proceedings of the 2021 AAAI/ACM Conference on AI, Ethics, and Society}.
\bseries{AIES '21}
\bpages{100–111}.
\bpublisher{Association for Computing Machinery}, \baddress{New York, NY, USA}.
\bdoi{10.1145/3461702.3462538}
\end{binproceedings}
\endbibitem

\bibitem[\protect\citeauthoryear{Fuller}{2011}]{fuller2011sampling}
\begin{bbook}[author]
\bauthor{\bsnm{Fuller},~\bfnm{Wayne~A}\binits{W.~A.}}
(\byear{2011}).
\btitle{Sampling statistics}
\bvolume{560}.
\bpublisher{John Wiley \& Sons}.
\end{bbook}
\endbibitem

\bibitem[\protect\citeauthoryear{Graham, Olchowski and Gilreath}{2007}]{graham2007many}
\begin{barticle}[author]
\bauthor{\bsnm{Graham},~\bfnm{John~W}\binits{J.~W.}}, \bauthor{\bsnm{Olchowski},~\bfnm{Allison~E}\binits{A.~E.}} \AND \bauthor{\bsnm{Gilreath},~\bfnm{Tamika~D}\binits{T.~D.}}
(\byear{2007}).
\btitle{How many imputations are really needed? Some practical clarifications of multiple imputation theory}.
\bjournal{Prevention science}
\bvolume{8}
\bpages{206--213}.
\end{barticle}
\endbibitem

\bibitem[\protect\citeauthoryear{Heckman}{1979}]{heckman1979sample}
\begin{barticle}[author]
\bauthor{\bsnm{Heckman},~\bfnm{James~J}\binits{J.~J.}}
(\byear{1979}).
\btitle{Sample selection bias as a specification error}.
\bjournal{Econometrica: Journal of the econometric society}
\bpages{153--161}.
\end{barticle}
\endbibitem

\bibitem[\protect\citeauthoryear{Horvitz and Thompson}{1952}]{horvitz1952generalization}
\begin{barticle}[author]
\bauthor{\bsnm{Horvitz},~\bfnm{Daniel~G}\binits{D.~G.}} \AND \bauthor{\bsnm{Thompson},~\bfnm{Donovan~J}\binits{D.~J.}}
(\byear{1952}).
\btitle{A generalization of sampling without replacement from a finite universe}.
\bjournal{Journal of the American statistical Association}
\bvolume{47}
\bpages{663--685}.
\end{barticle}
\endbibitem

\bibitem[\protect\citeauthoryear{Hubbard et~al.}{2010}]{hubbard2010gee}
\begin{barticle}[author]
\bauthor{\bsnm{Hubbard},~\bfnm{Alan~E}\binits{A.~E.}}, \bauthor{\bsnm{Ahern},~\bfnm{Jennifer}\binits{J.}}, \bauthor{\bsnm{Fleischer},~\bfnm{Nancy~L}\binits{N.~L.}}, \bauthor{\bparticle{Van~der} \bsnm{Laan},~\bfnm{Mark}\binits{M.}}, \bauthor{\bsnm{Satariano},~\bfnm{Sheri~A}\binits{S.~A.}}, \bauthor{\bsnm{Jewell},~\bfnm{Nicholas}\binits{N.}}, \bauthor{\bsnm{Bruckner},~\bfnm{Tim}\binits{T.}} \AND \bauthor{\bsnm{Satariano},~\bfnm{William~A}\binits{W.~A.}}
(\byear{2010}).
\btitle{To GEE or not to GEE: comparing population average and mixed models for estimating the associations between neighborhood risk factors and health}.
\bjournal{Epidemiology}
\bpages{467--474}.
\end{barticle}
\endbibitem

\bibitem[\protect\citeauthoryear{Huggins}{1989}]{huggins1989statistical}
\begin{barticle}[author]
\bauthor{\bsnm{Huggins},~\bfnm{RM991431}\binits{R.}}
(\byear{1989}).
\btitle{On the statistical analysis of capture experiments}.
\bjournal{Biometrika}
\bvolume{76}
\bpages{133--140}.
\end{barticle}
\endbibitem

\bibitem[\protect\citeauthoryear{Is{\'e}ki}{1957}]{iseki1957theorem}
\begin{barticle}[author]
\bauthor{\bsnm{Is{\'e}ki},~\bfnm{Kiyoshi}\binits{K.}}
(\byear{1957}).
\btitle{A theorem on continuous convergence}.
\bjournal{Proceedings of the Japan Academy}
\bvolume{33}
\bpages{355--356}.
\end{barticle}
\endbibitem

\bibitem[\protect\citeauthoryear{Kang and Schafer}{2007}]{kang2007demystifying}
\begin{barticle}[author]
\bauthor{\bsnm{Kang},~\bfnm{Joseph~DY}\binits{J.~D.}} \AND \bauthor{\bsnm{Schafer},~\bfnm{Joseph~L}\binits{J.~L.}}
(\byear{2007}).
\btitle{Demystifying double robustness: A comparison of alternative strategies for estimating a population mean from incomplete data}.
\bjournal{Statistical science}
\bvolume{22}
\bpages{523--539}.
\end{barticle}
\endbibitem

\bibitem[\protect\citeauthoryear{Kochel, Wilson and Mastrofski}{2011}]{kochel2011effect}
\begin{barticle}[author]
\bauthor{\bsnm{Kochel},~\bfnm{Tammy~Rinehart}\binits{T.~R.}}, \bauthor{\bsnm{Wilson},~\bfnm{David~B}\binits{D.~B.}} \AND \bauthor{\bsnm{Mastrofski},~\bfnm{Stephen~D}\binits{S.~D.}}
(\byear{2011}).
\btitle{EFFECT OF SUSPECT RACE ON OFFICERS’ARREST DECISIONS}.
\bjournal{Criminology}
\bvolume{49}
\bpages{473--512}.
\end{barticle}
\endbibitem

\bibitem[\protect\citeauthoryear{Lantz and Wenger}{2019}]{lantz2019co}
\begin{barticle}[author]
\bauthor{\bsnm{Lantz},~\bfnm{Brendan}\binits{B.}} \AND \bauthor{\bsnm{Wenger},~\bfnm{Marin~R}\binits{M.~R.}}
(\byear{2019}).
\btitle{The co-offender as counterfactual: A quasi-experimental within-partnership approach to the examination of the relationship between race and arrest}.
\bjournal{Journal of experimental criminology}
\bpages{1--24}.
\end{barticle}
\endbibitem

\bibitem[\protect\citeauthoryear{Lee and Chao}{1994}]{lee1994estimating}
\begin{barticle}[author]
\bauthor{\bsnm{Lee},~\bfnm{Shen-Ming}\binits{S.-M.}} \AND \bauthor{\bsnm{Chao},~\bfnm{Anne}\binits{A.}}
(\byear{1994}).
\btitle{Estimating population size via sample coverage for closed capture-recapture models}.
\bjournal{Biometrics}
\bpages{88--97}.
\end{barticle}
\endbibitem

\bibitem[\protect\citeauthoryear{Liang and Zeger}{1986}]{liang1986longitudinal}
\begin{barticle}[author]
\bauthor{\bsnm{Liang},~\bfnm{Kung-Yee}\binits{K.-Y.}} \AND \bauthor{\bsnm{Zeger},~\bfnm{Scott~L}\binits{S.~L.}}
(\byear{1986}).
\btitle{Longitudinal data analysis using generalized linear models}.
\bjournal{Biometrika}
\bvolume{73}
\bpages{13--22}.
\end{barticle}
\endbibitem

\bibitem[\protect\citeauthoryear{Little and Rubin}{2019}]{little2019statistical}
\begin{bbook}[author]
\bauthor{\bsnm{Little},~\bfnm{Roderick~JA}\binits{R.~J.}} \AND \bauthor{\bsnm{Rubin},~\bfnm{Donald~B}\binits{D.~B.}}
(\byear{2019}).
\btitle{Statistical analysis with missing data}
\bvolume{793}.
\bpublisher{John Wiley \& Sons}.
\end{bbook}
\endbibitem

\bibitem[\protect\citeauthoryear{Loeffler, Hyatt and Ridgeway}{2019}]{loeffler2019measuring}
\begin{barticle}[author]
\bauthor{\bsnm{Loeffler},~\bfnm{Charles~E}\binits{C.~E.}}, \bauthor{\bsnm{Hyatt},~\bfnm{Jordan}\binits{J.}} \AND \bauthor{\bsnm{Ridgeway},~\bfnm{Greg}\binits{G.}}
(\byear{2019}).
\btitle{Measuring self-reported wrongful convictions among prisoners}.
\bjournal{Journal of Quantitative Criminology}
\bvolume{35}
\bpages{259--286}.
\end{barticle}
\endbibitem

\bibitem[\protect\citeauthoryear{Lohr}{2007}]{lohr2007}
\begin{barticle}[author]
\bauthor{\bsnm{Lohr},~\bfnm{Sharon~L.}\binits{S.~L.}}
(\byear{2007}).
\btitle{{Comment: Struggles with Survey Weighting and Regression Modeling}}.
\bjournal{Statistical Science}
\bvolume{22}
\bpages{175 -- 178}.
\bdoi{10.1214/088342307000000159}
\end{barticle}
\endbibitem

\bibitem[\protect\citeauthoryear{Lumley and Scott}{2017}]{lumley2017fitting}
\begin{barticle}[author]
\bauthor{\bsnm{Lumley},~\bfnm{Thomas}\binits{T.}} \AND \bauthor{\bsnm{Scott},~\bfnm{Alastair}\binits{A.}}
(\byear{2017}).
\btitle{Fitting regression models to survey data}.
\bjournal{Statistical Science}
\bpages{265--278}.
\end{barticle}
\endbibitem

\bibitem[\protect\citeauthoryear{Lytle}{2014}]{lytle2014effects}
\begin{barticle}[author]
\bauthor{\bsnm{Lytle},~\bfnm{Daniel~J}\binits{D.~J.}}
(\byear{2014}).
\btitle{The effects of suspect characteristics on arrest: A meta-analysis}.
\bjournal{Journal of Criminal Justice}
\bvolume{42}
\bpages{589--597}.
\end{barticle}
\endbibitem

\bibitem[\protect\citeauthoryear{Morgan and Truman}{2021}]{morgan2021criminal}
\begin{barticle}[author]
\bauthor{\bsnm{Morgan},~\bfnm{Rachel~E}\binits{R.~E.}} \AND \bauthor{\bsnm{Truman},~\bfnm{JL}\binits{J.}}
(\byear{2021}).
\btitle{Criminal victimization, 2020}.
\bjournal{Washington, DC: National Crime Victimization Survey, Bureau of Justice Statistics. Retrieved Jan}
\bvolume{4}
\bpages{2022}.
\end{barticle}
\endbibitem

\bibitem[\protect\citeauthoryear{Morgan et~al.}{2017}]{morgan2017race}
\begin{barticle}[author]
\bauthor{\bsnm{Morgan},~\bfnm{Rachel~E}\binits{R.~E.}}, \bauthor{\bparticle{of} \bsnm{Justice Statistics~(BJS)},~\bfnm{Bureau}\binits{B.}}, \bauthor{\bparticle{of} \bsnm{Justice},~\bfnm{US~Dept}\binits{U.~D.}}, \bauthor{\bparticle{of} \bsnm{Justice~Programs},~\bfnm{Office}\binits{O.}} \AND \bauthor{\bparticle{of} \bsnm{America},~\bfnm{United~States}\binits{U.~S.}}
(\byear{2017}).
\btitle{Race and hispanic origin of victims and offenders, 2012-15}.
\bjournal{Victims and Offenders}
\bvolume{2012}
\bpages{15}.
\end{barticle}
\endbibitem

\bibitem[\protect\citeauthoryear{Nagin}{2013}]{nagin2013deterrence}
\begin{barticle}[author]
\bauthor{\bsnm{Nagin},~\bfnm{Daniel~S}\binits{D.~S.}}
(\byear{2013}).
\btitle{Deterrence in the twenty-first century}.
\bjournal{Crime and justice}
\bvolume{42}
\bpages{199--263}.
\end{barticle}
\endbibitem

\bibitem[\protect\citeauthoryear{Newey and McFadden}{1994}]{newey1994large}
\begin{barticle}[author]
\bauthor{\bsnm{Newey},~\bfnm{Whitney~K}\binits{W.~K.}} \AND \bauthor{\bsnm{McFadden},~\bfnm{Daniel}\binits{D.}}
(\byear{1994}).
\btitle{Large sample estimation and hypothesis testing}.
\bjournal{Handbook of econometrics}
\bvolume{4}
\bpages{2111--2245}.
\end{barticle}
\endbibitem

\bibitem[\protect\citeauthoryear{Petersen}{1896}]{petersen1896yearly}
\begin{barticle}[author]
\bauthor{\bsnm{Petersen},~\bfnm{Carl Georg~Johannes}\binits{C.~G.~J.}}
(\byear{1896}).
\btitle{The yearly immigration of young plaice in the Limfjord from the German sea}.
\bjournal{Rept. Danish Biol. Sta.}
\bvolume{6}
\bpages{1--48}.
\end{barticle}
\endbibitem

\bibitem[\protect\citeauthoryear{Piquero and Brame}{2008}]{piquero2008assessing}
\begin{barticle}[author]
\bauthor{\bsnm{Piquero},~\bfnm{Alex~R}\binits{A.~R.}} \AND \bauthor{\bsnm{Brame},~\bfnm{Robert~W}\binits{R.~W.}}
(\byear{2008}).
\btitle{Assessing the race--crime and ethnicity--crime relationship in a sample of serious adolescent delinquents}.
\bjournal{Crime \& Delinquency}
\bvolume{54}
\bpages{390--422}.
\end{barticle}
\endbibitem

\bibitem[\protect\citeauthoryear{Polley and Van Der~Laan}{2010}]{polley2010super}
\begin{barticle}[author]
\bauthor{\bsnm{Polley},~\bfnm{Eric~C}\binits{E.~C.}} \AND \bauthor{\bsnm{Van Der~Laan},~\bfnm{Mark~J}\binits{M.~J.}}
(\byear{2010}).
\btitle{Super learner in prediction}.
\end{barticle}
\endbibitem

\bibitem[\protect\citeauthoryear{Pope and Snyder}{2003}]{pope2003race}
\begin{bbook}[author]
\bauthor{\bsnm{Pope},~\bfnm{Carl~E}\binits{C.~E.}} \AND \bauthor{\bsnm{Snyder},~\bfnm{Howard~N}\binits{H.~N.}}
(\byear{2003}).
\btitle{Race as a factor in juvenile arrests}.
\bpublisher{Citeseer}.
\end{bbook}
\endbibitem

\bibitem[\protect\citeauthoryear{Racine and Li}{2004}]{racine2004nonparametric}
\begin{barticle}[author]
\bauthor{\bsnm{Racine},~\bfnm{Jeff}\binits{J.}} \AND \bauthor{\bsnm{Li},~\bfnm{Qi}\binits{Q.}}
(\byear{2004}).
\btitle{Nonparametric estimation of regression functions with both categorical and continuous data}.
\bjournal{Journal of Econometrics}
\bvolume{119}
\bpages{99--130}.
\end{barticle}
\endbibitem

\bibitem[\protect\citeauthoryear{Rennison}{2010}]{rennison2010investigation}
\begin{barticle}[author]
\bauthor{\bsnm{Rennison},~\bfnm{Callie~Marie}\binits{C.~M.}}
(\byear{2010}).
\btitle{An investigation of reporting violence to the police: A focus on Hispanic victims}.
\bjournal{Journal of Criminal Justice}
\bvolume{38}
\bpages{390--399}.
\end{barticle}
\endbibitem

\bibitem[\protect\citeauthoryear{Richardson, Schultz and Crawford}{2019}]{richardson2019dirty}
\begin{barticle}[author]
\bauthor{\bsnm{Richardson},~\bfnm{Rashida}\binits{R.}}, \bauthor{\bsnm{Schultz},~\bfnm{Jason}\binits{J.}} \AND \bauthor{\bsnm{Crawford},~\bfnm{Kate}\binits{K.}}
(\byear{2019}).
\btitle{Dirty data, bad predictions: How civil rights violations impact police data, predictive policing systems, and justice}.
\bjournal{New York University Law Review Online, Forthcoming}.
\end{barticle}
\endbibitem

\bibitem[\protect\citeauthoryear{Roberts and Lyons}{2009}]{roberts2009victim}
\begin{barticle}[author]
\bauthor{\bsnm{Roberts},~\bfnm{Aki}\binits{A.}} \AND \bauthor{\bsnm{Lyons},~\bfnm{Christopher~J}\binits{C.~J.}}
(\byear{2009}).
\btitle{Victim-offender racial dyads and clearance of lethal and nonlethal assault}.
\bjournal{Journal of research in crime and delinquency}
\bvolume{46}
\bpages{301--326}.
\end{barticle}
\endbibitem

\bibitem[\protect\citeauthoryear{Roberts and Lyons}{2011}]{roberts2011hispanic}
\begin{barticle}[author]
\bauthor{\bsnm{Roberts},~\bfnm{Aki}\binits{A.}} \AND \bauthor{\bsnm{Lyons},~\bfnm{Christopher~J}\binits{C.~J.}}
(\byear{2011}).
\btitle{Hispanic victims and homicide clearance by arrest}.
\bjournal{Homicide Studies}
\bvolume{15}
\bpages{48--73}.
\end{barticle}
\endbibitem

\bibitem[\protect\citeauthoryear{Rubin-Bleuer and Kratina}{2005}]{rubin2005two}
\begin{barticle}[author]
\bauthor{\bsnm{Rubin-Bleuer},~\bfnm{Susana}\binits{S.}} \AND \bauthor{\bsnm{Kratina},~\bfnm{Ioana~Schiopu}\binits{I.~S.}}
(\byear{2005}).
\btitle{On the two-phase framework for joint model and design-based inference}.
\bjournal{The Annals of Statistics}
\bvolume{33}
\bpages{2789--2810}.
\end{barticle}
\endbibitem

\bibitem[\protect\citeauthoryear{S{\"a}rndal, Swensson and Wretman}{2003}]{sarndal2003model}
\begin{bbook}[author]
\bauthor{\bsnm{S{\"a}rndal},~\bfnm{Carl-Erik}\binits{C.-E.}}, \bauthor{\bsnm{Swensson},~\bfnm{Bengt}\binits{B.}} \AND \bauthor{\bsnm{Wretman},~\bfnm{Jan}\binits{J.}}
(\byear{2003}).
\btitle{Model assisted survey sampling}.
\bpublisher{Springer Science \& Business Media}.
\end{bbook}
\endbibitem

\bibitem[\protect\citeauthoryear{Skogan}{1974}]{skogan1974validity}
\begin{barticle}[author]
\bauthor{\bsnm{Skogan},~\bfnm{Wesley~G}\binits{W.~G.}}
(\byear{1974}).
\btitle{The validity of official crime statistics: An empirical investigation}.
\bjournal{Social Science Quarterly}
\bpages{25--38}.
\end{barticle}
\endbibitem

\bibitem[\protect\citeauthoryear{Skogan}{1977}]{skogan1977dimensions}
\begin{barticle}[author]
\bauthor{\bsnm{Skogan},~\bfnm{Wesley~G}\binits{W.~G.}}
(\byear{1977}).
\btitle{Dimensions of the dark figure of unreported crime}.
\bjournal{Crime \& Delinquency}
\bvolume{23}
\bpages{41--50}.
\end{barticle}
\endbibitem

\bibitem[\protect\citeauthoryear{Solon, Haider and Wooldridge}{2015}]{solon2015we}
\begin{barticle}[author]
\bauthor{\bsnm{Solon},~\bfnm{Gary}\binits{G.}}, \bauthor{\bsnm{Haider},~\bfnm{Steven~J}\binits{S.~J.}} \AND \bauthor{\bsnm{Wooldridge},~\bfnm{Jeffrey~M}\binits{J.~M.}}
(\byear{2015}).
\btitle{What are we weighting for?}
\bjournal{Journal of Human resources}
\bvolume{50}
\bpages{301--316}.
\end{barticle}
\endbibitem

\bibitem[\protect\citeauthoryear{Steffensmeier et~al.}{2011}]{steffensmeier2011reassessing}
\begin{barticle}[author]
\bauthor{\bsnm{Steffensmeier},~\bfnm{Darrell}\binits{D.}}, \bauthor{\bsnm{Feldmeyer},~\bfnm{Ben}\binits{B.}}, \bauthor{\bsnm{Harris},~\bfnm{Casey~T}\binits{C.~T.}} \AND \bauthor{\bsnm{Ulmer},~\bfnm{Jeffery~T}\binits{J.~T.}}
(\byear{2011}).
\btitle{Reassessing trends in black violent crime, 1980--2008: Sorting out the “Hispanic effect” in Uniform Crime Reports arrests, National Crime Victimization Survey offender estimates, and US prisoner counts}.
\bjournal{Criminology}
\bvolume{49}
\bpages{197--251}.
\end{barticle}
\endbibitem

\bibitem[\protect\citeauthoryear{Sugiyama, Krauledat and M{\"u}ller}{2007}]{sugiyama2007covariate}
\begin{barticle}[author]
\bauthor{\bsnm{Sugiyama},~\bfnm{Masashi}\binits{M.}}, \bauthor{\bsnm{Krauledat},~\bfnm{Matthias}\binits{M.}} \AND \bauthor{\bsnm{M{\"u}ller},~\bfnm{Klaus-Robert}\binits{K.-R.}}
(\byear{2007}).
\btitle{Covariate shift adaptation by importance weighted cross validation.}
\bjournal{Journal of Machine Learning Research}
\bvolume{8}.
\end{barticle}
\endbibitem

\bibitem[\protect\citeauthoryear{Thompson}{1997}]{thompson1997theory}
\begin{bbook}[author]
\bauthor{\bsnm{Thompson},~\bfnm{Mary}\binits{M.}}
(\byear{1997}).
\btitle{Theory of sample surveys}
\bvolume{74}.
\bpublisher{CRC Press}.
\end{bbook}
\endbibitem

\bibitem[\protect\citeauthoryear{Tibshirani}{1996}]{tibshirani1996regression}
\begin{barticle}[author]
\bauthor{\bsnm{Tibshirani},~\bfnm{Robert}\binits{R.}}
(\byear{1996}).
\btitle{Regression shrinkage and selection via the lasso}.
\bjournal{Journal of the Royal Statistical Society: Series B (Methodological)}
\bvolume{58}
\bpages{267--288}.
\end{barticle}
\endbibitem

\bibitem[\protect\citeauthoryear{United States Department~of Justice}{2008a}]{nibrs06}
\begin{barticle}[author]
\bauthor{\bparticle{United States Department~of} \bsnm{Justice},~\bfnm{Federal Bureau of~Investigation}\binits{F.~B. o.~I.}}
(\byear{2008}a).
\btitle{National Incident-Based Reporting System, 2006}.
\bdoi{10.3886/ICPSR22407.v1}
\end{barticle}
\endbibitem

\bibitem[\protect\citeauthoryear{United States Department~of Justice}{2008b}]{leoka06}
\begin{barticle}[author]
\bauthor{\bparticle{United States Department~of} \bsnm{Justice},~\bfnm{Federal Bureau of~Investigation}\binits{F.~B. o.~I.}}
(\byear{2008}b).
\btitle{Uniform Crime Reporting Program Data {$[$}United States{$]$}: Police Employee (LEOKA) Data, 2006}.
\bdoi{10.3886/ICPSR22402.v1}
\end{barticle}
\endbibitem

\bibitem[\protect\citeauthoryear{United States Department~of Justice}{2009a}]{nibrs07}
\begin{barticle}[author]
\bauthor{\bparticle{United States Department~of} \bsnm{Justice},~\bfnm{Federal Bureau of~Investigation}\binits{F.~B. o.~I.}}
(\byear{2009}a).
\btitle{National Incident-Based Reporting System, 2007}.
\bdoi{10.3886/ICPSR25113.v1}
\end{barticle}
\endbibitem

\bibitem[\protect\citeauthoryear{United States Department~of Justice}{2009b}]{leoka07}
\begin{barticle}[author]
\bauthor{\bparticle{United States Department~of} \bsnm{Justice},~\bfnm{Federal Bureau of~Investigation}\binits{F.~B. o.~I.}}
(\byear{2009}b).
\btitle{Uniform Crime Reporting Program Data {$[$}United States{$]$}: Police Employee (LEOKA) Data, 2007}.
\bdoi{10.3886/ICPSR25104.v1}
\end{barticle}
\endbibitem

\bibitem[\protect\citeauthoryear{United States Department~of Justice}{2010a}]{nibrs08}
\begin{barticle}[author]
\bauthor{\bparticle{United States Department~of} \bsnm{Justice},~\bfnm{Federal Bureau of~Investigation}\binits{F.~B. o.~I.}}
(\byear{2010}a).
\btitle{National Incident-Based Reporting System, 2008}.
\bdoi{10.3886/ICPSR27647.v1}
\end{barticle}
\endbibitem

\bibitem[\protect\citeauthoryear{United States Department~of Justice}{2010b}]{leoka08}
\begin{barticle}[author]
\bauthor{\bparticle{United States Department~of} \bsnm{Justice},~\bfnm{Federal Bureau of~Investigation}\binits{F.~B. o.~I.}}
(\byear{2010}b).
\btitle{Uniform Crime Reporting Program Data {$[$}United States{$]$}: Police Employee (LEOKA) Data, 2008}.
\bdoi{10.3886/ICPSR27646.v1}
\end{barticle}
\endbibitem

\bibitem[\protect\citeauthoryear{United States Department~of Justice}{2011a}]{nibrs09}
\begin{barticle}[author]
\bauthor{\bparticle{United States Department~of} \bsnm{Justice},~\bfnm{Federal Bureau of~Investigation}\binits{F.~B. o.~I.}}
(\byear{2011}a).
\btitle{Uniform Crime Reporting: National Incident-Based Reporting System, 2009}.
\bdoi{10.3886/ICPSR30770.v1}
\end{barticle}
\endbibitem

\bibitem[\protect\citeauthoryear{United States Department~of Justice}{2011b}]{leoka09}
\begin{barticle}[author]
\bauthor{\bparticle{United States Department~of} \bsnm{Justice},~\bfnm{Federal Bureau of~Investigation}\binits{F.~B. o.~I.}}
(\byear{2011}b).
\btitle{Uniform Crime Reporting Program Data {$[$}United States{$]$}: Police Employee (LEOKA) Data, 2009}.
\bdoi{10.3886/ICPSR30765.v1}
\end{barticle}
\endbibitem

\bibitem[\protect\citeauthoryear{United States Department~of Justice}{2012a}]{nibrs10}
\begin{barticle}[author]
\bauthor{\bparticle{United States Department~of} \bsnm{Justice},~\bfnm{Federal Bureau of~Investigation}\binits{F.~B. o.~I.}}
(\byear{2012}a).
\btitle{Uniform Crime Reporting: National Incident-Based Reporting System, 2010}.
\bdoi{10.3886/ICPSR33530.v1}
\end{barticle}
\endbibitem

\bibitem[\protect\citeauthoryear{United States Department~of Justice}{2012b}]{leoka10}
\begin{barticle}[author]
\bauthor{\bparticle{United States Department~of} \bsnm{Justice},~\bfnm{Federal Bureau of~Investigation}\binits{F.~B. o.~I.}}
(\byear{2012}b).
\btitle{Uniform Crime Reporting Program Data: Police Employee (LEOKA) Data, 2010}.
\bdoi{10.3886/ICPSR33525.v1}
\end{barticle}
\endbibitem

\bibitem[\protect\citeauthoryear{United States Department~of Justice}{2013a}]{nibrs11}
\begin{barticle}[author]
\bauthor{\bparticle{United States Department~of} \bsnm{Justice},~\bfnm{Federal Bureau of~Investigation}\binits{F.~B. o.~I.}}
(\byear{2013}a).
\btitle{Uniform Crime Reporting Program Data: National Incident-Based Reporting System, 2011}.
\bdoi{10.3886/ICPSR34585.v1}
\end{barticle}
\endbibitem

\bibitem[\protect\citeauthoryear{United States Department~of Justice}{2013b}]{leoka11}
\begin{barticle}[author]
\bauthor{\bparticle{United States Department~of} \bsnm{Justice},~\bfnm{Federal Bureau of~Investigation}\binits{F.~B. o.~I.}}
(\byear{2013}b).
\btitle{Uniform Crime Reporting Program Data: Police Employee (LEOKA) Data, 2011}.
\bdoi{10.3886/ICPSR34584.v1}
\end{barticle}
\endbibitem

\bibitem[\protect\citeauthoryear{United States Department~of Justice}{2014a}]{nibrs12}
\begin{barticle}[author]
\bauthor{\bparticle{United States Department~of} \bsnm{Justice},~\bfnm{Federal Bureau of~Investigation}\binits{F.~B. o.~I.}}
(\byear{2014}a).
\btitle{Uniform Crime Reporting Program Data: National Incident-Based Reporting System, 2012}.
\bdoi{10.3886/ICPSR35035.v1}
\end{barticle}
\endbibitem

\bibitem[\protect\citeauthoryear{United States Department~of Justice}{2014b}]{leoka12}
\begin{barticle}[author]
\bauthor{\bparticle{United States Department~of} \bsnm{Justice},~\bfnm{Federal Bureau of~Investigation}\binits{F.~B. o.~I.}}
(\byear{2014}b).
\btitle{Uniform Crime Reporting Program Data: Police Employee (LEOKA) Data, 2012}.
\bdoi{10.3886/ICPSR35020.v1}
\end{barticle}
\endbibitem

\bibitem[\protect\citeauthoryear{United States Department~of Justice}{2015a}]{nibrs13}
\begin{barticle}[author]
\bauthor{\bparticle{United States Department~of} \bsnm{Justice},~\bfnm{Federal Bureau of~Investigation}\binits{F.~B. o.~I.}}
(\byear{2015}a).
\btitle{Uniform Crime Reporting Program Data: National Incident-Based Reporting System, 2013}.
\bdoi{10.3886/ICPSR36120.v2}
\end{barticle}
\endbibitem

\bibitem[\protect\citeauthoryear{United States Department~of Justice}{2015b}]{leoka13}
\begin{barticle}[author]
\bauthor{\bparticle{United States Department~of} \bsnm{Justice},~\bfnm{Federal Bureau of~Investigation}\binits{F.~B. o.~I.}}
(\byear{2015}b).
\btitle{Uniform Crime Reporting Program Data: Police Employee (LEOKA) Data, 2013}.
\bdoi{10.3886/ICPSR36119.v1}
\end{barticle}
\endbibitem

\bibitem[\protect\citeauthoryear{United States Department~of Justice}{2016a}]{nibrs14}
\begin{barticle}[author]
\bauthor{\bparticle{United States Department~of} \bsnm{Justice},~\bfnm{Federal Bureau of~Investigation}\binits{F.~B. o.~I.}}
(\byear{2016}a).
\btitle{Uniform Crime Reporting Program Data: National Incident-Based Reporting System, 2014}.
\bdoi{10.3886/ICPSR36398.v1}
\end{barticle}
\endbibitem

\bibitem[\protect\citeauthoryear{United States Department~of Justice}{2016b}]{leoka14}
\begin{barticle}[author]
\bauthor{\bparticle{United States Department~of} \bsnm{Justice},~\bfnm{Federal Bureau of~Investigation}\binits{F.~B. o.~I.}}
(\byear{2016}b).
\btitle{Uniform Crime Reporting Program Data: Police Employee (LEOKA) Data, 2014}.
\bdoi{10.3886/ICPSR36395.v1}
\end{barticle}
\endbibitem

\bibitem[\protect\citeauthoryear{United States Department~of Justice}{2017a}]{ncvstech16}
\begin{bmanual}[author]
\bauthor{\bparticle{United States Department~of} \bsnm{Justice},~\bfnm{Bureau of Justice~Statistics}\binits{B.~o. J.~S.}}
(\byear{2017}a).
\btitle{National Crime Victimization Survey, 2016. Technical Documentation}.
\end{bmanual}
\endbibitem

\bibitem[\protect\citeauthoryear{United States Department~of Justice}{2017b}]{nibrs15}
\begin{barticle}[author]
\bauthor{\bparticle{United States Department~of} \bsnm{Justice},~\bfnm{Federal Bureau of~Investigation}\binits{F.~B. o.~I.}}
(\byear{2017}b).
\btitle{Uniform Crime Reporting Program Data: National Incident-Based Reporting System, 2015}.
\bdoi{10.3886/ICPSR36795.v1}
\end{barticle}
\endbibitem

\bibitem[\protect\citeauthoryear{United States Department~of Justice}{2017c}]{leoka15}
\begin{barticle}[author]
\bauthor{\bparticle{United States Department~of} \bsnm{Justice},~\bfnm{Federal Bureau of~Investigation}\binits{F.~B. o.~I.}}
(\byear{2017}c).
\btitle{Uniform Crime Reporting Program Data: Police Employee (LEOKA) Data, 2015}.
\bdoi{10.3886/ICPSR36791.v1}
\end{barticle}
\endbibitem

\bibitem[\protect\citeauthoryear{United States Department~of Justice}{2019}]{nibrs_manual19}
\begin{bmanual}[author]
\bauthor{\bparticle{United States Department~of} \bsnm{Justice},~\bfnm{Federal Bureau of Investigation~(FBI)}\binits{F.~B. o. I.~F.}}
(\byear{2019}).
\btitle{2019 National Incident-Based Reporting System User Manual}.
\end{bmanual}
\endbibitem

\bibitem[\protect\citeauthoryear{United States Department~of Justice}{2021}]{ncvs9220}
\begin{barticle}[author]
\bauthor{\bparticle{United States Department~of} \bsnm{Justice},~\bfnm{Bureau of Justice~Statistics}\binits{B.~o. J.~S.}}
(\byear{2021}).
\btitle{National Crime Victimization Survey, Concatenated File, {$[$}United States{$]$}, 1992-2020}.
\bdoi{10.3886/ICPSR38136.v1}
\end{barticle}
\endbibitem

\bibitem[\protect\citeauthoryear{Van Der~Heijden et~al.}{2003}]{van2003point}
\begin{barticle}[author]
\bauthor{\bsnm{Van Der~Heijden},~\bfnm{Peter~Gm}\binits{P.~G.}}, \bauthor{\bsnm{Bustami},~\bfnm{Rami}\binits{R.}}, \bauthor{\bsnm{Cruyff},~\bfnm{Maarten~JLF}\binits{M.~J.}}, \bauthor{\bsnm{Engbersen},~\bfnm{Godfried}\binits{G.}} \AND \bauthor{\bsnm{Van~Houwelingen},~\bfnm{Hans~C}\binits{H.~C.}}
(\byear{2003}).
\btitle{Point and interval estimation of the population size using the truncated Poisson regression model}.
\bjournal{Statistical Modelling}
\bvolume{3}
\bpages{305--322}.
\end{barticle}
\endbibitem

\bibitem[\protect\citeauthoryear{Van~der Laan, Polley and Hubbard}{2007}]{van2007super}
\begin{barticle}[author]
\bauthor{\bparticle{Van~der} \bsnm{Laan},~\bfnm{Mark~J}\binits{M.~J.}}, \bauthor{\bsnm{Polley},~\bfnm{Eric~C}\binits{E.~C.}} \AND \bauthor{\bsnm{Hubbard},~\bfnm{Alan~E}\binits{A.~E.}}
(\byear{2007}).
\btitle{Super learner}.
\bjournal{Statistical applications in genetics and molecular biology}
\bvolume{6}.
\end{barticle}
\endbibitem

\bibitem[\protect\citeauthoryear{Van~der Vaart}{2000}]{van2000asymptotic}
\begin{bbook}[author]
\bauthor{\bparticle{Van~der} \bsnm{Vaart},~\bfnm{Aad~W}\binits{A.~W.}}
(\byear{2000}).
\btitle{Asymptotic statistics}
\bvolume{3}.
\bpublisher{Cambridge university press}.
\end{bbook}
\endbibitem

\bibitem[\protect\citeauthoryear{White}{2014}]{white2014asymptotic}
\begin{bbook}[author]
\bauthor{\bsnm{White},~\bfnm{Halbert}\binits{H.}}
(\byear{2014}).
\btitle{Asymptotic theory for econometricians}.
\bpublisher{Academic press}.
\end{bbook}
\endbibitem

\bibitem[\protect\citeauthoryear{Xie and Baumer}{2019a}]{xie2019neighborhood}
\begin{barticle}[author]
\bauthor{\bsnm{Xie},~\bfnm{Min}\binits{M.}} \AND \bauthor{\bsnm{Baumer},~\bfnm{Eric~P}\binits{E.~P.}}
(\byear{2019}a).
\btitle{Neighborhood immigrant concentration and violent crime reporting to the police: A multilevel analysis of data from the National Crime Victimization Survey}.
\bjournal{Criminology}
\bvolume{57}
\bpages{237--267}.
\end{barticle}
\endbibitem

\bibitem[\protect\citeauthoryear{Xie and Baumer}{2019b}]{xie2019crime}
\begin{barticle}[author]
\bauthor{\bsnm{Xie},~\bfnm{Min}\binits{M.}} \AND \bauthor{\bsnm{Baumer},~\bfnm{Eric~P}\binits{E.~P.}}
(\byear{2019}b).
\btitle{Crime victims’ decisions to call the police: Past research and new directions}.
\bjournal{Annual Review of Criminology}.
\end{barticle}
\endbibitem

\bibitem[\protect\citeauthoryear{Xie and Lauritsen}{2012}]{xie2012racial}
\begin{barticle}[author]
\bauthor{\bsnm{Xie},~\bfnm{Min}\binits{M.}} \AND \bauthor{\bsnm{Lauritsen},~\bfnm{Janet~L}\binits{J.~L.}}
(\byear{2012}).
\btitle{Racial context and crime reporting: A test of Black’s stratification hypothesis}.
\bjournal{Journal of quantitative criminology}
\bvolume{28}
\bpages{265--293}.
\end{barticle}
\endbibitem

\bibitem[\protect\citeauthoryear{Xie and Lynch}{2017}]{xie2017effects}
\begin{barticle}[author]
\bauthor{\bsnm{Xie},~\bfnm{Min}\binits{M.}} \AND \bauthor{\bsnm{Lynch},~\bfnm{James~P}\binits{J.~P.}}
(\byear{2017}).
\btitle{The effects of arrest, reporting to the police, and victim services on intimate partner violence}.
\bjournal{Journal of research in crime and delinquency}
\bvolume{54}
\bpages{338--378}.
\end{barticle}
\endbibitem

\bibitem[\protect\citeauthoryear{Yuan and Jennrich}{1998}]{yuan1998asymptotics}
\begin{barticle}[author]
\bauthor{\bsnm{Yuan},~\bfnm{Ke-Hai}\binits{K.-H.}} \AND \bauthor{\bsnm{Jennrich},~\bfnm{Robert~I}\binits{R.~I.}}
(\byear{1998}).
\btitle{Asymptotics of estimating equations under natural conditions}.
\bjournal{Journal of Multivariate Analysis}
\bvolume{65}
\bpages{245--260}.
\end{barticle}
\endbibitem

\end{thebibliography}

\end{document}